\documentclass[12pt,a4]{article}

\usepackage{amsthm,amsmath,latexsym,multibox,amssymb,amsfonts}
\usepackage{graphicx,lscape,fancyhdr,array,stmaryrd}

\newcommand{\bep}{\begin{picture}}
\newcommand{\eep}{\end{picture}}

\newcommand{\circleBox}[5]{\multiput(#1,#2)(#5,0){#3}{\multiput(0,0)(0,#5){#4}{\circle{5}}}}
\newcommand{\circleBoldBox}[5]{\multiput(#1,#2)(#5,0){#3}{\multiput(0,0)(0,#5){#4}{\circle*{5}}}}
\newcommand{\Boxof}[6]{\multiput(#1,#2)(#5,0){#3}{\multiput(0,0)(0,#5){#4}{#6}}}

\newcommand{\TwoModulesGlued}{{\bep(200,170)%
\put(10,150){$\overbrace{\phantom{aaaaaaaaaa}}^{\module_1}$}%
\put(125,150){$\overbrace{\phantom{aaaaaaaa}}^{\module_2}$}%
\put(10,150){\axeXLabel{200}{$g$}}%
\put(10,150){\axeYLabelUpdown{150}{$q$}}%
\circleBox{10}{30}{4}{7}{20}\circleBoldBox{10}{50}{4}{1}{20}%
\Boxof{30}{50}{3}{6}{20}{\put(-4,-4){\vector(-1,-1){12}}}%
\circleBox{130}{90}{3}{4}{20}\circleBoldBox{130}{90}{3}{1}{20}%
\Boxof{150}{110}{2}{3}{20}{\put(-4,-4){\vector(-1,-1){12}}}%
\multiput(130,90)(0,20){4}{\put(-4,-4){\vector(-1,-1){52}}}%
\put(-15,50){$W^1\fm{p}$}\put(-20,70){$\xi^1\fm{p-1}$}\put(-15,150){$\xi^1\fm{0}$}%
\put(180,90){$W^2\fm{r}$}\put(180,110){$\xi^2\fm{r-1}$}\put(180,150){$\xi^2\fm{0}$}%
\eep}}

\newcommand{\GaugeParameters}{{\bep(160,150){%
\put(0,0){\RectT{5}{4}{\TextTopLeftIn{\ensuremath{s_N-1}}{\ensuremath{p_N}}}}%
\put(50,20){\RectT{1}{2}{\TextRight{$k_N$}}}%
\put(50,0){\BlockADotted{1}{2}}%
\put(0,40){\RectT{9}{5}{\TextTopLeftIn{\ensuremath{s_2-1}}{\ensuremath{p_2}}}}%
\put(90,60){\RectT{1}{3}{\TextRight{$k_2$}}}%
\put(90,40){\BlockADotted{1}{2}}%
\put(0,90){\RectT{14}{6}{\TextTopLeftIn{\ensuremath{s_1-1}}{\ensuremath{p_1}}}}%
\put(140,120){\RectT{1}{3}{\TextRight{$k_1$}}}%
\put(140,90){\BlockADotted{1}{3}}%
}\eep}}

\newcommand{\PhysicalVielbein}{{\bep(160,170)(0,-20){%
\put(0,-20){\RectT{1}{2}{\TextRight{$\alpha_{N+1}$}}}%
\put(0,0){\RectT{5}{4}{\TextTopLeftIn{\ensuremath{s_N-1}}{\ensuremath{p_N}}}}%
\put(50,20){\RectT{1}{2}{\TextRight{$\alpha_N$}}}%
\put(0,40){\RectT{9}{5}{\TextTopLeftIn{\ensuremath{s_2-1}}{\ensuremath{p_2}}}}%
\put(90,60){\RectT{1}{3}{\TextRight{$\alpha_2$}}}%
\put(0,90){\RectT{14}{6}{\TextTopLeftIn{\ensuremath{s_1-1}}{\ensuremath{p_1}}}}%
\put(140,120){\RectT{1}{3}{\TextRight{$\alpha_1$}}}%
}\eep}}

\newcommand{\ABProduct}{{\bep(400,170){%
\put(0,0){\TypicalDiagram}%
\put(140,80){$\bigotimes$}%
\put(160,30){\RectT{1}{10}{\TextCenter{$q$}}}
\put(180,80){$=\displaystyle\sum_{\{\alpha_j\},\{\beta_i\}}N_{\{\alpha_j\},\{\beta_i\}}$}%
\put(280,0){\TracedTensorProduct}%
}\eep}}

\newcommand{\TypicalDiagram}{{\bep(160,170){%
\put(0,20){\RectT{5}{4}{\TextTopLeftIn{\ensuremath{s_N}}{\ensuremath{p_N}}}}%
\put(0,60){\RectT{9}{5}{\TextTopLeftIn{\ensuremath{s_2}}{\ensuremath{p_2}}}}%
\put(0,110){\RectT{14}{6}{\TextTopLeftIn{\ensuremath{s_1}}{\ensuremath{p_1}}}}%
}\eep}}

\newcommand{\TracedTensorProduct}{{\bep(160,170){%
\put(0,0){\RectT{1}{2}{\TextRight{$\alpha_{N+1}$}}}
\put(0,20){\ABRectT{5}{4}{1}{1}{N}}%
\put(0,60){\ABRectT{9}{5}{2}{2}{2}}%
\put(0,110){\ABRectT{14}{6}{2}{2}{1}}%
}\eep}}

\newcommand{\FieldContentDynamical}{{\unitlength=0.34mm\bep(160,100){%
\put(0,0){\RectT{5}{3}{\TextTopLeftIn{\ensuremath{s_N-1}}{\ensuremath{p_N}}}}%
\put(0,30){\RectT{10}{3}{\TextTopLeftIn{\ensuremath{s_2-1}}{\ensuremath{p_2}}}}%
\put(0,60){\RectT{15}{4}{\TextTopLeftIn{\ensuremath{s_1-1}}{\ensuremath{p_1}}}}%
}\eep}}

\newcommand{\FieldContentN}{{\unitlength=0.34mm\bep(160,120)(0,-10){%
\put(0,-10){\RectT{3}{1}{\TextCenter{$k$}}}%
\put(0,0){\RectT{5}{3}{\TextTopLeftIn{\ensuremath{s_N-1}}{\ensuremath{p_N}}}}%
\put(0,30){\RectT{10}{3}{\TextTopLeftIn{\ensuremath{s_2-1}}{\ensuremath{p_2}}}}%
\put(0,60){\RectT{15}{4}{\TextTopLeftIn{\ensuremath{s_1-1}}{\ensuremath{p_1}}}}%
}\eep}}

\newcommand{\FieldContentkk}{{\unitlength=0.34mm\bep(160,120)(0,-10){%
\put(0,-10){\RectT{5}{1}{\TextCenter{$\scriptstyle s_N-1$}}}%
\put(50,-10){\RectT{1}{4}{\TextCenter{$ $}}}%
\put(0,0){\RectT{5}{3}{\TextTopLeftIn{\ensuremath{s_N-1}}{\ensuremath{p_N}}}}%
\put(60,20){\RectT{3}{1}{\TextCenter{$\scriptstyle k$}}}%
\put(0,30){\RectT{10}{3}{\TextTopLeftIn{\ensuremath{s_2-1}}{\ensuremath{p_2}}}}%
\put(0,60){\RectT{15}{4}{\TextTopLeftIn{\ensuremath{s_1-1}}{\ensuremath{p_1}}}}%
}\eep}}

\newcommand{\FieldContentWeyl}{{\unitlength=0.34mm\bep(150,110)(0,-10){%
\put(0,-10){\RectT{5}{1}{\TextCenter{$s_N-1$}}}%
\put(50,-10){\RectT{1}{4}{\TextCenter{$ $}}}%
\put(0,0){\RectT{5}{3}{\TextTopLeftIn{\ensuremath{s_N-1}}{\ensuremath{p_N}}}}%
\put(60,20){\RectT{4}{1}{\TextCenter{$\scriptstyle s_2-s_3-1$}}}%
\put(100,20){\RectT{1}{4}{\TextCenter{$ $}}}%
\put(0,30){\RectT{10}{3}{\TextTopLeftIn{\ensuremath{s_2-1}}{\ensuremath{p_2}}}}%
\put(110,50){\RectT{5}{1}{\TextCenter{$\scriptstyle s_1-s_2-1$}}}%
\put(160,50){\RectT{1}{5}{\TextCenter{$ $}}}%
\put(0,60){\RectT{16}{4}{\TextTopLeftIn{\ensuremath{s_1-1}}{\ensuremath{p_1}}}}%
\put(170,90){\RectT{4}{1}{\TextCenter{$k+1$}}}%
}\eep}}

\newcommand{\FieldContentWeylTensor}{{\unitlength=0.34mm\bep(220,110)(0,-10){%
\put(0,-10){\RectT{6}{3}{\TextTopLeftIn{\ensuremath{s_N}}{\ensuremath{p_N}}}}%
\put(0,20){\RectT{11}{3}{\TextTopLeftIn{\ensuremath{s_2}}{\ensuremath{p_2}}}}%
\put(0,50){\RectT{17}{5}{\TextTopLeftIn{\ensuremath{s_1}}{\ensuremath{p_1+1}}}}%
}\eep}}

\newcommand{\axeYLabelUpdown}[2]{\Add{TGR0}{#1}{-10}\put(-10,-\value{TGR0}){#2}\put(0,0){\vector(0,-1){#1}}}

\newcommand{\axeXLabel}[2]{\Add{TGR0}{#1}{-10}\put(\value{TGR0},5){#2}\put(0,0){\vector(1,0){#1}}}

\newcounter{YoungHeight}\newcounter{YoungWidth}

\newcounter{Mul1}\newcounter{Mul2}\newcounter{Mul3}\newcounter{Mul4}
\newcounter{A0}\newcounter{A1}\newcounter{A2}\newcounter{A3}\newcounter{A4}\newcounter{A5}\newcounter{A6}
\newcounter{B0}\newcounter{B1}\newcounter{B2}\newcounter{B3}
\newcounter{C1}\newcounter{C2}\newcounter{C3}\newcounter{C4}\newcounter{C6}\newcounter{C7}
\newcounter{D1}\newcounter{D2}\newcounter{D3}\newcounter{D4}\newcounter{D5}
\newcounter{T0}\newcounter{T1}

\newcounter{TGR0}
\newcounter{R0}\newcounter{R1}\newcounter{R2}\newcounter{R3}
\newcounter{AR0}\newcounter{AR1}\newcounter{AR2}\newcounter{AR3}\newcounter{AR5}
\newcounter{Dotted0}\newcounter{Dotted1}\newcounter{Dotted2}\newcounter{Dotted3}

\newcounter{reptA}

\newlength{\txtHShift}

\newlength{\txtWidth}

\newcommand{\HalfLength}[2]{\setcounter{Mul1}{#1}\setcounter{Mul2}{#1}\addtocounter{Mul1}{\value{Mul2}}\addtocounter{Mul1}{\value{Mul2}}%
\addtocounter{Mul1}{\value{Mul2}}\addtocounter{Mul1}{\value{Mul2}}\setcounter{#2}{\value{Mul1}}}

\newcommand{\Add}[3]{\setcounter{#1}{#2}\addtocounter{#1}{#3}}

\newcommand{\Length}[1]{#10}
\newcommand{\YoungScale}{}

\newcommand{\reptInit}[1]{{\setcounter{reptA}{#1}}}

\newcommand{\reptbyone}[1]{{\ensuremath{{#1}_{\arabic{reptA}}}\addtocounter{reptA}{1}}}

\newcommand{\shiftedText}[2]{{\hspace{#1}#2}}
\newcommand{\calcHShift}[1]{\settowidth{\txtWidth}{#1}\setlength{\txtHShift}{-0.5\txtWidth}}

\newcommand{\TextCenter}[3]{{\HalfLength{#2}{T0}%
\HalfLength{#3}{T1}\addtocounter{T1}{-3}\calcHShift{#1}%
\put(\value{T0},\value{T1}){\shiftedText{\txtHShift}{#1}}}}

\newcommand{\TextTop}[3]{{\calcHShift{#1}\HalfLength{#2}{T0}\Add{T1}{\Length{#3}}{-9}\put(\value{T0},\value{T1}){\shiftedText{\txtHShift}{#1}}}}

\newcommand{\TextLeftIn}[3]{{\HalfLength{#3}{T1}\addtocounter{T1}{-5}\put(2,\value{T1}){#1}}}

\newcommand{\TextRight}[3]{{\HalfLength{#3}{T1}\addtocounter{T1}{-5}\Add{T0}{\Length{#2}}{2}\put(\value{T0},\value{T1}){#1}}}

\newcommand{\TextTopLeftIn}[4]{{\TextTop{#1}{#3}{#4}\TextLeftIn{#2}{#3}{#4}}}

\newcommand{\CrossMarker}[2]{{\put(\value{#1},\value{#2}){\circle{4}}}}

\newcommand{\Contraction}[5]{{\Add{C1}{\Length{#1}}{5}\Add{C2}{\Length{#2}}{5}\Add{C3}{\Length{#3}}{5}\Add{C4}{\Length{#4}}{5}%
\Add{C6}{#1}{#3}\HalfLength{\value{C6}}{C6}\Add{C7}{#2}{#4}\HalfLength{\value{C7}}{C7}\Add{C7}{\value{C7}}{#5}%
\qbezier(\value{C1},\value{C2})(\value{C6},\value{C7})(\value{C3},\value{C4})\CrossMarker{C1}{C2}\CrossMarker{C3}{C4}}}

\newcommand{\dottedline}[2]{{\HalfLength{#1}{Dotted0}\HalfLength{#2}{Dotted1}\Add{Dotted2}{#1}{#2}\HalfLength{\value{Dotted2}}{Dotted3}%
\qbezier[\value{Dotted3}](0,0)(\value{Dotted0},\value{Dotted1})(\Length{#1},\Length{#2})}}

\newcommand{\BlockA}[2]{{\YoungScale\bep(\Length{#1},\Length{#2}){\Add{A1}{#1}{1}\Add{A2}{#2}{1}}%
\multiput(0,0)(10,0){\value{A1}}{\line(0,1){\Length{#2}}}\multiput(0,0)(0,10){\value{A2}}{\line(1,0){\Length{#1}}}%
\setcounter{YoungHeight}{\Length{#2}}\setcounter{YoungWidth}{\Length{#1}}\eep}}

\newcommand{\BlockB}[4]{{\YoungScale\Add{B3}{\Length{#2}}{\Length{#4}}%
\bep(\Length{#1},\value{B3})\put(0,\Length{#4}){\BlockA{#1}{#2}}%
\put(0,0){\BlockA{#3}{#4}}\setcounter{YoungHeight}{\value{B3}}\setcounter{YoungWidth}{\Length{#1}}\eep}}

\newcommand{\BlockC}[6]{{\YoungScale\Add{C3}{\Length{#4}}{\Length{#6}}\Add{C4}{\value{C3}}{\Length{#2}}%
\bep(\Length{#1},\value{C4})\put(0,\value{C3}){\BlockA{#1}{#2}}\put(0,\Length{#6}){\BlockA{#3}{#4}}%
\put(0,0){\BlockA{#5}{#6}}\setcounter{YoungHeight}{\value{C4}}\setcounter{YoungWidth}{\Length{#1}}\eep}}

\newcommand{\BlockADotted}[2]{{\YoungScale\bep(\Length{#1},\Length{#2}){\Add{A1}{#1}{1}\Add{A2}{#2}{1}}%
\multiput(0,0)(10,0){\value{A1}}{\dottedline{0}{#2}}\multiput(0,0)(0,10){\value{A2}}{\dottedline{#1}{0}}%
\setcounter{YoungHeight}{\Length{#2}}\setcounter{YoungWidth}{\Length{#1}}\eep}}

\newcommand{\RectT}[3]{\bep(\Length{#1},\Length{#2})\put(0,0){\line(1,0){\Length{#1}}}\put(0,0){\line(0,1){\Length{#2}}}%
\put(\Length{#1},\Length{#2}){\line(-1,0){\Length{#1}}}\put(\Length{#1},\Length{#2}){\line(0,-1){\Length{#2}}}#3{#1}{#2}\eep}

\newcommand{\RectARow}[2]{{\bep(\Length{#1},10)\put(0,0){\RectT{#1}{1}{\TextTop{#2}}}\eep}}

\newcommand{\RectBRow}[4]{{\bep(\Length{#1},20)\put(0,0){\RectT{#2}{1}{\TextTop{#4}}}%
\put(0,10){\RectT{#1}{1}{\TextTop{#3}}}\eep}}

\newcommand{\RectAYoung}[3]{{\bep(0,0)#3\eep\Add{A0}{\value{YoungWidth}}{\Length{#1}}%
\Add{A1}{\value{YoungHeight}}{-10}\bep(\value{A0},\value{YoungHeight})%
\put(\value{YoungWidth},\value{A1}){\RectT{#1}{1}{\TextTop{#2}}}\eep}}

\newcommand{\RectBYoung}[3]{{\bep(0,0)\put(0,0){#3}\eep\Add{A0}{\value{YoungHeight}}{10}%
\bep(\Length{#1},\value{A0})%
\put(0,\value{YoungHeight}){\RectT{#1}{1}{\TextTop{#2}}}\eep}}

\newcommand{\ABRectT}[5]{{\Add{AR0}{\Length{#1}}{10}\Add{AR1}{\Length{#2}}{-\Length{#3}}\bep(\value{AR0},\Length{#2}){%
\Add{AR2}{\Length{#1}}{-10}\put(0,0){\line(1,0){\value{AR2}}}\put(\value{AR2},0){\line(0,1){\Length{#4}}}%
\put(\value{AR2},\Length{#4}){\line(1,0){10}}\put(\value{AR2},0){\dottedline{1}{0}}%
\put(\Length{#1},0){\dottedline{0}{#4}}%
\Add{AR3}{\Length{#2}}{-\Length{#3}}\addtocounter{AR3}{-\Length{#4}}%
\put(0,0){\line(0,1){\Length{#2}}}\put(0,\Length{#2}){\line(1,0){\value{AR0}}}%
\put(\Length{#1},\Length{#4}){\line(0,1){\value{AR3}}}\put(\Length{#1},\value{AR1}){\line(1,0){10}}%
\put(\value{AR0},\value{AR1}){\line(0,1){\Length{#3}}}\put(\Length{#1},\value{AR1}){\dottedline{0}{#3}}%
\put(\value{AR2},0){\TextCenter{$\ \beta_{#5}$}{1}{#4}}%
\put(\Length{#1},\value{AR1}){\TextCenter{$\alpha_{#5}$}{1}{#3}}%
\Add{AR5}{#2}{-#3}\addtocounter{AR5}{-#4}\put(\Length{#1},\Length{#4}){\TextLeftIn{$\epsilon_{#5}$}{2}{\value{AR5}}}%
\put(0,0){\TextLeftIn{$p_{#5}$}{#1}{#2}}%
\put(0,0){\TextCenter{$s_{#5}$}{#1}{#2}}%
}\eep}}

\newcommand{\BlockAEnum}[4]{{\YoungScale\Add{A1}{#1}{1}\Add{A2}{#2}{1}\Add{A3}{\Length{#1}}{5}\Add{A4}{\Length{#2}}{-10}\bep(\value{A3},\Length{#2}){%
\multiput(0,0)(10,0){\value{A1}}{\line(0,1){\Length{#2}}}\multiput(0,0)(0,10){\value{A2}}{\line(1,0){\Length{#1}}}%
\reptInit{#3}\multiput(\value{A3},\value{A4})(0,-10){#2}{\reptbyone{#4}}
}\eep}}

\newcommand{\BlockBEnum}[5]{{\YoungScale\Add{B2}{\Length{#1}}{10}\Add{B1}{#2}{1}\Add{B3}{\Length{#2}}{\Length{#4}}\bep(\value{B2},\value{B3})\put(0,\Length{#4}){\BlockAEnum{#1}{#2}{1}{#5}}%
\put(0,0){\BlockAEnum{#3}{#4}{\value{B1}}{#5}}\eep}}

\newcommand{\BlockCEnum}[7]{{\YoungScale\Add{C2}{\Length{#1}}{10}\Add{C1}{#2}{#4}\addtocounter{C1}{1}%
\Add{C3}{\Length{#4}}{\Length{#6}}\Add{C4}{\value{C3}}{\Length{#2}}\bep(\value{C2},\value{C4})%
\put(0,\Length{#6}){\BlockBEnum{#1}{#2}{#3}{#4}{#7}}\put(0,0){\BlockAEnum{#5}{#6}{\value{C1}}{#7}}\eep}}

\newcommand{\BlockDEnum}[9]{{\YoungScale\Add{D5}{\Length{#1}}{10}\Add{D4}{#2}{#4}\addtocounter{D4}{#6}\addtocounter{D4}{1}%
\Add{D1}{\Length{#8}}{\Length{#6}}\Add{D2}{\value{D1}}{\Length{#4}}\Add{D3}{\value{D2}}{\Length{#2}}%
\bep(\value{D5},\value{D3})\put(0,0){\BlockAEnum{#7}{#8}{\value{D4}}{#9}}\put(0,\Length{#8}){\BlockCEnum{#1}{#2}{#3}{#4}{#5}{#6}{#9}}\eep}}

\newcommand{\YoungBAzA}{{\bep(60,20){\put(0,0){\BlockB{2}{1}{1}{1}}\Contraction{1}{1}{5}{1}{20}\put(22,2){\ensuremath{\otimes_{\scriptstyle{sl(n)}}}} \put(50,0){\YoungAA{}}}\eep}}
\newcommand{\YoungBAzB}{{\bep(60,20){\put(0,0){\BlockB{2}{1}{1}{1}}\Contraction{0}{0}{5}{1}{10}\put(22,2){\ensuremath{\otimes_{\scriptstyle{sl(n)}}}}\put(50,0){\YoungAA{}}}\eep}}
\newcommand{\YoungBAzAB}{{\bep(60,20){\put(0,0){\BlockB{2}{1}{1}{1}}\Contraction{0}{0}{5}{0}{20}\Contraction{1}{1}{5}{1}{20}\put(22,2){\ensuremath{\otimes_{\scriptstyle{sl(n)}}}}\put(50,0){\YoungAA{}}}\eep}}

\newcommand{\YoungA}{\BlockA{1}{1}}
\newcommand{\YoungB}{\BlockA{2}{1}}
\newcommand{\YoungC}{\BlockA{3}{1}}

\newcommand{\YoungAA}{\BlockA{1}{2}}
\newcommand{\YoungBA}{\BlockB{2}{1}{1}{1}}
\newcommand{\YoungBB}{\BlockA{2}{2}}
\newcommand{\YoungCA}{\BlockB{3}{1}{1}{1}}
\newcommand{\YoungCB}{\BlockB{3}{1}{2}{1}}

\newcommand{\YoungAAA}{\BlockA{1}{3}}
\newcommand{\YoungBAA}{\BlockB{2}{1}{1}{2}}
\newcommand{\YoungBBA}{\BlockB{2}{2}{1}{1}}
\newcommand{\YoungBBB}{\BlockA{2}{3}}
\newcommand{\YoungCAA}{\BlockB{3}{1}{1}{2}}
\newcommand{\YoungCBA}{\BlockC{3}{1}{2}{1}{1}{1}}

\newcommand{\YoungCCB}{\BlockB{3}{2}{2}{1}}

\newcommand{\YoungDCB}{\BlockC{4}{1}{3}{1}{2}{1}}

\newcommand{\YoungAAAA}{\BlockA{1}{4}}
\newcommand{\YoungBAAA}{\BlockB{2}{1}{1}{3}}

\usepackage{amsthm,amsmath,latexsym,multibox,amssymb,amsfonts}
\usepackage{graphicx,lscape,fancyhdr,array,stmaryrd,euscript,wrapfig}

\renewcommand{\theequation}{\arabic{section}.\arabic{equation}}

\newcommand{\dSAdS}{{\ensuremath{(A)dS_d}}}

\newcommand{\AlgebraFont}[1]{\mathfrak{#1}}

\newcommand{\iso}{{\ensuremath{\AlgebraFont{iso}(d-1,1)}}}
\newcommand{\lorentz}{{\ensuremath{\AlgebraFont{so}(d-1,1)}}}
\newcommand{\msv}{{\ensuremath{\AlgebraFont{so}(d-1)}}}
\newcommand{\mls}{{\ensuremath{\AlgebraFont{so}(d-2)}}}

\newcommand{\sons}{{\ensuremath{\scriptstyle{so(n)}}}}
\newcommand{\slns}{{\ensuremath{\scriptstyle{sl(n)}}}}

\newcommand{\sod}{{\ensuremath{so(d)}}}

\newcommand{\sln}{{\ensuremath{sl(n)}}}
\newcommand{\sld}{{\ensuremath{sl(d)}}}

\newcommand{\ads}{{\ensuremath{\AlgebraFont{so}(d-1,2)}}}
\newcommand{\ds}{{\ensuremath{\AlgebraFont{so}(d,1)}}}

\newcommand{\Verma}[2]{\ensuremath{\EuScript{D}\left({\textstyle{#1}};#2\right)}}

\newcommand{\Irrep}[2]{\ensuremath{\EuScript{H}\left({\textstyle{#1}};#2\right)}}

\newcommand{\pl}{\partial}

\newcommand{\be}{\begin{equation}}
\newcommand{\ee}{\end{equation}}
\newcommand{\bes}{\begin{split}}
\newcommand{\es}{\end{split}}
\newcommand{\bee}{\begin{eqnarray}}
\newcommand{\eee}{\end{eqnarray}}
\newcommand{\beee}{\begin{array}}
\newcommand{\bem}{\begin{multline}}
\newcommand{\eem}{\end{multline}}

\theoremstyle{definition}
\newtheorem{Definition}{Definition}[section]

\theoremstyle{plain}

\newtheorem*{Lemma}{Lemma}
\newtheorem*{Corollary}{Corollary}

\theoremstyle{remark}
\newtheorem*{Comment}{Comment}
\newtheorem*{Sketch}{Sketch of the proof}

\newtheorem{Example}{Example}[subsection]

\newcommand{\bec}{\begin{Comment}}
\newcommand{\ec}{\end{Comment}}

\newcommand{\Y}[1]{{\ensuremath{\mathbf{Y}\{#1\}}}}
\newcommand{\Yy}{\ensuremath{\mathbf{Y}}}

\newcommand{\tr}[3]{{#1}^{#3}_{\phantom{#3}#3#2}}

\newcommand{\ComplexB}[2]{\ensuremath{0 \longrightarrow #1\longrightarrow #2 \longrightarrow 0}}
\newcommand{\ComplexC}[3]{\ensuremath{0 \longrightarrow #1\longrightarrow #2 \longrightarrow #3 \longrightarrow 0}}
\newcommand{\ComplexD}[4]{\ensuremath{0 \longrightarrow #1\longrightarrow #2 \longrightarrow #3 \longrightarrow #4  \longrightarrow 0}}

\newcommand{\fm}[1]{_{\mathbf{{#1}}}}

\newcommand{\aA}{{\ensuremath{\mathcal{A}}}}
\newcommand{\aB}{{\ensuremath{\mathcal{B}}}}
\newcommand{\aC}{{\ensuremath{\mathcal{C}}}}

\newcommand{\coh}{{\mathsf{H}}}

\newcommand{\irrep}{\textit{irrep}}
\newcommand{\irreps}{\textit{irreps}}

\newcommand{\uirrep}{\textit{uirrep}}

\newcommand{\onmassshell}{{\it on-mass-shell}}

\newcommand{\onshell}{{\it on-shell}}
\newcommand{\offshell}{{\it off-shell}}

\newcommand{\particle}{{\it particle}}
\newcommand{\particles}{{\it particles}}
\newcommand{\fields}{{\it fields}}
\newcommand{\field}{{\it field}}
\newcommand{\dual}{{\it dual}}
\newcommand{\Dual}{{\it Dual}}

\newcommand{\auxiliary}{{\it auxiliary}}
\newcommand{\Stueckelberg}{{\it Stueckelberg}}
\newcommand{\dynamical}{{\it dynamical}}
\newcommand{\physical}{{\it physical}}
\newcommand{\redundant}{{\it redundant}}
\newcommand{\metric}{{\it metric-like}}
\newcommand{\minimalformulation}{{\it minimal formulation}}
\newcommand{\interpretation}{{\it interpretation}}
\newcommand{\interpretations}{{\it interpretations}}
\newcommand{\framelike}{{\it frame-like}}
\newcommand{\Framelike}{{\it Frame-like}}
\newcommand{\lowestgrade}{{\it lowest grade}}

\newcommand{\partition}[1]{{\ensuremath{\mathtt{P}(#1)}}}
\newcommand{\dimension}{{\ensuremath{\mathrm{dim}}}}

\newcommand{\WW}{{\ensuremath{\mathcal{W}}}}

\newcommand{\smallpic}[1]{{\unitlength=0.2mm#1}}
\newcommand{\module}{{\ensuremath{\mathcal{R}}}}
\newcommand{\moduleP}{{\ensuremath{\mathcal{P}}}}
\newcommand{\Comp}{{\ensuremath{\mathcal{C}}}}

\newcommand{\kmax}[1]{{k^{{\scriptstyle max}}_{#1}}}

\newcommand{\YProjector}[1]{{{\mathtt\Pi}\left[#1\right]}}

\newcommand{\DD}{{{\mathcal{D}}}}
\newcommand{\ferm}{{\frac12}}
\newcommand{\fpl}{{\slash\!\!\!\partial}}
\newcommand{\fsp}{{;}}

\begin{document}
{\begin{titlepage}
\begin{flushright}
\vspace{1mm}
FIAN/TD/14-08\\
\end{flushright}

\vspace{1cm}

\begin{center}
{\bf \Large  Mixed-Symmetry Massless Fields in Minkowski space
Unfolded} \vspace{1cm}

\textsc{E.D.
Skvortsov\footnote{skvortsov@lpi.ru}}

\vspace{.7cm}

{ I.E.Tamm Department of Theoretical Physics, P.N.Lebedev Physical
Institute,\\Leninsky prospect 53, 119991, Moscow, Russia}

\end{center}

\vspace{0.5cm}
\begin{abstract}
The unfolded formulation for arbitrary massless mixed-symmetry bosonic and fermionic fields in
Minkowski space is constructed. The unfolded form is proved to be uniquely determined by the requirement that all gauge symmetries are manifest. The unfolded equations have the form of a covariant constancy condition. The gauge fields and gauge parameters are differential forms with values in certain irreducible Lorentz tensors. The unfolded equations for bosons determine completely those for fermions. The proposed unfolded
formulation also contains dual formulations for massless mixed-symmetry fields.
\end{abstract}

\end{titlepage}

\tableofcontents

\section*{Introduction}

In four dimensional Minkowski space spin degrees of freedom are
known to be classified by non-negative integers or half-integers.
However, in dimensions higher than four spin degrees of freedom are
described by a set of (half)integers, according to the weights of
the Wigner's little group. The simplest and the most developed are
the cases of totally symmetric \cite{Fronsdal:1978rb,
Fradkin:1986qy, Fradkin:1987ks, Vasiliev:1990en, Francia:2002pt,
Francia:2002aa, Vasiliev:2003ev, Bekaert:2003uc, Sorokin:2004ie,
Bekaert:2005vh, Sagnotti:2005ns} and totally antisymmetric fields
\cite{Townsend:1979hd, Sezgin:1980tp, Blau:1989bq}. All other
types are referred to collectively as mixed-symmetry.
Mixed-symmetry fields naturally arise in field theories in
higher-dimensions, for instance, in (super)string theory
\cite{Bonelli:2003kh}.

The simplest mixed-symmetry fields were originally considered in
\cite{Curtright:1980yk} and \cite{Aulakh:1986cb}. The most general
type of mixed-symmetry fields was studied in
\cite{Ouvry:1986dv,Labastida:1986gy,Labastida:1986ft,Labastida:1987kw}  though,
the rigorous proof of the fact that the proposed in
\cite{Labastida:1986ft} fields/gauge symmetries content and
equations describe massless \particles\ properly was given in
\cite{Bekaert:2002dt, Bekaert:2006ix}. In terms of BRST approach
mixed-symmetry fields, characterized by at most two non-zero
weights, were studied in \cite{Burdik:2001hj, Burdik:2000kj,
Moshin:2007jt, Buchbinder:2007ix}. An elegant approach to
the description of mixed-symmetry fields was proposed in
\cite{Zinoviev:2002ye, Zinoviev:2003ix, Zinoviev:2003dd} on the basis of the simplest mixed-symmetry fields.

In this paper massless mixed-symmetry fields are reformulated within the unfolded approach \cite{Vasiliev:1988xc, Vasiliev:1988sa, Vasiliev:1992gr} because it is the unfolded approach that underlies the full nonlinear theory
of interacting massless fields with arbitrary totally symmetric spins
\cite{Vasiliev:2003ev}, being the only approach succeeded in
constructing the full theory, though exhaustive results concerned
with cubic vertices of higher-spin fields were obtained within the
light-cone approach in \cite{Metsaev:1993mj,Metsaev:2005ar,Metsaev:2007rn}.
Therefore, to unfold an arbitrary spin, viz., mixed-symmetry, fields
in the Minkowski space is considered as the first step towards the full
nonlinear theory of arbitrary spin fields.

The main statement of the paper is that a free massless field with spin degrees of freedom characterized by an arbitrary bosonic or fermionic unitary irreducible
representation of the Wigner's little algebra can be uniquely described within the unfolded approach, in which all gauge symmetries are manifest. The unfolded system has the form of a covariant constancy equation. The gauge fields and gauge parameters are differential forms on the Minkowski space with values in certain irreducible representations of the Lorentz algebra, i.e., irreducible tensors or spin-tensors. The full unfolded system is described in terms of a single nilpotent operator $\sigma_-$, whose cohomology groups correspond to independent differential gauge parameters, dynamical fields, gauge-invariant equations and Bianchi identities.

Another advantage of the unfolded approach is in that the equations for bosons and fermions have literally the same form, the only difference being in change of tensor representations, in which the fields takes values, by the corresponding spin-tensors.
The form and the order of dynamical equations, second for bosons and first for fermions, turns out to be encoded in $\sigma_-$ cohomology. In fact the unfolded system is constructed for the bosonic case and, then, proved to have the same form for fermions. The similarity between bosons and fermions within the unfolded approach can have deep applications in theories with supersymmetries.

Despite the deep relations of the unfolded approach to the nonlinear theory of higher-spin fields, unfolding by itself provides a very powerfull
method for analysis of dynamical systems. For instance,
once some linear dynamical system is unfolded it is given a direct interpretation in terms of Lie algebras/modules
and all gauge symmetries become manifest.

The paper is organized as follows: the main result, i.e., the unfolded
form of equations describing a massless field with the spin
that corresponds to an arbitrary irreducible representation of the
Wigner's little algebra is stated in Section \ref{Results}. All the necessary
information about mixed-symmetry fields in the Minkowski space-time is
collected in Section \ref{FlatMS}. The basic facts concerning the
unfolded approach, viz., the very definition, the relation to Lie
algebras/modules, to the Chevalley-Eilenberg cohomology are recalled in
Section \ref{Unfld}, illustrated on the examples of a scalar field,
spin-one field and totally symmetric spin-$s$ and spin-$(s+\frac12)$ fields in
Section \ref{UnfldExamples}. The proof of the general statement of Section
\ref{Results} is in Section \ref{MSUnfld}. The physical degrees of freedom are analyzed in Section \ref{MSPDOFCounting}. The discussion of the results and conclusions are in Section \ref{Conclusions}.
Multi-index notation and basic facts on Young diagrams and irreducible representations are collected in Appendices.
\section*{Conventions}
As the most general type of irreducible representations of
orthogonal algebras, viz., the Wigner's little algebra, the Lorentz
algebra, is considered, the essential use is made of Young diagrams'
language. A certain Young diagram is denoted by $\Yy$ with
subscripts or directly enumerating the lengths of the rows as
\Y{s_1,s_2,..} or, when rows of equal lengths are combined to
blocks, as \Y{(s_1,p_1),(s_2,p_2),...}, $p_i$ being the number of
rows of length $s_i$. Loosely speaking we do not make any difference
between irreducible finite-dimensional representations of
orthogonal algebras, Young diagrams\footnote{In addition to the Young symmetry conditions, extra restrictions (with the aid of invariant tensors: Levi-Civita for \sld, metric and Levi-Civita for \sod) have to be imposed on the tensors to make them irreducible. In what follows it is important that irreducible \sod-tensors are traceless. No special consideration is given to (anti)-self dual fields, see Appendix B.}  and the corresponding irreducible (spin)-tensors\footnote{In the case of fermionic representations of orthogonal algebras, i.e., spin-tensors, the tensor part of a spin-tensor (all but one spinor indices can be converted pairwise to tensor indices by means of $\Gamma$-matrices, hence, we consider spin-tensors with one spinor index only) is characterized by Young diagram, which is labeled by the subscript $\frac12$, e.g., an irreducible rank-two symmetric tensor-spinor $\psi^{\alpha\fsp ab}$ satisfies ${\Gamma^\alpha_a}_\beta\psi^{\beta\fsp ab}=0$ and belongs to $\Y{2}_\ferm$. The connection with the standard Gelfan-Zeitlin labels is obvious. Additional conditions, viz., Majorana, Weyl and Majorana-Weyl are irrelevant to the problems concerned.}, e.g., rank-two symmetric
traceless tensor-valued field $\phi^{ab}$, i.e.,
$\phi^{ab}=\phi^{ba}$ and $\eta_{ab}\phi^{ab}=0$, can be
equivalently denoted either as $\phi^{\smallpic{\YoungB}}$ or
$\phi^{\Yy}$ with $\Yy=\Y{2}\equiv\Y{(2,1)}$. The scalar
representation $\Y{0}$ is denoted by $\bullet$. Unless otherwise
stated, all Young diagrams are of orthogonal algebras, viz.,
\lorentz\ or \mls. For more detail on Young diagrams see in Appendix B. Greek indices $\mu$, $\nu$,...=0...(d-1) are the
world indices of the Minkowski space-time $\mathcal{M}_d$. $d\equiv
dx^\mu\frac\pl{\pl x^\mu}$ is the exterior differential on
$\mathcal{M}_d$. The degree of differentials forms on
$\mathcal{M}_d$ is indicated by the bold subscript, e.g., a
degree-$q$ differential form $\omega$ on $\mathcal{M}_d$ with values
in \lorentz-\irrep\ characterized by the Young diagram $\Yy$ is
denoted as $\omega^{\Yy}\fm{q}$ (loosely speaking $\Yy$-valued
degree-$q$ form $\omega^{\Yy}\fm{q}$). The wedge symbol $\wedge$ is
systematically omitted. Lowercase Latin letters $a$,
$b$,...=0...(d-1) are vector indices of \lorentz, fiber indices of the sections of
tensor bundles over the Minkowski space-time. Greek indices $\alpha, \beta, \gamma=1...2^{[\frac{d}2]}$ are fiber spinor indices of \lorentz. The multi-index
condensed notation is used in the paper: the (anti)-symmetrization
is denoted by placing the corresponding indices in (square) round
brackets, for details on the multi-index notation see Appendix A.

\section{Summary of Results}\label{Results}

The main statement of the paper is that given a unitary irreducible
bosonic(fermionic) representation of the massless Wigner's little algebra \mls,
which is characterized by Young diagram
$\Yy=\Y{(s_1,p_1),...,(s_N,p_N)}$ ($\Yy=\Y{(s_1,p_1),...,(s_N,p_N)}_\ferm$), there exists a uniquely
determined unfolded system that describes a massless spin-\Yy\ field, with all gauge symmetries being manifest. The system has
the form of a covariant constancy equation
\begin{align}\label{ResultsFullSystem}
    \DD\omega\fm{p}&=0, &\qquad \omega\fm{p}&\in\WW\fm{p},\nonumber\\
    \delta \omega\fm{p}&=\DD\xi\fm{p-1},& \qquad \xi\fm{p-1}&\in\WW\fm{p-1},\nonumber\\
    \delta \xi\fm{p-1}&=\DD\xi\fm{p-2},&\qquad \xi\fm{p-2}&\in\WW\fm{p-2},\nonumber\\
    ..&..,& ..&..,\nonumber\\
    \delta \xi\fm{1}&=\DD\xi\fm{0}, &\qquad \xi\fm{0}&\in\WW\fm{0},
    \end{align}
where $\WW=\bigoplus_{q=0}^{\infty}\WW\fm{q}$ is certain graded space,  $\DD$ is a nilpotent operator of degree $(+1)$, $\DD:\WW\fm{q}\rightarrow\WW\fm{q+1}$ and $\DD^2=0$. Gauge fields take values in $\WW\fm{p}$, where $p=\sum_{i=1}^{i=N}p_i$ is the height of Young diagram \Yy, the first level gauge parameters in $\WW\fm{p-1}$,
the second level gauge parameters in $\WW\fm{p-2}$ and so on. The gauge invariance is manifest by virtue of $\DD^2=0$. The reducibility of gauge transformations is similar to those of totally anti-symmetric fields.

Space $\WW\fm{p}$, which contains the gauge fields of the
unfolded system, is a graded by nonnegative integer $g=0,1,...$ set
of differential forms
$\WW\fm{p}=\{\omega^{\Yy_0}\fm{q_0},\omega^{\Yy_1}\fm{q_1},...,\omega^{\Yy_g}\fm{q_g},...\}$,
$q_0=p$. Diagrams $\Yy_g$ that characterize \lorentz-irreducible
representations, in which the fields and gauge parameters take values,
are uniquely determined by the initial diagram
$\Yy=\Y{(s_1,p_1),...,(s_N,p_N)}_{(\ferm)}$ of \mls.

The \dynamical\ field is
incorporated in a $p$-form $\omega^{\Yy_0}\fm{p}\in\WW^{g=0}\fm{p}$
that takes values in the irreducible representation(\irrep) of the
Lorentz algebra \lorentz\ that is characterized by Young diagram
$\Yy_0$ of the
form\be\label{ResultsPhysicalVielbein}\FieldContentDynamical,\ee
i.e., it is obtained by cutting off the first column of $\Yy$. All
other gauge fields in the system are auxiliary and can be expressed
in terms of derivatives of $\omega^{\Yy_0}\fm{p}$. It is convenient to
enumerate the Lorentz-\irreps\ in which gauge fields
$\omega^{\Yy_g}\fm{q_g}$ take values by a pair $\{n,k\}$ of
integers. Roughly speaking, the first integer is related to the
number of the block of $\Yy$, $n=N,...,0$, the second one is related
to the relative length of the $n$-th and $(n+1)$-th blocks. The
\lorentz-\irreps\ $\Yy_{g=0}$,...,$\Yy_{g=(s_N-1)}$ are given by
$\Yy_{\{n,k\}}$ with $n=N$, $k=0...(s_N-1)$ of the form
\be\label{ResultsFirstAuxiliary}\FieldContentN.\ee
The diagrams $\Yy_{g}$ with $g=s_N...(s_{N-1}-1)$ are given by $\Yy_{\{n,k\}}$ with $n=N-1$, $k=0...(s_{N-1}-s_N-1)$ of the form
\be\FieldContentkk.\ee
and analogously for the rest of $\Yy_{g}$ with $g<s_1$. The diagrams $\Yy_{g}$ with $g=s_1,s_1+1,...$ are given by $\Yy_{\{n,k\}}$ with $n=0$, $k=0,1,...$ of the form
\be\FieldContentWeyl.\ee
Gauge fields $\omega^{\Yy_g}\fm{q_g}$ for $g\sim\{n=N,k\}$ are $(p_1+...+p_N)$-forms, gauge fields $\omega^{\Yy_g}\fm{q_g}$ for $g\sim\{n=N-1,k\}$ are $(p_1+...+p_{N-1})$-forms, ..., gauge fields $\omega^{\Yy_g}\fm{q_g}$ for $g\sim\{n=0,k\}$ are zero-forms. Hence, form degree function $q_g$ is completely defined. The grade $g$ is equal to the element number (starting from zero) in the set of pairs $\{n,k\}$ ordered by $k$ in increasing order and, then, by $n$ in decreasing order.

Space $\WW\fm{p-1}$, which contains the first level gauge parameters of the system, is a set of forms with values in the same \lorentz-\irreps\ as gauge fields but the form degree is less by one, i.e., $\WW\fm{p-1}=\{\xi^{\Yy_0}\fm{q_0-1},\xi^{\Yy_1}\fm{q_1-1},...,\xi^{\Yy_g}\fm{q_g-1},...\}$. The sector $\WW\fm{p-1}^{g}$ for $g\geq s_1$ is trivial du to $q_{g}=0$. Analogously, spaces $\WW\fm{q}$ for $q<p-1$ and $q>p$ can be defined.
To sum up, the element of space $\WW\fm{p\pm i}$ at grade $g\sim\{n,k\}$ is a degree-$(q_g\pm i)$ form with values in $\Yy_g\equiv\Yy_{\{n,k\}}$ \irrep\ of the Lorentz algebra.

The Minkowski background space is described in terms of vielbein(tetrad) $h^a_\mu dx^\mu$ and Lorentz spin-connection $\varpi^{a,b}_\mu dx^\mu$, which determines Lorentz-covariant derivative $D_L=d+\varpi$.

When reduced to $\WW^g\fm{q}$ the full system has the form
\begin{align} \label{ResultsSystem}
    D_L\omega^{\Yy_g}\fm{q_g}&=\sigma_-\left(\omega^{\Yy_{g+1}}\fm{q_{g+1}}\right), & &
    \omega^{\Yy_g}\fm{q_g}\in\WW^g\fm{p},\quad \omega^{\Yy_{g+1}}\fm{q_{g+1}}\in\WW^{g+1}\fm{p}, \nonumber\\
    \delta\omega^{\Yy_g}\fm{q_g}&=D_L\xi^{\Yy_g}\fm{q_g-1}+\sigma_-\left(\xi^{\Yy_{g+1}}\fm{q_{g+1}-1}\right), & &
    \xi^{\Yy_g}\fm{q_g-1}\in\WW^g\fm{p-1},\quad \xi^{\Yy_{g+1}}\fm{q_{g+1}-1}\in\WW^{g+1}\fm{p-1},\nonumber\\
    \delta\xi^{\Yy_g}\fm{q_g-1}&=...,& & &
\end{align}
where operator $\DD$ is a sum $\DD=D_L-\sigma_-$ of
Lorentz-covariant derivative $D_L$ and certain nilpotent operator
$\sigma_-:\bigwedge^q\otimes\Yy_{g+1}\rightarrow\bigwedge^{q+\Delta_g+1}\otimes\Yy_g$,
$\Delta_g=q_g-q_{g+1}\geq0$, $(\sigma_-)^2=0$, built of
background vielbein $h^a_\mu dx^\mu$, which is unambiguously fixed
by the symmetry of $\Yy_g$ and $\Yy_{g+1}$. $\sigma_-$ contracts $(\Delta_g+1)$ vielbeins $h^a$ with the tensor representing $\Yy_{g+1}$
to obtain the tensor with the symmetry of $\Yy_g$, appropriate Young
symmetrizers are implied.

Since $\sigma_-$ affects tensor indices only it is correctly defined on spin-tensors too and does not violate the $\Gamma$-tracelessness condition.
The unfolded equations for fermions have the same form as for bosons, the irreducible tensors are to be replaced with corresponding spin-tensors.

The case of the last block of the length one, i.e., single column,
is not special but requires some comments. Since $(s_N-1)=0$, it is not possible to add a cell to the
bottom-left of the $N$-th block of $\Yy_{0}$, therefore, $\Yy_{\{n=N,k=0\}}\equiv\Yy_{0}$ is the only diagram with $n=N$ and
diagram $\Yy_{\{n=N-1,k=0\}}\equiv\Yy_{g=1}$ has the symmetry of $\Y{(s_1-1,p_1),...,(s_{N-1}-1,p_{N-1}),(1,p_N+1)}$.

Subspace $\WW^{\{n,k\}}\fm{q}$ with definite $n$ forms an
irreducible module of \iso, whereas the subspace with definite both
$n$ and $k$ forms a finite-dimensional irreducible module of
\lorentz, i.e., an irreducible Lorentz tensor, characterized by
$\Yy_g$, as was stated above. The intervals of constancy of $q_g$
correspond to $g_{n,k}$ with definite $n$, i.e., the set of forms,
on which a certain irreducible \iso-module is realized, all have the
same degree.

Let us also note that $\WW\fm{k\geq p}$ are infinite-dimensional, whereas $\WW\fm{k<p}$ are finite-dimensional. The higher degree spaces $\WW\fm{k\geq p}$ correspond to the equations of motion ($\WW\fm{k=p+1}$) and Bianchi identities  ($\WW\fm{k>p+1}$), which manifest the gauge symmetries. Most of the equations express higher grade $g>0$ fields via the derivatives of \physical\ field $\omega^{\Yy_0}\fm{q_0}$ and only certain elements of $\WW^{g=2}\fm{k=p+1}$ impose on $\omega^{\Yy_0}\fm{q_0}$ second order dynamical equations. The significance of the fields with $g>0$ is to make all gauge symmetries be manifest.

\section{Mixed-Symmetry
Fields in Minkowski Space}\setcounter{equation}{0}\label{FlatMS} The types of the Minkowski
space \particles, being, by definition, in one-to-one
correspondence with unitary irreducible representations(\uirrep)
of the Poincare algebra \iso, in the case of four space-time
dimensions were classified by Wigner in \cite{Wigner:1939cj}.

Leaving out the details of the Wigner's construction, for recent
reviews and for generalization to an arbitrary space-time dimension $d$ see \cite{Bekaert:2006py}, important is that given
$m^2\geqslant0$ and a \uirrep\ $\Yy$ of the Wigner's little
algebra, being \msv\ for $m^2>0$ and \mls\ for $m^2=0$, there exists a
standard procedure to construct a \uirrep\ of \iso, which is
called a massive(massless) \particle\ of spin-$\Yy$. So-called continuous or infinite spin particles
\cite{Wigner:1939cj}, \cite{Bekaert:2005in} are not considered in this paper. Therefore, physical degrees of freedom for massive and massless \particles\ are classified by irreducible tensors of \msv\ and \mls, respectively.

More elaborated are the two
cases of totally-symmetric
spin-$s$
\particles\ $\Yy=\Y{s}\equiv\Y{(s,1)}$ \cite{Singh:1974qz, Fronsdal:1978rb, Bekaert:2005vh}, viz., scalar,
vector, graviton, and of totally anti-symmetric spin-$p$ \particles\
$\Yy=\Y{(1,p)}$ \cite{Townsend:1979hd, Sezgin:1980tp, Blau:1989bq}. The
others are referred to collectively as mixed-symmetry, the
simplest one being $\Yy=\Y{2,1}\equiv\Y{(2,1),(1,1)}$.

Yet different problem is to realize a \uirrep\ of \iso\ on the
solutions of a wave equation for a field $\phi_{\Yy_M}(x)$,
which takes values in a certain representation $\Yy_M$ of the Lorentz
algebra \lorentz, i.e., $\phi_{\Yy_M}(x)$ is a Lorentz tensor or a set of
tensors. As field theories free \particles\ can be
described in either non-gauge or gauge way, in the former case a \uirrep\ of \iso\ is realized on the
solutions of the wave equation directly, whereas for the latter case a \uirrep\ of \iso\
is realized on the quotient of the solutions by certain specific solutions,
called pure gauge. The wave equation $(\square+m^2)\phi_{\Yy_M}(x)=0$,
which fixes the quadratic Casimir $m^2$ of \iso, is generally
supplemented with a set of algebraic/differential constraints to
single out an irreducible, in the sense of the little algebra,
component. Field $\phi_{\Yy_M}(x)$ takes values in a certain
finite-dimensional representation of \lorentz, which is
not irreducible in most cases, nevertheless it can be made irreducible when dealing with free equations of motion only. Yet
more different problem is to realize an \irrep\ of \iso\ on the
solutions of the variational problem for some action, in most
cases the procedure requires a set of auxiliary fields, which carry no physical degrees of freedom.

The choice of \lorentz-representation $\Yy_M$ (even if irreducible), in
which field $\phi_{\Yy_M}(x)$ takes values, is not unique, e.g., a free
massless spin-one \particle\ can be described either by a
gauge potential $A_\mu(x)$ subjected to \be \square
A_\mu-\pl_\mu\pl^\nu A_\nu=0, \qquad \delta A_\mu=\pl_\mu\xi,\ee or
by a field strength $F_{\mu\nu}$ subjected to \be \pl^\mu
F_{\mu\nu}=0,\quad\pl_{[\mu}F_{\lambda\rho]}=0.\ee In the case of
a massless spin-two \particle\ ($\Yy={\smallpic\YoungB_{\mls}}$), in
addition to the conventional description in terms of the metric
$g_{\mu\nu}$ \be R_{\mu\nu}-\frac12g_{\mu\nu}R=0\ee the Weyl
tensor $C_{\mu\nu,\lambda\rho}$ known to have the symmetry of
\smallpic{\YoungBB}\footnote{In antisymmetric basis
$C_{\mu\nu,\lambda\rho}$ satisfies
$C_{\mu\nu,\lambda\rho}=-C_{\nu\mu,\lambda\rho}$,
$C_{\mu\nu,\lambda\rho}=-C_{\mu\nu,\rho\lambda}$ and
$C_{[\mu\nu,\lambda]\rho}=0$.} can describe a free spin-two in a
non gauge way \be \pl_{[\mu}C_{\mu\mu],\nu\nu}=0,\quad
\pl^{\lambda}C_{\mu\lambda,\nu\nu}=0.\ee

Another example is a $4d$ massless scalar particle, which can be
described either by a scalar field $\phi(x)$ subjected to $\square
\phi=0$ or, more exotically, by an antisymmetric gauge field $\omega_{\mu\nu}$ (so-called notoph
\cite{Ogievetsky:1967ij})
subjected to \be \square \omega_{\mu\nu}-\pl_\mu
\pl^{\rho}\omega_{\rho\nu}+\pl_\nu
\pl^{\rho}\omega_{\rho\mu}=0,\quad \delta \omega_{\mu\nu}=\pl_\mu
\xi_\nu-\pl_\nu \xi_\mu.\ee This equation describes a massless
\particle\ with spin-${\smallpic\YoungAA_{\mls}}\equiv\Y{1,1}_{\mls}$, which for $d=4$ by virtue
of the \mls\ Levi-Civita tensor $\varepsilon_{ij}$ is equivalent to a
scalar, $\Y{1,1}_{\mls}\sim\Y{0,0}_{\mls}$. On the other hand, a scalar
\particle\ can be described by a rank-$d$ antisymmetric field $\omega_{\mu_1...\mu_d}$
satisfying $\square\omega_{\mu_1...\mu_d}=0$, where the use of
\lorentz-duality is made of,
$\omega_{\mu_1...\mu_d}=\varepsilon_{\mu_1...\mu_d}\phi$. These
two types of duality are referred to as trivial.

The general statement is that free \particles\ can be described by
an infinite number of ways, called \dual\ descriptions
\cite{Curtright:1980yj,Boulanger:2003vs,Zinoviev:2005ux,
Zinoviev:2005qp,Zinoviev:2005zj}, but \dual\ theories exhibit
certain difficulties while introducing interactions
\cite{Bekaert:2002uh,Deser:1980fy,Townsend:1981nu,Townsend:1979yv},
e.g., despite the fact that free massless spin-one and spin-two
\particles\ can be described by the Maxwell field strength and by the
Weyl tensor, respectively, introducing interactions requires the
corresponding gauge potentials $A_\mu$ and $g_{\mu\nu}$ to be
brought in.

There exists a distinguished choice of \lorentz-\irrep\ $\Yy_M$, in which field $\phi_{\Yy_M}(x)$
takes values, that can be referred to as fundamental or
minimal. For the minimal description of a spin-$\Yy$ \particle\
the spin degrees of freedom and the \lorentz-\irrep\ $\Yy_M$  are characterized by the same Young
diagram, i.e., $\Yy_M=\Yy$, e.g., a spin-one particle by $A_\mu$, a spin-two by
$g_{\mu\nu}$. All other descriptions are referred to as \dual. It is the minimal descriptions that will be discussed further, by this reason the
term spin-$\Yy$ \particle\ can be substituted for more accepted spin-$\Yy$ \field.

A massive totally symmetric spin-$s$ \field\ can be described \cite{Dirac:1936tg}
by a totally symmetric tensor field $\phi_{(\mu_1...\mu_s)}$
subjected to
\be\begin{split}\label{FlatMSSymmetric}&(\square+m^2)\phi_{\mu_1...\mu_s}=0,\\
&\pl^\nu\phi_{\nu\mu_2...\mu_s}=0,\\
&\tr{\phi}{\mu_3...\mu_s}{\nu}=0,\end{split}\ee where the last
equation (tracelessness condition) makes the tensor irreducible in the \lorentz\ sense,
the first one puts the system \onmassshell\ and the second one projects out the
components orthogonal to the momentum, restricting
$\phi_{\mu_1...\mu_s}$ to contain only a spin-$s$ \irrep\ of \msv.

Analogously, a massive totally anti-symmetric spin-$p$ \field,
i.e., $\Yy=\Y{(1,p)}$, can be described by a totally
anti-symmetric tensor field $\omega_{[\mu_1...\mu_p]}$ subjected
to
\be\begin{split}\label{FlatMSAntiSymmetric}&(\square+m^2)\omega_{\mu_1...\mu_p}=0,\\
&\pl^\nu\omega_{\nu\mu_2...\mu_p}=0,
\end{split}\ee
where the tracelessness condition becomes trivial for anti-symmetric tensors.

These two results are easily generalized to an arbitrary
spin-$\Y{s_1,...,s_n}$ massive \field, which can be minimally
described by a symmetric in each group of indices tensor field
$\phi_{{\mu}_1(s_1),{\mu}_2(s_2),...,{\mu}_n(s_n)}$ subjected to
\be\begin{split}\label{FlatMSEquations} &(\square+m^2)\phi_{{\mu}_1(s_1),{\mu}_2(s_2),...,{\mu}_n(s_n)}=0, \\
&\pl^{\mu_i}\phi_{{\mu}_1(s_1),{\mu}_2(s_2),...,{\mu}_n(s_n)}=0, \quad\qquad\qquad\qquad i\in[1,n],\\
&\phi_{{\mu}_1(s_1),...,({\mu}_k(s_k),...,{\mu}_k)
{\mu}_i(s_i-1),...,{\mu}_n(s_n)}=0, \qquad k\in[1,n-1], \quad k<i,
\\ &\eta^{{\mu}_i {\mu}_j}\phi_{{\mu}_1(s_1),{\mu}_2(s_2),...,{\mu}_n(s_n)}=0,\qquad\qquad\qquad
i,j\in[1,n],
\end{split}\ee
where the last two conditions are just the Young symmetry and the
tracelessness conditions, which make the field carry an
\irrep-$\Y{s_1,...,s_n}$ of \lorentz, and can be thought of as the part of the definition of
field $\phi_{\Yy_M}(x)$. The first equation puts the system \onmassshell, the
second one projects out all \msv-\irreps, which the tensor
$\phi_{{\mu}_1(s_1),{\mu}_2(s_2),...,{\mu}_p(s_p)}$ decomposes
into, except for the one with the symmetry of $\Y{s_1,...,s_n}$
\msv-diagram.

Let \Verma{m^2}{\Y{s_1,...,s_n}} be an \iso\ module extracted by
(\ref{FlatMSEquations}), being irreducible for $m^2\neq0$ or for
$\Y{0}$ (scalar \field). Since minimally described massless
\fields\ are gauge theories a \uirrep\ \Irrep{0}{\Y{s_1,...,s_n}} of \iso\
corresponding to a massless spin-\Y{s_1,...,s_n}\
\field\ should be defined as appropriate quotient by the pure gauge solutions of the form
\be\ComplexC{\mbox{gauge solutions}}{\mbox{all
solutions}}{\Irrep{0}{\Y{s_1,...,s_n}}},\ee where the sequence is
non-split, as there is no \lorentz-covariant way to extract
\mls-\irrep\ $\Y{s_1,...,s_n}$ from the Lorentz tensor with the same
symmetry of $\Y{s_1,...,s_n}$ by virtue of a single
\lorentz-vector, the momentum $p_\mu\sim\pl_\mu$.

A massless totally symmetric spin-$s$ \field\ can be described
as the quotient of (\ref{FlatMSSymmetric}) with $m^2=0$ by
pure gauge solutions of the form
$\delta\phi_{\mu_1...\mu_s}=\pl_{(\mu_1}\xi_{\mu_2...\mu_s)}$,
where $\xi_{\mu_1...\mu_{s-1}}$ is a totally symmetric gauge
parameter subjected to the equations of the same form
(\ref{FlatMSSymmetric}), i.e., \onmassshell, tracelessness and
transversality, thus, belonging to $\Verma{0}{\Y{s-1}}$. The
definition of \Irrep{0}{\Y{s}} is given by \be
\ComplexC{\Verma{0}{\Y{(s-1,1)}}}{\Verma{0}{\Y{(s,1)}}}{\Irrep{0}{\Y{(s,1)}}}.\ee
There is a bit difference for a totally anti-symmetric spin-$p$
massless \field. Pure gauge solutions are defined analogously
as $\delta\omega_{\mu_1...\mu_p}=
\pl_{[\mu_1}\xi_{\mu_2...\mu_p]}$, where
$\xi_{[\mu_1...\mu_{p-1}]}$ is a rank-$(p-1)$ antisymmetric gauge
parameter subjected to the equation of the same form
(\ref{FlatMSAntiSymmetric}) and, thus, belonging to
$\Verma{0}{\Y{(1,p-1)}}$. In contrast to a totally symmetric
spin-$s$ massless \field, the gauge transformations are
reducible in the sense that $\delta\omega_{\mu_1...\mu_p}\equiv0$
provided that one transforms gauge parameter
$\delta\xi_{\mu_1...\mu_{p-1}}=
\pl_{[\mu_1}\xi_{\mu_2...\mu_{p-1}]}$ with a second level
rank-$(p-2)$ antisymmetric gauge parameter
$\xi_{[\mu_1...\mu_{p-2}]}$, and so on. Some components
(parallel to the momentum) of the first level gauge parameter
$\xi_{\mu_1...\mu_{p-1}}$ do not contribute to the gauge law for
$\omega_{\mu_1...\mu_p}$, these are represented by the second
level gauge parameter $\xi_{\mu_1...\mu_{p-2}}$ modulo those
components of $\xi_{\mu_1...\mu_{p-2}}$ that do not contribute to
$\xi_{\mu_1...\mu_{p-1}}$ and so on till $\delta \xi_\mu=\pl_\mu
\xi$. Gauge parameters $\xi_{\mu_1...\mu_{p-1}}$,
$\xi_{\mu_1...\mu_{p-k}}$, $k\in[1,p]$ and $\xi$ are referred to
as the first level, the $k$-th level and the deepest level of
reducibility, respectively. The corresponding \iso-\uirrep\ is
given by a non-split exact sequence of the form
\begin{align}0\longrightarrow\Verma{0}{\Y{(0,0)}}&\longrightarrow\Verma{0}{\Y{(1,1)}}\longrightarrow\nonumber\\...&...\longrightarrow\Verma{0}{\Y{(1,p)}}\longrightarrow\Irrep{0}{\Y{(1,p)}}\longrightarrow0.\end{align}
For example, Maxwell gauge potential $A_\mu$ possesses gauge parameters of the first
level only as $p=1$ in this case.

Though, a considerable success in describing \onmassshell\
massless \fields\ was achieved, for instance, within the
light-cone approach \cite{Metsaev:1997nj}, there are many reasons
to have an \offshell\ gauge symmetry, i.e., to construct equations
that are invariant with respect to gauge transformations with gauge
parameters not subjected to $\square\xi=0$. To make the symmetry
\offshell\ generally requires to relax the irreducibility of the
\lorentz-representation in which field $\phi_\Yy(x)$ takes values.

For instance, a massless totally symmetric spin-$s$ field can be
described\cite{Skvortsov:2007kz} by a traceless rank-$s$ symmetric
field $\phi_{\mu_1...\mu_s}$ with an \offshell\ gauge symmetry
\begin{align}
    &\square\phi_{\mu_1...\mu_s}-s\pl_{(\mu_1}\pl^\nu\phi_{\nu\mu_2...\mu_s)}+%
    {\scriptstyle\frac{s(s-1)}{(d+2s-4)}}\eta_{(\mu_1\mu_2}\pl^\nu\pl^\rho\phi_{\nu\rho\mu_3...\mu_s)}=0,\quad
    \tr{\phi}{\mu_3...\mu_s}{\nu}=0,\nonumber\\
    &\delta\phi_{\mu_1...\mu_s}=\pl_{(\mu_1}\xi_{\mu_2...\mu_s)},\quad%
    \tr{\xi}{\mu_3...\mu_{s-1}}{\nu}=0,\quad \pl^\nu\xi_{\nu\mu_2...\mu_{s-1})}=0,
\end{align}
where in order for $\phi_{\mu_1...\mu_s}$ to be traceless the
gauge parameter has to be not only traceless but transverse also,
this still being true for general mixed-symmetry fields.
Apparently, to get rid of any differential constraints on gauge
parameters, the tracelessness constraint for field
$\phi_{\mu_1...\mu_s}$ has to be relaxed. The same spin-$s$ massless
\field\ can be described\cite{Fronsdal:1978rb} by field $\phi_{\mu_1...\mu_s}$
subjected to\footnote{Obtained from the Lagrangian, equations of \cite{Fronsdal:1978rb} have the form $G_{\mu_1...\mu_s}-\frac{s(s-1)}4\eta_{(\mu_1\mu_2}G^{\nu}_{\phantom{\nu}\nu\mu_3...\mu_s)}=0$, where $G_{\mu_1...\mu_s}$ is equal to (\ref{FlatMSFronsdal}), the two forms being equivalent.}
\be\label{FlatMSFronsdal}\begin{split}
    &\square\phi_{\mu_1...\mu_s}-s\pl_{(\mu_1}\pl^\nu\phi_{\nu\mu_2...\mu_s)}+
    {\scriptstyle\frac{s(s-1)}{2}}\pl_{(\mu_1}\pl_{\mu_2}\tr{\phi}{\mu_3...\mu_s)}{\nu}=0,\quad
    \tr{\phi}{\mu_5...\mu_s}{\nu\rho}=0,\\
    &\delta
    \phi_{\mu_1...\mu_s}=\pl_{(\mu_1}\xi_{\mu_2...\mu_s)},\quad
    \tr{\xi}{\mu_3...\mu_{s-1}}{\nu}=0,
\end{split}\ee
where in order to get rid of any differential constraints on gauge
parameters a tracelessness is relaxed to a double-tracelessness,
i.e., field $\phi_{\mu_1...\mu_s}$ takes values in the direct
sum of two \lorentz-\irreps, being symmetric traceless tensors of
ranks $s$ and $(s-2)$.  The general statement is that for
equations of motion to have an \offshell\ gauge symmetry the
\lorentz-\irrep\ in which field $\phi_\Yy(x)$ takes values has to be
reducible, with additional direct summands representing certain
nonzero traces. These additional fields are called auxiliary and
carry no physical degrees of freedom. Imposing certain gauge, equations (\ref{FlatMSEquations}) can be restored.

A massless totally antisymmetric spin-$p$ \field\ can be
described by an antisymmetric rank-$p$ tensor field
$\omega_{[\mu_1...\mu_p]}$ subjected to \be\square
\omega_{\mu_1...\mu_p}-p\pl_{[\mu_1}\pl^{\nu}\omega_{\nu\mu_2...\mu_p]}=0,
\quad \delta\omega_{\mu_1...\mu_p}=
\pl_{[\mu_1}\xi_{\mu_2...\mu_p]},\ee where the symmetry at the all
levels of reducibility is manifest and \offshell, i.e.,
$\delta\xi_{\mu_1...\mu_{p-1}}=
\pl_{[\mu_1}\xi_{\mu_2...\mu_{p-1}]}$ for the second level gauge
parameter $\xi_{\mu_1...\mu_{p-2}}$ not subjected to any
differential constraints, and analogously for the gauge symmetries
at deeper levels.

Much similar to totally anti-symmetric fields mixed-symmetry
massless fields possess reducible gauge transformations, i.e., for
a field $\phi$ with equations of motion invariant under gauge
transformations $\delta_{\xi_1} \phi=\sum_{i_1}\pl\xi^{i_1}_1$
there exist the second level gauge transformations
$\delta_{\xi_2}\xi^{i_1}_1=\sum_{i_2}\pl\xi^{i_2}_2$ such that
$\delta_{\xi_2}\phi\equiv0$, there exist the third level gauge
transformations
$\delta_{\xi_3}\xi^{i_2}_2=\sum_{i_3}\pl\xi^{i_3}_3$ such that
$\delta_{\xi_3}\xi_1\equiv0$ and so on. The difference from
totally anti-symmetric fields is in that there are generally more
than one gauge parameters at each level of reducibility
(enumerated by index $i_k$ at the $k$-th level).

For instance, the simplest mixed-symmetry massless \field\ has the
spin-\smallpic{\YoungBA} and can be minimally
described\cite{Curtright:1980yk, Aulakh:1986cb} by a field
$\phi_{[\mu\mu],\nu}$, which is anti-symmetric in the first two
indices and satisfies Young symmetry condition\footnote{In
more detail, the field $\phi_{\mu\nu,\lambda}$ satisfies
$\phi_{\mu\nu,\lambda}=-\phi_{\nu\mu,\lambda}$,
$\phi_{\mu\nu,\lambda}+\phi_{\nu\lambda,\mu}+\phi_{\lambda\mu,\nu}=0$.
Equivalently, a symmetric basis can be used, i.e,
$\phi^S_{\mu\nu,\lambda}=\phi^S_{\nu\mu,\lambda}$,
$\phi^S_{\mu\nu,\lambda}+\phi^S_{\nu\lambda,\mu}+\phi^S_{\lambda\mu,\nu}=0$.
The two bases are related by
$\phi^S_{(\mu\nu),\lambda}=\frac1{\sqrt{3}}(\phi^A_{\mu\lambda,\nu}+\phi^A_{\nu\lambda,\mu})$,
$\phi^A_{[\mu\nu],\lambda}=\frac1{\sqrt{3}}(\phi^S_{\mu\lambda,\nu}-\phi^S_{\nu\lambda,\mu})$.}
$\phi_{[\mu\mu,\nu]}=0$. The equations of
motion\be\label{FlatMSHookA}\square\phi_{\mu\mu,\lambda}+2\pl_{[\mu}\pl^{\lambda}\phi_{\mu]\lambda,\nu}-
\pl_\nu\pl^\lambda\phi_{\mu\mu,\lambda}-2\pl_\nu\pl_{[\mu}\phi_{\mu]\lambda,}^{\phantom{\mu]\lambda,}\lambda}=0\ee
are invariant under
\be\label{FlatMSHookB}\delta\phi_{\mu\mu,\nu}=\pl_{[\mu}\xi^S_{\mu]\nu}+\pl_{[\mu}\xi^A_{\mu]\nu}-\pl_\nu\xi^A_{\mu\mu},\ee
with symmetric and anti-symmetric gauge parameters
$\xi^S_{(\mu\nu)}$ and $\xi^A_{[\mu\nu]}$. Let us stress that to
maintain an \offshell\ gauge symmetry field
$\phi_{[\mu\mu],\nu}$ has to take values in a reducible \lorentz\
representation, $\smallpic{\YoungBA}\oplus\smallpic{\YoungA}$, so
does symmetric gauge parameter $\xi^S_{(\mu\mu)}$,
$\smallpic{\YoungB}\oplus\bullet$, which is in accordance with
$\phi_{\mu\nu,}^{\phantom{\mu\nu,}\nu}\neq0$ and
$\xi_{\nu}^{S\nu}\neq0$. Analogously to totally anti-symmetric fields
there exist second level gauge transformations with a vector parameter $\xi_\mu$
\be\label{FlatMSHookC}\begin{split}
    &\delta\phi_{\mu\mu,\nu}=0,\\
    &\delta\xi^A_{\mu\nu}={\scriptstyle\frac23}(\pl_\mu\xi_\nu-\pl_\nu\xi_\mu), \\& \delta
    \xi^S_{\mu\nu}=\pl_\mu \xi_\nu+\pl_\nu \xi_\mu.
\end{split}\ee In this case \Irrep{0}{\smallpic{\YoungBA}} is given by a non-split exact
sequence of the form
\be\label{FlatMSHookSequence}\ComplexD{\Verma{0}{\smallpic{\YoungA}}}{\Verma{0}{\smallpic{\YoungAA}}\oplus\Verma{0}{\smallpic{\YoungB}}}{\Verma{0}{\smallpic{\YoungBA}}}{\Irrep{0}{\smallpic{\YoungBA}}}\ee

In the general case of a massless
spin-$\Yy=\Y{(s_1,p_1),...,(s_N,p_N)}$ \field\ the depth of
reducibility of gauge transformations is equal to
$p=\sum_{i=1}^{i=N}p_i$, where $p$ is the height of the first
column of the Young diagram and at the $r$-th level of
reducibility gauge parameters have the symmetry of
\be\label{FlatMSGaugeParameters}\GaugeParameters:\quad\sum_{i=1}^{i=N}k_i=\sum_{i=1}^{i=N}p_i-r\ee
This pattern of reducibility of gauge transformations will be of
great importance while constructing the unfolded formulation in
Section \ref{MSUnfld}. As is easily seen, at the first level of
reducibility, the gauge parameters are various tensors, whose Lorentz
Young diagrams are obtained by cutting off one cell from the
original \mls-diagram in all possible ways, i.e., the number of
gauge parameters at the first level is equal to the number of
blocks $N$, e.g., two for spin-\smallpic{\YoungBA}. There is
only one gauge parameter at the deepest level of reducibility,
whose Young diagram is obtained by cutting off the first column
from $\Yy$, e.g., \smallpic{\YoungA} for spin-\smallpic{\YoungBA}.

This pattern corresponds, of course, to \onshell\ equations, i.e.,
the gauge parameters taking values in \lorentz-\irreps\ of Young symmetry (\ref{FlatMSGaugeParameters}) are subjected to (\ref{FlatMSEquations})-like
equations. To obtain an \offshell\ gauge symmetry
the field content has to be extended to
take values in certain reducible \lorentz-representations, additional
components turn out can be identified with certain traces of a single
field with the symmetry of $\Yy$ as an \sld-tensor. In general, gauge parameters are also reducible tensors with the symmetry of (\ref{FlatMSGaugeParameters}). In the case of a spin-\smallpic{\YoungBA} massless
\field, the trace $\phi_{\mu\nu,}^{\phantom{\mu\nu,}\nu}$ in the sector of fields and the trace $\xi^{S\nu}_{\nu}$ in the sector of
gauge parameters are the additional components. In the general
case of a spin-$\Yy=\Y{s_1,...,s_n}$ massless \field,  field
$\phi_{{\mu}_1(s_1),{\mu}_2(s_2),...,{\mu}_n(s_n)}$ should satisfy
\cite{Labastida:1986ft}
\be\label{FlatMSLabastidaDoubleTracelessness}\eta^{\mu_i\mu_i}\eta^{\mu_i\mu_i}\phi_{{\mu}_1(s_1),{\mu}_2(s_2),...,{\mu}_n(s_n)}=0,\qquad
i\in[1,n],\ee which is a generalization of the Fronsdal's
double-trace condition (\ref{FlatMSFronsdal}). As it will be shown,
this condition naturally arises in the unfolded approach. Note
that double-tracelessness is imposed on each group of
symmetric indices and it is not required for cross-traces to
vanish.

Let a generalized Weyl tensor for a minimally described spin-$\Yy$ massless
\field\ be a gauge-invariant combination of the least order in derivatives of field
$\phi_{\Yy}(x)$ that is allowed to be nonzero \onmassshell. On the
other hand it is the generalized Weyl tensor that the minimal non-gauge
description of a massless spin-$\Yy$ \field\ is based on. For a
spin-$\Yy=\Y{(s_1,p_1),...,(s_N,p_N)}_{\mls}$ massless
\field\ the generalized Weyl tensor has the symmetry of
$\Y{(s_1,p_1+1),(s_2,p_2),...,(s_N,p_N)}_{\lorentz}$ and is of the
$s_1$-th order in derivatives. In the case of spin-one
($\Yy={\smallpic\YoungA_{\mls}}$)  Maxwell field strength $F_{\mu\nu}$
with the symmetry of \smallpic{\YoungAA} can be also called a generalized Weyl
tensor.

Fermionic mixed-symmetry fields share most features of bosonic ones, viz., the reducibility of gauge transformations, the enlargement of the field content for the equations of motion to possess an \offshell\ gauge invariance. The difference is that the equations for fermions have the first order in derivatives and the irreducibility of spin-tensors is achieved by the $\Gamma$-tracelessness condition\footnote{$\Gamma_\mu$ are the Clifford algebra generators and satisfy $\Gamma_\mu\Gamma_\nu+\Gamma_\nu\Gamma_\mu=2\eta_{\mu\nu}$, $\fpl\equiv \Gamma^\mu\pl_\mu$. $\Gamma$-trace is a contraction of a spinor index and one tensor index with $\Gamma_\mu$, e.g., $\Gamma_\nu\phi^{\mu\nu}\equiv{\Gamma^\alpha_\nu}_\beta\phi^{\beta\fsp\mu\nu}$. } instead of the tracelessness one.

For example, a massless totally symmetric spin-$(s+\frac12)$ field can be described \offshell\cite{Fang:1978wz} by a totally symmetric spin-tensor field $\phi_{\alpha\fsp(\mu_1...\mu_s)}$ subjected to\footnote{Obtained from the Lagrangian, equations of \cite{Fang:1978wz} have the form $G_{\mu_1...\mu_s}-\frac{s}2\Gamma_{(\mu_1}\Gamma^\nu G_{\nu\mu_2...\mu_s)}-\frac{s(s-1)}4\eta_{(\mu_1\mu_2}G^{\nu}_{\phantom{\nu}\nu\mu_3...\mu_s)}=0$, $G_{\mu_1...\mu_s}=\fpl\phi_{\mu_1...\mu_s}-s\pl_{(\mu_1}\Gamma^{\nu}\phi_{\nu\mu_2...\mu_s)}=0$, which is equivalent to (\ref{FlatMSFronsdalFermionic}).}
\be\label{FlatMSFronsdalFermionic}\begin{split}
    &\fpl\phi_{\mu_1...\mu_s}-s\pl_{(\mu_1}\Gamma^{\nu}\phi_{\nu\mu_2...\mu_s)}=0,\\
    &\delta \phi_{\mu_1...\mu_s}=\pl_{(\mu_1}\xi_{\mu_2...\mu_s)},\\
    &\Gamma^\nu\Gamma^\rho\Gamma^\lambda\phi_{\nu\rho\lambda\mu_4...\mu_s}=0,\qquad \Gamma^\nu\xi_{\nu\mu_2...\mu_{s-1}}=0,
\end{split}\ee
where in order to get an \offshell\ gauge invariance the irreducibility of $\phi_{\alpha\fsp(\mu_1...\mu_s)}$ has to be relaxed to the triple $\Gamma$-tracelessness.

A massless totally anti-symmetric spin-$\Y{(1,p)}_\ferm$ field can be described \offshell\ by a totally anti-symmetric spin-tensor field $\omega_{\alpha\fsp[\mu_1...\mu_p]}$ subjected to
\be\label{FlatMSFronsdalFermionic}\begin{split}
    &\fpl\omega_{\mu_1...\mu_p}-p\pl_{[\mu_1}\Gamma^{\nu}\omega_{\nu\mu_2...\mu_p]}=0,\\
    &\delta \omega_{\mu_1...\mu_p}=\pl_{[\mu_1}\xi_{\mu_2...\mu_p]},\\
\end{split}\ee
Similar to the bosonic case, (\ref{FlatMSFronsdalFermionic}) possesses reducible gauge transformations. The difference is that one can impose the triple $\Gamma$-tracelessness on $\omega_{\alpha\fsp[\mu_1...\mu_p]}$ but this restricts the first level gauge parameter to be $\Gamma$-traceless, $\Gamma^\nu\xi_{\nu\mu_2...\mu_{p-1}}=0$, and, hence, the second order gauge parameter has to be \onmassshell, i.e., $\fpl\xi_{\mu_1...\mu_{p-2}}=0$. Therefore, for equations of motion to possess an \offshell\ gauge symmetry of all orders no $\Gamma$-trace conditions have to be imposed on field/gauge parameters.

In the general case of a massless spin-$\Yy=\Y{(s_1,p_1),...,(s_N,p_N)}_\ferm$ \field, the pattern of gauge symmetries is given by the spin-tensors with the tensor part described by (\ref{FlatMSGaugeParameters}), the definition of the Weyl tensor remains unchanged also.

The descriptions based on tensor field $\phi_{{\mu}_1(s_1),{\mu}_2(s_2),...,{\mu}_n(s_n)}$, which is analogous to the metric $g_{\mu\nu}$, are referred to as \metric. At least in writing
$\square\equiv\pl_\mu\pl_\nu\eta^{\mu\nu}$ the explicit use of
metric $\eta^{\mu\nu}$ is made of, which complicates the issue of introducing interactions with gravitation.

Let us note that in principle one can verify the gauge invariance of the field equations for massless mixed-symmetry fields despite the fact that the very form of gauge transformation is cumbersome due to Young symmetrizers. The advantage of the unfolded approach is in that gauge invariance at all levels of reducibility is manifest.

\section{Unfolding Dynamics and $\sigma_-$
Cohomology}\setcounter{equation}{0}\label{Unfld} In this section we,
first, recall the definition of unfolding and the relation of the
simplest unfolded systems to Lie algebras/modules and
Chevalley-Eilenberg cohomology with coefficients. Second, peculiar
properties of unfolded systems that describe free \fields\ and
specifically the so-called $\sigma_-$ cohomology concept are
recalled. Third, the very procedure of constructing the unfolded
form is illustrated on the examples of massless
spin-zero, spin-one, arbitrary totally symmetric spin-$s$ and spin-$(s+\frac12)$ fields,
the relation with the general statement  of Section \ref{Results} is
pointed out in each of the examples.
\subsection{General Features}

Some dynamical system is said to be unfolded
\cite{Vasiliev:1988xc, Vasiliev:1988sa, Vasiliev:1992gr} if it has
the form \be\label{UnfldEquations} d W^\aA = F^\aA(W),\ee where
$W^\aA$ is a set\footnote{In this section indices $\aA$, $\aB$,
$\aC$ are of arbitrary nature. In the cases of practical
significance $\aA$, $\aB$, $\aC$ vary over certain \irreps\ of the
Lorentz algebra.} of differential forms of degree-$q_\aA$ on some
$d$-dimensional manifold $\mathcal{M}_d$, $d$ - exterior
differential on $\mathcal{M}_d$ and $F^\aA(W)$ is an arbitrary
degree-$(q_\aA+1)$ function of $W^\aA$ assumed to be expandable in
terms of exterior (wedge) products only\footnote{The wedge symbol
$\wedge$ will be systematically omitted further.}

\be
F^\aA(W)=\sum_{n=1}^{\infty}\sum_{q_{\aB_1}+...+q_{\aB_n}=q_\aA+1}
f^\aA_{\phantom{\aA} \aB_1 ... \aB_n}W^{\aB_1}\wedge...\wedge
W^{\aB_n},\ee where $f^\aA_{\phantom{\aA} \aB_1 ... \aB_n}$ are
constant coefficients satisfying $f^\aA_{\phantom{\aA}
\aB_1...\aB_i\aB_j...\aB_n}=(-)^{q_{\aB_i}q_{\aB_j}}f^\aA_{\phantom{\aA}
\aB_1...\aB_j\aB_i...\aB_n}$. Moreover, $F^\aA(W)$ must satisfy the
integrability condition (called generalized Jacobi identity)
obtained by applying $d$ to (\ref{UnfldEquations})
\be\label{UnfldBianchi} F^\aB\frac{\delta F^\aA}{\delta
W^\aB}\equiv0.\ee Any solution of (\ref{UnfldBianchi}) defines a
free differential algebra (FDA) \cite{Sullivan77, D'Auria:1982nx,
D'Auria:1982pm, Nieuwenhuizen:1982zf}. If Jacobi identities
(\ref{UnfldBianchi}) are satisfied irrespective of $\mathcal{M}_d$
dimension\footnote{As the forms with the rank greater than $d$ are
identically zero, there exist certain identities, e.g.,
$W\fm{n}\wedge W\fm{n}\equiv0$ for $n+m>d$, which make the operator
$\frac{\delta }{\delta W^\aA}$ to be ill-behaved.}, the free
differential algebra is referred to as universal
\cite{Bekaert:2005vh, Vasiliev:2007yc}. It is the universal algebras
only that will be considered further.

Equations (\ref{UnfldEquations}) are invariant under gauge
transformations  \be\label{UnfldGauge}\delta
W^\aA\fm{q_{\aA}}=d\epsilon^\aA\fm{q_{\aA}-1}+\epsilon^\aB\fm{q_{\aB}-1}\frac{\delta
F^\aA}{\delta W^\aB},\qquad \mbox{for }q_{\aA}>0,\ee
\be\label{UnfldGaugeZeroForm}\delta
W^\aA\fm{0}=\epsilon^{\aB'}\fm{0}\frac{\delta F^\aA}{\delta
W^{\aB'}\fm{1}},\quad \aB': q_{\aB'}=1,\qquad \mbox{for
}q_{\aA}=0,\ee where $\epsilon^\aA\fm{q_{\aA}-1}$ is a
degree-$(q_{\aA}-1)$ form taking values in the same space as
$W^\aA\fm{q_{\aA}}$. In its
turn, $\delta W^\aA\fm{q_{\aA}}=0$ can be treated as unfolded-like system for
$\epsilon^\aA\fm{q_{\aA}-1}$, i.e.,
$d\epsilon^\aA\fm{q_{\aA}-1}=-\epsilon^\aB\fm{q_{\aB}-1}\frac{\delta
F^\aA}{\delta W^\aB}$, there emerge second level gauge
transformations
\be\label{UnfldGaugeRed}\delta
\epsilon^\aA\fm{q_{\aA}-1}=d\xi^\aA\fm{q_{\aA}-2}-\xi^\aB\fm{q_{\aB}-2}\frac{\delta
F^\aA}{\delta W^\aB},\ee provided that $F^\aA(W)$ is linear in
matter fields and analogously for the gauge transformations at
deeper levels. Therefore, the reducibility of gauge transformations
is manifest in the unfolded approach. For a degree-$q_\aA$ gauge
field $W^\aA\fm{q_\aA}$ there exist $q_\aA$ levels of gauge
transformations.

The use of the exterior algebra respects diffeomorphisms, which is
very appropriate for introducing interactions with the gravitation.
The whole information about the dynamics turns out to be contained
in $F^\aA(W)$ and one can even extend an unfolded system to other
manifolds\cite{Vasiliev:2007yc} simply by changing the exterior
differential, the new unfolded system has literally the same form.

Note that introducing enough auxiliary fields it is possible to
reformulate any dynamical system in the unfolded form, although it
may be difficult to unfold some particular system or to find all
unfolded forms.

Collected below are some important cases of unfolded systems,
which have a direct bearing on Lie algebras \cite{Sullivan77, Vasiliev:2007yc}.
\begin{description}
  \item[Lie algebras/Flat connections.] Let $\Omega^I\equiv\Omega^I_\mu dx^\mu$ be a subsector of degree-one
    forms. The only self-closed unfolded equations are of the form
    \be d\Omega^I=-f^I_{JK}\Omega^J\Omega^K.\label{UnfldFlatness}\ee
    Generalized Jacobi identity (\ref{UnfldBianchi}) implies the Jacobi
    identity $f^I_{JK}f^J_{LM}\Omega^K\Omega^L\Omega^M\equiv0$ for some Lie algebra $\mathfrak{g}$  with structure constants $f^I_{JK}$. Therefore, the closed subsector of one-forms is in one-to-one
    correspondence with Lie algebras and (\ref{UnfldFlatness}) is
    the flatness condition for a connection $\Omega^I$ of
    $\mathfrak{g}$. This provides a coordinate-independent framework for describing background geometry.
    In the cases of interest, $\mathfrak{g}$ is \iso, \ads, \ds\ and $\mathfrak{sp}(2n)$ \cite{Vasiliev:2007yc}. Background geometry connection $\Omega^I$ is assumed to be
    of order zero, whereas all matter fields, including dynamical
    gravitation, are of the first order. All equations
    are assumed to be of the first order in matter fields and, hence, describe
    free fields only.

    In this paper
    $\mathfrak{g}=\iso$ and $\Omega^I=\{\varpi^{a,b}, h^a\}$, where $\varpi^{a,b}\equiv\varpi^{a,b}_\mu dx^\mu$
    is a Lorentz spin-connection and $h^a\equiv h^a_\mu dx^\mu$ is
    a background vielbein, which defines a non-holonomic basis of a tangent space at each point of the manifold. Flatness equation (\ref{UnfldFlatness}) for \iso-connection $\Omega^I$ reads
    \be\label{UnfldIsoFlatness}\begin{split}
        &dh^a+\varpi^{a,}_{\phantom{a,}b}h^b=0,\\
        &d\varpi^{a,b}+\varpi^{a,}_{\phantom{a,}c}\varpi^{c,b}=0.
    \end{split}\ee
    The first is the zero torsion equation that expresses the Lorentz
    spin-connection via the first derivative of
    $h^a_\mu$. The second can be recognized as the zero curvature
    equation.
    For example, in Cartesian coordinates the explicit solution is
    $\varpi^{a,b}_\mu=0$ and $h^a_\mu=\delta^a_\mu$.

    The advantage
    of description of background geometry as the flatness condition
    for a connection of the space-time symmetry algebra is in that this way is
    coordinate-independent. In what follows we assume $\varpi^{a,b}$, $h^a$ to satisfy
    (\ref{UnfldIsoFlatness}), which is enough if there is no need for the explicit form of the solution
    in some particular coordinate system. For instance, in
    (anti)-de Sitter space (\ref{UnfldIsoFlatness}) is modified by
    the terms proportional to the cosmological constant $\lambda^2$
     \be\label{UnfldAdSFlatness}\begin{split}
        &dh^a+\varpi^{a,}_{\phantom{a,}b}h^b=0,\\
        &d\varpi^{a,b}+\varpi^{a,}_{\phantom{a,}c}\varpi^{c,b}+\lambda^2h^ah^b=0
    \end{split}\ee
    and admits no simple solutions with $h^a_\mu=\delta^a_\mu$, $\varpi^{a,b}=0$. Nevertheless, without giving the explicit solution, to imply
    that $\varpi^{a,b}$, $h^a$ satisfy (\ref{UnfldAdSFlatness})
    is sufficient for general analysis, e.g., for constructing Lagrangians \cite{Alkalaev:2006rw}.
  \item[Contractible FDA.] The simplest equations linear in matter
    fields of the form
    \be dW\fm{q}^\aA=f^\aA_{\phantom{\aA}\aB}W^{\aB}\fm{q+1}\label{UfldContractible}\ee
    can be reduced by linear transformations to either
    \be dW\fm{q}^\aA=W^{\aA}\fm{q+1},\qquad dW^{\aA}\fm{q+1}=0,\label{UfldContractibleA}\ee
    or \be dW\fm{q}^\aA=0,\ee where the second equation of
    (\ref{UfldContractibleA}) is the consequence of Jacobi
    identities (\ref{UnfldBianchi}) for the first one.
    In the first case, by virtue of gauge transformations (\ref{UnfldGauge}) $\delta W\fm{q}^\aA=d\xi^\aA\fm{q-1}+\chi^\aA\fm{q}$,
    $\delta W^{\aB}\fm{q+1}=d\chi^\aA\fm{q}$ the field $W\fm{q}^\aA$ can be gauged away. In
    both cases, by virtue of the Poincare's Lemma these equations are dynamically empty and correspond to the co-called contractible FDAs \cite{Sullivan77}.

  \item[$\mathfrak{g}$-modules/Covariant constancy equations.]
    Let $W^\aA\fm{q}$ be a closed subsector of $q$-forms of matter gauge fields. Linear
    in matter fields equations may involve the background
    $\mathfrak{g}$-connection $\Omega^I$, which is of zeroth order.
    Such equations referred to as linearized over $\mathfrak{g}$ background
    (described by any solution $\Omega^I$ of (\ref{UnfldFlatness})) have the
    form
    \be d W^\aA\fm{q}=-\Omega^I{f_I}^\aA_{\phantom{\aA}\aB}W^\aB\fm{q}.\ee
    Jacobi identity (\ref{UnfldBianchi}), where
    (\ref{UnfldFlatness}) is also taken into account,
    \be\Omega^J\Omega^K\left(-f^I_{JK}{f_I}^\aA_{\phantom{\aA}\aB}+{f_J}^\aA_{\phantom{\aA}\aC}{f_K}^\aC_{\phantom{\aC}\aB}\right)W^\aB\fm{q}=0\ee implies
    ${f_I}^\aA_{\phantom{\aA}\aB}$ to realize a representation of
    $\mathfrak{g}$. Therefore, the closed subsector of forms of definite
    degree is in one-to-one correspondence with $\mathfrak{g}$-modules, whereas \be \label{UnfldCovariantConstancy}D_\Omega W^\aA\fm{q} \equiv d
    W^\aA\fm{q}+\Omega^I{f_I}^\aA_{\phantom{\aA}\aB}W^\aB\fm{q}=0\ee
    is a covariant constancy equation and $(D_\Omega)^2=0$ since the connection is flat (\ref{UnfldFlatness}).
    In the cases of interest, $\aA$ runs over certain
    finite-dimensional \lorentz-\irreps, i.e.,
    $\mathfrak{g}$-modules decompose with respect to its subalgebra $\lorentz\subset\mathfrak{g}$ into a direct sum
    of certain irreducible tensors. This is why we single out the Lorentz-covariant derivative $D_L$ from the
    whole $\mathfrak{g}$-covariant derivative $D_\Omega$, the
    remaining part acts vertically (algebraically).
    In the case of $\Omega^I$ being an \iso\ flat connection, the Lorentz covariant derivative
    $D_L=d+\varpi$ satisfy ${D_L}^2=0$, as the exterior differential
    $d$ does. \be D_L T^{ab...}=dT^{ab...}+\varpi^{a,}_{\phantom{a,}c}T^{cb...}+\varpi^{b,}_{\phantom{b,}c}T^{ac...}+...\ee
    In Cartesian coordinates $D_L\equiv d$ and we do not
    make any distinction between $d$ and $D_L$, whereas in (anti)-de Sitter $(D_L)^2\neq0$ and the difference between
    $d$ and $D_L$ has to be taken into account. Also, zero torsion equation (\ref{UnfldIsoFlatness}.1) can be rewritten as $D_L h^a=0$. The remaining part
    of $\Omega^I$, which acts algebraically and mixes different
    \lorentz-modules into $\mathfrak{g}$-modules is associated with the generators of translations with gauge field $h^a$.

  \item[Gluing $\mathfrak{g}$-modules/Chevalley-Eilenberg cohomology.]
    Let $W\fm{p}$, $W\fm{q}$ and $W\fm{r}$ take values in $\mathfrak{g}$-modules $\module_1$, $\module_2$ and $\module_3$, the representations are realized
    by operators $T_1$, $T_2$ and $T_3$, respectively. Still linear in matter fields
    but nonlinear in background connection $\Omega^I$ equations are of the form
    \be\begin{split}& D_\Omega W\fm{p}\equiv dW\fm{p}+T_1(\Omega) W\fm{p}= f_{12}(\Omega,...,\Omega)W\fm{q}, \\
                    & D_\Omega W\fm{q}\equiv dW\fm{q}+T_2(\Omega) W\fm{q}= f_{23}(\Omega,...,\Omega)W\fm{r},\\
                    & D_\Omega W\fm{r}\equiv....,\end{split}\label{UnfldGluing}\ee
    where the terms forming $\mathfrak{g}$-covariant derivative
    are isolated on the l.h.s. The two $\mathfrak{g}$-modules  $\module_1$, $\module_2$ appear to be glued together
    by the term $f_{12}(\Omega,...,\Omega)\in {\rm Hom}(\Lambda^{p-q+1}(\mathfrak{g})\otimes
    \module_2,\module_1)$.
    Jacobi identity (\ref{UnfldBianchi}) implies
    $f_{12}(\Omega,...,\Omega)$ to be a Chevalley-Eilenberg
    cocycle with coefficients in $\module_2^\ast\otimes\module_1$, where $\module_2^\ast$ is a module contragradient to $\module_2$,
    and $f_{12}(\Omega,...,\Omega)f_{23}(\Omega,...,\Omega)=0$.
    Coboundaries can be proved to be dynamically empty and can be removed by a field
    redefinition. Consequently, $f_{12}(\Omega,...,\Omega)$ should be a
    nontrivial representative of the Chevalley-Eilenberg cohomology
    group with coefficients in $\module_2^\ast\otimes\module_1$.
    It follows also that nothing but zero forms can be joined to
    zero forms, i.e., the only possible linearized unfolded equations
    on forms of zero degree are (\ref{UnfldCovariantConstancy}) with
    $q=0$.
\end{description}

For any dynamical system, linearized over certain $\mathfrak{g}$-background (described by
any solution $\Omega^I$ of (\ref{UnfldFlatness})),  equations of motion (\ref{UnfldCovariantConstancy}) along with gauge
transformations of all orders of reducibility acquire a very simple
form \be\label{UnfldLinearizedSystem}\begin{split}
    \delta \xi\fm{1}&=D_\Omega\xi\fm{0},\\
    &...,\\
    \delta \xi\fm{p-1}&=D_\Omega\xi\fm{p-2},\\
    \delta W\fm{p}&=D_\Omega\xi\fm{p-1},\\
    D_\Omega W\fm{p}&=0,\end{split}\ee
where $W\fm{p}$ takes values in certain $\mathfrak{g}$-module
$\module$, $D_\Omega$ is the associated $\mathfrak{g}$-covariant
derivative and $\xi\fm{p-k}$, $k\in[1,p]$ are the gauge parameters
of $k$-th order of reducibility taking values in the same
$\mathfrak{g}$-module $\module$ and being forms of degree
$(p-k)$. The gauge invariance at each order of reducibility and of
equations of motion is due to $(D_\Omega)^2=0$. In some sense the
gauge fields $W\fm{p}$ and the gauge parameters $\xi\fm{p-k}$ seem
to play the same role at the linearized level. This uniformity is
of essential importance when analyzing unfolded systems.

In the presence of gluing terms, e.g., when only two modules, say
$\module_1$ and $\module_2$, are glued by nontrivial
Chevalley-Eilenberg cocycle $f_{pr}(\Omega,...,\Omega)$  full
chain of equations and gauge transformations
(\ref{UnfldLinearizedSystem}) is modified to\footnote{The sign
factor $(-)^{p-r+1}$ before $f_{pr}(\Omega,...,\Omega)$ is ignored
and is thought of as the part of the definition of
$f_{pr}(\Omega,...,\Omega)$.}
\begin{align}\label{UnfldLinearizedSystemWithCocycles}
    \delta \xi^1\fm{1}&=D_\Omega\xi^1\fm{0}, &  & \nonumber\\
    & ..., & &\nonumber\\
    \delta \xi^1\fm{p-r}&=D_\Omega \xi^1\fm{p-r-1}, & \nonumber\\
    \delta \xi^1\fm{p-r+1}&=D_\Omega \xi^1\fm{p-r}+f_{pr}(\Omega,...,\Omega)\xi^2\fm{0}, & \delta \xi^2\fm{1}&=D_\Omega\xi^2\fm{0},\nonumber\\
    & ...,& &...,\\
    \delta \xi^1\fm{p-1}&=D_\Omega \xi^1\fm{p-2}+f_{pr}(\Omega,...,\Omega)\xi^2\fm{r-2}, & \delta \xi^2\fm{r-1}&=D_\Omega\xi^2\fm{r-2},\nonumber\\
    \delta W^1\fm{p}&=D_\Omega \xi\fm{p-1}+f_{pr}(\Omega,...,\Omega)\xi^2\fm{r-1}, & \delta W^2\fm{r}&=D_\Omega\xi^2\fm{r-1},\nonumber\\
     D_\Omega W^1\fm{p}&=f_{pr}(\Omega,...,\Omega)W^2\fm{r} & D_\Omega W^2\fm{r}&=0,\nonumber
\end{align}
where $W^1\fm{p}$, $\xi^1\fm{p-k}$ and $W^2\fm{r}$, $\xi^2\fm{r-m}$
take values in $\mathfrak{g}$-modules $\module_1$ and $\module_2$,
respectively. Important is that the gauge fields/parameters
taking values in $\module_2$ contribute to the r.h.s. of the
equations/gauge transformations for the gauge fields/parameters
taking values in $\module_1$ but for forms of degree greater than
$(p-r)$. The above two systems are nothing but the specializations
of (\ref{UnfldEquations}), (\ref{UnfldGauge}),
(\ref{UnfldGaugeRed}).

Also, let us note that there is no strong reason to make any
distinction between the terms linear in the background
$\mathfrak{g}$-connection $\Omega^I$, which correspond to
$\mathfrak{g}$-modules, and the terms of higher order in
$\Omega^I$, which correspond to Chevalley-Eilenberg cocycles. Both
terms can be combined into a single object, the generalized
covariant-derivative,
$\DD=d+T(\Omega)+f(\Omega,...,\Omega)\equiv
D_\Omega+f(\Omega,...,\Omega)$ with the property
$\DD^2=0$. Moreover, by means of $\DD$
(\ref{UnfldLinearizedSystemWithCocycles}) and
(\ref{UnfldLinearizedSystem}) can be rewritten in a similar
manner. Consequently, at the linearized level the most general
unfolded equations and gauge transformations have the form
\be\label{UnfldLinearizedSystemGeneral}\begin{split}
    \delta \xi\fm{1}&=\DD\xi\fm{0},\\
    &...,\\
    \delta \xi\fm{p-1}&=\DD\xi\fm{p-2},\\
    \delta W\fm{p}&=\DD\xi\fm{p-1},\\
    \DD W\fm{p}&=0,\end{split}\ee
where $\DD$ is built of $d$, background connection $\Omega^I$ and
satisfies $\DD^2=0$. The gauge fields/parameters take values in
certain spaces $\WW\fm{q}$, viz., $W\fm{p}\in\WW\fm{p}$, $\xi\fm{p-1}\in \WW\fm{p-1}$,...,
$\xi\fm{0}\in \WW\fm{0}$. The elements of $\WW\fm{q}$ are
differential forms with values in certain $\mathfrak{g}$-modules. It
is also convenient to define higher degree spaces $\WW\fm{q>p}$: the
equations $R\fm{p+1}=\DD W\fm{p}=0$ take values in $\WW\fm{p+1}$,
$R\fm{p+2}=\DD R\fm{p+1}$ belongs to $\WW\fm{p+2}$ and by virtue of
$\DD^2=0$ satisfies $R\fm{p+2}\equiv0$ and so on. Therefore, there
exist certain identities, which belong to $\WW\fm{p+2}$, for the
equations, which belong to $\WW\fm{p+1}$, and there exist certain
identities, which belong to $\WW\fm{p+3}$ for the identities in
$\WW\fm{p+2}$ and so on. At the field theoretical level these
identities correspond to Bianchi identities for the first level
gauge transformations, higher identities correspond to the Bianchi
identities for the deeper levels of gauge transformations.

A remark should be made that the forms belonging to $\WW\fm{q}$ may
have different degrees if there are Chevalley-Eilenberg cocycles
present. As a result, spaces $\WW\fm{q}$ may contain different number
of elements. For instance, (\ref{UnfldLinearizedSystemWithCocycles})
can be reduced to  (\ref{UnfldLinearizedSystemGeneral}) if
one defines \be\WW\fm{q}=\begin{cases}
    \{W^1\fm{q}\},& q<p-r,\\
    \{W^1\fm{q},W^2\fm{q-p+r}\},& q\geq(p-r),\\
\end{cases}\ee
where $W^1\fm{q}$ and $W^2\fm{q-p+r}$ take values in
$\mathfrak{g}$-modules $\module_1$ and $\module_2$, respectively. In
terms of (\ref{UnfldLinearizedSystemWithCocycles}) the elements of
$\WW\fm{q}$ for $q<(p-r)$ were referred to as $\xi^1\fm{q}$, the
elements of $\WW\fm{q}$ for $q=(p-r)...(p-1)$ to as $\xi^1\fm{q}$,
$\xi^2\fm{q-p+r}$, the elements of $\WW\fm{q}$ for $q=p$ to as
$W^1\fm{p}$, $W^2\fm{r}$, the elements of $\WW\fm{q}$ for $q>p$ to
as $R^1\fm{q}$, $R^2\fm{q-p+r}$ (the equations of motion and Bianchi
identities).

\subsection{Interpretation of Unfolded Systems Describing
free Fields}\label{UnfldInterpretation}

When describing free fields, unfolded formulations will be referred to
as \framelike\ \cite{Alkalaev:2003qv, Alkalaev:2006rw}, though the
very term \framelike\ has a more broad definition. These systems
consist of equations (\ref{UnfldFlatness}), describing background
geometry, of a finite chain of (\ref{UnfldGluing})-like equations,
describing the gauge fields of the model and  of
(\ref{UnfldGluing})-like equations with $q=0$ glued to gauge forms, i.e., it
requires a Lie algebra $\mathfrak{g}$ (commonly \iso, \ads\ or \ds),
a set of $\mathfrak{g}$-modules $\module_0$,...,$\module_N$ and an
appropriate set of nontrivial Chevalley-Eilenberg cocycles
$f_{0,1}$,...,$f_{N-1,N}$. The physical degrees of freedom are
contained in the forms of zero degree and the module $\module_N$ in
which zero degree forms take values is infinite-dimensional.

Since \framelike\ unfolded systems contain inevitably infinitely many
fields, most of them being either auxiliary or
\Stueckelberg\footnote{A field is called \Stueckelberg\ if it can be
gauged away by pure algebraic symmetry.}, there arises a problem of
reconstruction a \metric\ formulation by a given unfolded
formulation or of extracting the dynamical content for a given
unfolded system.

The questions to be answered are: what are the dynamical fields,
what are the differential gauge parameters and what are the gauge
invariant equations of motion. These answers are not universal and
depend on the chosen scheme of \interpretation. Different
\interpretations\ correspond to \dual\ descriptions of the same
dynamical system.

\Framelike\ unfolded systems are endowed at the linearized level
with additional structures, providing a natural way for
\interpretation. Indices $\aA, \aB,...$ vary over certain finite
dimensional Lorentz \irreps, i.e., each of $\module_n$, $n=0,..,N$
decomposes as $\module_n=\sum_k\moduleP_{n,k}$, where
$\moduleP_{n,k}$ are certain \lorentz-modules characterized by $\Yy_{\{n,k\}}$ Young diagrams of Section \ref{Results}.

We say that some \framelike\ unfolded system is given an
\interpretation\ if \cite{Shaynkman:2000ts}

\begin{enumerate}
    \item On the whole space $\WW=\bigoplus_{q=0}^{\infty}\WW\fm{q}$, where the matter gauge fields, gauge parameters,
    equations of motion and Bianchi identities take values,
    there exists a bounded from below grading $g=0,1,...$, i.e., $\WW\fm{q}=\bigoplus_{g=0}^{\infty}\WW\fm{q}^g$. The homogeneous
    element of $\WW^g\fm{q}$ is a certain differential form with values in certain \lorentz-\irrep\ and is denoted as $W^g\fm{q}$. In simplest cases the
    grade is just the rank of the \lorentz-\irrep\ the element of  $\WW^g\fm{q}$
    takes values in.
    \item The closed subsector of one-forms $\Omega^I$ describes a
    background geometry, viz., Minkowski or (anti)-de Sitter. The
    generalized covariant derivative $\DD$ is to be divided into three parts,
    the last one not being necessary nontrivial
    \be \DD=D_L-\sigma_-+\sigma_+,\ee
    where $D_L$ is a background Lorentz-covariant derivative and has a zero grade,
    $\sigma_-$ is an operator of grade $(-1)$ and $\sigma_+$ contains positive grade operators. The only differential part is in $D_L$,
    whereas $\sigma_{\pm}$ acts vertically and is built of the background vielbein $h^a$.
    If two modules are glued the gluing element is supposed to be of
    grade $(-1)$ and also denoted as  $\sigma_-$.
\end{enumerate}
The equations then have the form \be D_L
W^n+\sum_{i=1}^{g-n}\sigma_+\left(W^{n+i}\right)=\sigma_-\left(W^{n+1}\right),\qquad
g=0,1,...\ee and analogously for the gauge transformations. In the
flat space all operators of positive grade appear to be trivial for
the massless case. As it does not affect the analysis let us
consider a simplified version with $\sigma_+=0$. Then,
$\DD^2={D_L}^2+\sigma_-D_L+D_L\sigma_-+(\sigma_-)^2=0$ is equivalent
to ${D_L}^2=0$, $(\sigma_-)^2=0$ and $\sigma_-D_L+D_L\sigma_-=0$.
The first is a part of the flatness condition
(\ref{UnfldIsoFlatness}) for the \iso-connection. The second is the
nilpotancy condition for $\sigma_-$. The third is satisfied provided
that $\sigma_-$ is twisted by the factor $(-)^{\Delta_g}$, where
$(\Delta_g+1)$ is a degree of $\sigma_-$, which is equal to the number
of the vielbeins $h^a$ that $\sigma_-$ is built of. Indeed, the
vielbeins $h^a$ anticommute with $D_L$ and, hence, the action of $\sigma_-$ is
to be twisted by the factor $(-)^{\Delta_g}$ in order $\sigma_-D_L+D_L\sigma_-=0$ holds true. Without further
mention, the pure sign factor $(-)^{\Delta_g}$ will be though of as
a part of the definition of $\sigma_-$.

It is useful
to illustrate the action of $\sigma_-$ when, for example, unfolded equations on fields with values in two modules $\module_1$
and $\module_2$ are glued as in (\ref{UnfldLinearizedSystemWithCocycles}) \be\TwoModulesGlued,\ee where short arrows
stand for the action of $\sigma_-$ within each module and long
arrows for the action of $\sigma_-$ between two modules (gluing
terms) and dots stand for gauge fields/parameters at different
grades,  bold dots represent the elements of $\WW\fm{p}=\{W^1\fm{p},W^2\fm{r}\}$.

All information about the dynamics turns out to be concealed in
cohomology groups of $\sigma_-$ \cite{Shaynkman:2000ts},
$\coh(\sigma_-)=\frac{Ker(\sigma_-)}{Im(\sigma_-)}$.
Let us analyze unfolded system
(\ref{UnfldLinearizedSystemGeneral}), which at grade $g$ has the
form \be\label{UnfldSimplest}\begin{split}
    \delta \xi^g\fm{1}&=D_L\xi^g\fm{0}+\sigma_-(\xi^{g+1}\fm{0}),\\
    &...,\\
    \delta W^g\fm{p}&=D_L\xi^g\fm{p-1}+\sigma_-(\xi^{g+1}\fm{p-1}),\\
    0&=D_L W^g\fm{p}+\sigma_-(W^{g+1}\fm{p}),\end{split}\ee
Beginning from the deepest level of gauge transformations and from
the lowest grade, it is obvious that those $\xi^{g+1}\fm{0}$ that
are not $\sigma_-$-closed can be treated as \Stueckelberg
(algebraic) gauge parameters for those $\xi^g\fm{1}$ that belong to
the image of $\sigma_-$. Therefore, those $\xi^g\fm{1}$ that are
$\sigma_-$-exact can be gauged away. The leftover gauge symmetry
satisfies $\delta
\xi^g\fm{1}=D_L\xi^g\fm{0}+\sigma_-(\xi^{g+1}\fm{0})=0$, so that
those $\xi^{g+1}\fm{0}$ that belong to the coimage of $\sigma_-$ are
expressed via derivative of $\xi^n\fm{0}$. Having sieved $\WW\fm{0}$
and $\WW\fm{1}$ in this way, only those $\xi^{g}\fm{0}$ are still
independent that are $\sigma_-$-closed and hence belong to
$\coh^0(\sigma_-)$, since only forms of zero-degree are elements of
$\WW\fm{0}$. Then, those $\xi^{g+1}\fm{1}$ that are not
$\sigma_-$-closed can be treated as \Stueckelberg\ gauge parameters
for those $\xi^g\fm{2}$ that belong to the image of $\sigma_-$.
Therefore, only those $\xi^{g}\fm{1}$ are still independent that are
$\sigma_-$-closed but not $\sigma_-$-exact and hence belong to
$\coh^1(\sigma_-)$. Having sieved $\WW\fm{0}$,...,$\WW\fm{p-1}$ one
after another, it turns out that independent differential gauge
parameters at the $k$-th level are given by $\coh^{p-k}(\sigma_-)$.
Analogously, those fields $W^g\fm{p}$ that are $\sigma_-$-exact can
be gauged away by virtue of \Stueckelberg\ gauge parameters at
$\WW\fm{p-1}$. Those fields $W^{g+1}\fm{p}$ that are not
$\sigma_-$-closed can be expressed via derivatives of fields at
lower grade by virtue of the equations $R^g\fm{p+1}\equiv D_L
W^g\fm{p}+\sigma_-(W^{g+1}\fm{p})=0$, these fields are called
\auxiliary. Therefore, the \dynamical\ fields, i.e., those that are
neither \auxiliary\ nor \Stueckelberg, are given by
$\coh^{p}(\sigma_-)$. The nilpotency of $\DD^2\equiv0$ implies
certain relations of the form $D_L
R^g\fm{p+1}+\sigma_-(R^{g+1}\fm{p+1})=0$ between
$R^g\fm{p+1}$. Therefore, \auxiliary\ fields are expressed by virtue
of $\sigma_-$-exact $R^g\fm{p+1}$ and $\sigma_-$-non-closed
$R^{g+1}\fm{p+1}$ are themselves expressed via derivatives of
$R^g\fm{p+1}$. Consequently, the independent equations on
\dynamical\ fields are given by $\coh^{p+1}(\sigma_-)$. From the
cohomological point of view higher degree forms correspond to
certain nontrivial relations between equations called Bianchi identities and
manifest gauge nature of equations. As there are generally more than one
levels of gauge transformations, one can expect the higher
cohomological group to be nontrivial.

 To sum up,

\begin{tabular}{|c|p{10cm}|}\hline
    cohomology group & interpretation \\ \hline\hline
    $\coh^{p-k},\ k=1...p$ &  differential gauge parameters at the $k$-th level of reducibility \\ \hline
    $\coh^p$ & \dynamical\ fields \\ \hline
    $\coh^{p+1}$ & independent gauge invariant equations on \dynamical\ fields \\ \hline
    $\coh^{p+k+1},\ k=1...p$ &  Bianchi identities for the $k$-th order gauge symmetry\\ \hline
    $\coh^{p+k},\ k>p$ & supposed to be trivial in a regular case \\ \hline
\end{tabular}
\vspace{0.2cm}\par To distinguish fields in cohomological sense
the following convention is introduced \cite{Gelfond:2003vh}:
\begin{description}
  \item[Stueckelberg fields(gauge parameters)] are those that
    can be eliminated by pure algebraic
    symmetry, i.e., these fields are $\sigma_-$-exact (are those that can be used to gauge away certain fields, i.e., these
    parameters are not $\sigma_-$-closed).
  \item[Auxiliary fields] are
    those that can be expressed via derivatives of some other fields, i.e., these fields are
    $\sigma_-$-nonclosed. Let us note that this definition is
    generally ambiguous, e.g., in system
    \be\begin{split}&\pl A=B,\\ &\pl B=A,\end{split}\ee
    which field is \auxiliary\ is a matter of choice. This ambiguity is removed by virtue of
    grade\cite{Shaynkman:2000ts, Gelfond:2003vh}.
  \item[Dynamical fields] are those that are neither \auxiliary\ nor \Stueckelberg, i.e., these fields are representatives of
  $\sigma_-$-cohomology classes.
\end{description}
Note that if certain \dynamical\ field given by $\coh^{p}$ belongs
to the grade-$g_f$ subspace $\WW^{g_f}\fm{p}$ and the corresponding
equations given by $\coh^{p+1}$ belong to the grade-$g_e$ subspace
$\WW^{g_e}\fm{p+1}$, the order of equations is equal to $g_e-g_f+1$.

If the unfolded system is supposed to admit a lagrangian
formulation, there has to be a one-to-one correspondence between the
number of \dynamical\ fields and that of equations, $k$-th level
gauge symmetries and $k$-th level Bianchi identities, i.e., in a certain sense
there should be a duality $\coh^{p-k}\sim\coh^{p+1+k}$, $k\in[0,p]$.
This takes place for all unfolded systems exemplified in this paper.

Since the operators involved, i.e., $D_L$ and $\sigma_-$, are of grade
$0$ and $-1$ only, the fields with the grade greater than $g_0$, for
any $g_0\geq0$, form a quotient module, thus describing the same
system on its own - \dual\ formulations. This picture partly breaks
in $\dSAdS$ because of appearance of grade $+1$ nilpotent operator
$\sigma_+$. \be D=D_L+\sigma_-+\lambda^2\sigma_+\ee

The specific character of \framelike\ unfolded systems, which will
be crucial for us, is that by virtue of the inverse background
vielbein $h^\mu_a$ all form indices of any gauge field/parameter
$W^{a_1(s_1),...,a_p(s_m)}\fm{q}$ can be converted to the fiber ones
\be
W^{a_1(s_1),...,a_m(s_m)|[d_1...d_q]}=W^{a_1(s_1),...,a_m(s_m)}_{\mu_1...\mu_q}h^{\mu_1d_1}...h^{\mu_qd_q}.\ee
The resulting tensor $W^{a_1(s_1),...,a_m(s_m)|[d_1...d_q]}$ is not
irreducible and can be decomposed into \lorentz-\irreps\ according
to the \lorentz-tensor product rule \be\label{UnfldConvertToFiber}
W^{a_1(s_1),...,a_m(s_m)|[d_1...d_q]}\Longrightarrow\Y{s_1,...,s_m}\bigotimes_{\lorentz}\Y{(1,q)}=\bigoplus_\alpha\Yy_\alpha,\ee
where one-column Young diagram $\Y{(1,q)}$ of height-$q$  represents
the anti-symmetric indices $d_1,...,d_q$ and $\Yy_\alpha$ are the
\lorentz-\irreps\ the tensor product decomposes into. Having
converted all fields (parameters, equations) of a given unfolded
system to the fiber tensors one can make the identification of the
fields in the unfolded approach and those of the \metric\ approach,
though the general covariance, gauge invariance and other advantages
of the unfolded formalism will be not manifest. While constructing
the unfolded formulations of a known \metric\ systems or
interpreting a given unfolded system in terms of \metric\ fields,
this technique will be widely used. Moreover, the calculation of
$\sigma_-$ cohomology groups is largely based on the explicit
evaluation of (\ref{UnfldConvertToFiber})-like expressions. Once an
unfolded form is known no use of this technique is needed either to
generalize it to other backgrounds or to introduce interactions.

Field $W^{a_1(s_1),...,a_m(s_m)|[d_1...d_q]}$ is traceless in
$a_1,...,a_m$ and $d_1,...,d_q$, separately. However, the cross traces need
not vanish and, hence, some of $\Yy_\alpha$ represent traces. To
distinguish between the traces of different orders let us introduce
the following: the $\Yy$-valued degree-$q$ form $W^{\Yy}\fm{q}$ is
said to be a trace of the $r$-th order iff it has the form \be
W^{\Yy}\fm{q}=\underbrace{h...h}_{r}\underbrace{h...h}_{q-r}
C^{\Yy'}\fm{0}\label{UnfldTraceStructure}\ee for some $\Yy'$-valued
degree-$0$ form $C^{\Yy'}\fm{0}$, where $(q-r)$ indices of the
background vielbeins $h^a$ are contracted with the tensor
representing $\Yy'$ and $r$ indices are free for the whole
expression to take values in $\Yy$, the appropriate Young
symmetrizer is implied. When the indices of $W^{\Yy}\fm{q}$ are
converted to the fiber ones, the expression takes the form
$\underbrace{\eta^{..}...\eta^{..}}_r C^{\Yy'}$, i.e., $C^{\Yy'}$
represents a trace of the $r$-th order.

Different aspects of unfolding are illustrated on the
examples of a massless scalar field, a massless spin-one field and a
massless spin-$s$ and spin-$(s+\frac12)$ fields.

\subsection{Examples of Unfolding}\label{UnfldExamples}
\begin{Example}\label{UnfldScalar} \textbf{Unfolding a scalar field \cite{Vasiliev:1988xc, Shaynkman:2000ts}.}
First, it should be noted that it is easy, of course, to convert the
Klein-Gordon equation $\Box C=0$ to a system of first order
equations \be\begin{split}\label{UnfldScalarIncorrect}&\pl_\mu
C=C_\mu,\\ &\pl^\mu C_\mu=0,\end{split}\ee but when writing the
second equation the explicit form of the metric is to be used, hence,
these equations are not of unfolded form (\ref{UnfldEquations}).

Described in terms of a scalar field $C(x)$, the theory is brought
into a non gauge form, therefore, zero-forms only are allowed, or
else there would be some gauge symmetry according to
(\ref{UnfldGauge}). The most general r.h.s. of $dC=...$ has the form
\be dC=h_a C^a\ee for some vector-valued zero-form $C^a$. This
equation just parameterizes the first derivative of $C(x)$ and is
similar to the first of (\ref{UnfldScalarIncorrect}). There are
three terms allowed on the r.h.s. of $dC^a=...$ \be
dC^a=h_bC^{a,b}+h_bC^{ab}+h^a C'\ee with $C^{a,b}$, $C^{ab}$ and
$C'$ taking values in antisymmetric, symmetric and scalar
\lorentz-\irreps, i.e., \smallpic{\YoungAA}, \smallpic{\YoungB},
$\bullet$. The first term is forbidden by Bianchi identity $h_a
dC^a\equiv0$. To analyze the remaining terms let us decompose $dC^a$
into \lorentz-\irreps\ \be\pl_\mu C^a h^{\mu
b}=\YoungA\otimes\YoungA=\YoungB\oplus\YoungAA\oplus\bullet,\ee
where $\smallpic{\YoungAA}=\pl^{[a}C^{b]}$,
$\smallpic{\YoungB}=\pl^{(a}C^{a)}-\frac1d \eta^{aa}\pl_b C^b$,
$\bullet=\pl_a C^a=\Box C$.

The \smallpic{\YoungAA}\ component represents 
Bianchi identity in the first order reformulation of the Klein-Gordon
equation. Indeed, expressing $C^a$ as $\pl^a C$ implies $\pl^{[a}
C^{b]}\equiv 0$. $\bullet$ represents the desired Klein-Gordon
equation, whereas \smallpic{\YoungB} is the only component allowed
to be nonzero \onmassshell. Therefore, to impose the Klein-Gordon
equation the term $h^a C'$ has to be omitted, \be d C^a = h_b
C^{ab}.\ee The only possible r.h.s. for $dC^{aa}=...$ compatible
with Jacobi identity $h_bdC^{ab}\equiv0$ is of the form
$dC^{aa}=h_bC^{aab}$ for $C^{aaa}$ taking values in a rank-three
symmetric \lorentz-\irrep, i.e., \smallpic{\YoungC}. The process of
unfolding continues unambiguously in this way and results in the
full system \be \label{UnfldScalar} dC^{a(k)}=h_b C^{a(k)b}\equiv
F^{a(k)}(C), \qquad k=0,1,...,\ee where $C^{a(k)}$ is rank-$k$
totally symmetric traceless field, i.e., taking values in
\RectARow{4}{$k$} \lorentz-\irrep.

To make a connection with the general statement of Section
\ref{Results}, for $\Yy=\Y{(0,0)}$ the general scheme results in
\begin{align*}
    \Yy_g&: & & \bullet & & \YoungA & & \YoungB & & \YoungC & & ... \\
    \{n,k\}&: & & \{0,0\} & & \{0,1\} & &\{0,2\} & & \{0,3\} & & ...\\
    g&:&& g=0 & &g=1 & &g=2 & &g=3 & & ...\\
    q_g&:&& q_0=0 & & q_1=0 & & q_2=0 & & q_3=0 & & ...
\end{align*}
therefore, spaces $\WW\fm{q}$ and operator $\sigma_-$ which enters
$\DD=D_L-\sigma_-$ should be defined as \be\begin{split}
    &\WW\fm{q}=\{W\fm{q}, W\fm{q}^a, W\fm{q}^{aa},...\},\\
\end{split}\ee
\be\sigma_-(W\fm{q}^g)=\begin{cases}
    0,& g=0,\\
    h_b W^{a(g-1)b}\fm{q}, & g>0.
\end{cases}\ee
Consequently, $\WW$ is a space of various differential forms with
values in totally symmetric traceless \lorentz-tensors, graded by
the form degree and by the rank of \lorentz-\irreps. Space
$\WW\fm{q}$ forms an irreducible \iso-module. System
(\ref{UnfldScalar}) reads \be \label{UnfldScalarSimple}\DD\omega\fm{0}=0,\qquad
\omega\fm{0}\in \WW\fm{0}.\ee The cohomology groups of $\sigma_-$ are
easy to find:
\begin{description}
    \item[$\coh^0(\sigma_-)=C(x)$,] as the lowest grade field is automatically
    closed, $\sigma_-(C)=0$, and it cannot be exact as a form of zero degree.
    Contrariwise, zero-forms with grade greater than zero are not
    closed, $\sigma_-(C^{k+1})=h_bC^{a(k)b}=0 \Longleftrightarrow
    C^{k+1}=0$;
    \item[$\coh^1(\sigma_-)=h^a B(x)\fm{0}$,] for some $B(x)\fm{0}\in\WW^0\fm{0}$.  One-form
    $B^a\fm{1}=h^a B(x)\fm{0}$ is closed,
    $\sigma_-(B^a\fm{1})=h_aB^a\fm{1}=h_a h^a B(x)\fm{0}\equiv0$, and
    it cannot be represented as $\sigma_-(C^2)=h_bC^{ab}$ as far as
    $C^{ab}$ is traceless.  The cohomology groups at higher grade are trivial since $\sigma_-(h^a B^{a(k)}(x)\fm{0})\neq0$ for $k>0$;
  \item[$\coh^{q>1}(\sigma_-)=\emptyset$,] as one can make sure.
\end{description}
The \interpretation\ in terms of Section \ref{UnfldInterpretation} is
as follows: as is expected, the \dynamical\ field given by $\coh^0$
is $C(x)$. The equations are given by the projection of
$dC^a=h_bC^{ab}$ to $\coh^1$, i.e., $\square C=0$. The fields with
$k>0$ are \auxiliary, being expressed via derivatives of $C(x)$,
$C^{a(k)}=\pl^a...\pl^aC(x)$. Inasmuch as there is no gauge symmetry
in the system, the higher cohomology groups are trivial.

An infinitely many \dual\ formulations of a massless scalar \field\ are also included. Indeed, the fields with $k\geq k_0>0$ form a quotient module $\module_{k_0}$, i.e., one can consistently write
\be dC^k=\sigma_-(C^{k+1}),\quad k\geq k_0\label{UnfldScalarDual}.\ee
The first equation
\be dC^{a(k_0)}=h_bC^{a(k_0)b}\ee
imposes on the \dynamical\ field $C^{a(k_0)}$
\be\begin{split}
    &\pl_b C^{ba(k_0-1)}=0,\\
    &\pl^{[b}C^{c]a(k_0-1)}=0,
\end{split}\ee
which imply, at least locally,
$C^{a(k_0)}=\overbrace{\pl^a...\pl^a}^{k_0}C$ and $\Box C=0$, for
some field $C(x)$. All fields at higher grade are \auxiliary. This
statement is confirmed by the cohomology analysis:
$\coh^0(\module_{k_0},\sigma_-)=C^{k_0}$, i.e., scalar is
described in terms of a traceless rank-$k_0$ symmetric tensor
field subjected to the equations given by
$\coh^1(\module_{k_0},\sigma_-)=h^{a}B^{a(k_0-1)}\fm{0}-\frac{k_0(k_0-1)}{2(d+2k_0-4)}\eta^{aa}h_cB^{a(k_0-2)c}\fm{0}+h_cB^{a(k_0),c}\fm{0}$,
for $B^{a(k_0-1)}\fm{0}$, $B^{a(k_0),c}\fm{0}$ taking values in
\lorentz-\irreps\ characterized by diagrams \RectARow{4}{$k_0-1$},
\RectBRow{5}{1}{$k_0$}{}. These two elements represents just the
components of
$\WW^{k_0}\fm{1}\sim\RectARow{5}{$k_0$}\otimes\YoungA$ with the
symmetry of \RectARow{4}{$k_0-1$} (trace) and
\RectBRow{5}{1}{$k_0$}{}, the third component with the symmetry of
\RectARow{5}{$k_0+1$} is exact.

The above results can be easily extended to  $\dSAdS$ background
and to a massive scalar field \cite{Shaynkman:2000ts}, \be
D_LC^k=\sigma_-(C^{k+1})+\sigma_+(C^{k-1})=0,\ee where there
appears a nontrivial positive grade operator
\be\sigma_+(C^{k-1})=(m^2-\lambda^2 k(k+d-1))\left(h^a
C^{a(k-1)}-\frac{(s-1)}{d+2k-4}\eta^{aa}h_b C^{a(k-1)b}\right)\ee and
$D_L$ is a Lorentz-covariant derivative in $\dSAdS$. In the
\dSAdS\ case the situation is more interesting due to the
appearance of $\sigma_+$ operator and the possibility of the
accidental degeneracy of $\sigma_+$ when $m^2=\lambda^2
k'(k'+d-1)$ for some $k'$ \cite{Shaynkman:2000ts}.
$\blacksquare$\end{Example}

\begin{Example}\label{UnfldMaxwell}
\textbf{Unfolding Maxwell equations \cite{Lopatin:1987hz,
Vasiliev:2005zu}.} An example of massless spin-one
\field\ is more interesting as a gauge system. The
main object is a one-form $A_\mu$. From gauge transformation law
\be\delta A\fm{1}=d\xi\fm{0}\ee it follows that there should not be
other forms of rank greater than zero. If this were the case the
gauge parameters of these new forms would be involved by virtue of
(\ref{UnfldGauge}), making the gauge law for $A\fm{1}$
inappropriate. Indeed, there are two possibilities to introduce
higher degree forms on the r.h.s. of $dA=...$: (i)
$dA\fm{1}=R\fm{2}$ for a scalar valued degree-two form $R\fm{2}$,
which possesses its own gauge parameter $\chi\fm{1}$, $\delta
R\fm{2}=d\chi\fm{1}$. In accordance with (\ref{UnfldGauge}) the
gauge transformation law for $A\fm{1}$ is modified to $\delta
A\fm{1}=d\xi\fm{0}+\chi\fm{1}$ so that $A\fm{1}$ can be gauged away
by virtue of \Stueckelberg\ gauge parameter $\chi\fm{1}$. This
case corresponds to contractible free differential algebras
(\ref{UfldContractible}). (ii) $dA\fm{1}=h_b\omega^b\fm{1}$ for a
vector valued degree-one form $\omega^a\fm{1}$, which possesses its
own gauge parameter $\epsilon^a\fm{0}$. In accordance with
(\ref{UnfldGauge}) the gauge transformation law for $A\fm{1}$ is
modified to $\delta A\fm{1}=d\xi\fm{0}+h_b\epsilon^b\fm{0}$ so that
$A\fm{1}$ can be gauged away by virtue of \Stueckelberg\ gauge
parameter $\epsilon^a\fm{0}$. Consequently, in both cases $A\fm{1}$
can be made non-dynamical, which is not what was expected.
Therefore, the only possibility is to introduce a zero-form
$C^{[ab]}\fm{0}$, anti-symmetric in $a,b$, parameterizing by virtue
of \be dA\fm{1}=h_a h_b C^{[ab]}\fm{0}\ee the Maxwell field
strength. Bianchi identity $h_a h_b dC^{[ab]}\fm{0}\equiv0$
tolerates two terms on the r.h.s of $dC^{[ab]}\fm{0}=...$
\be\label{UnfldMaxwellExt}
dC^{[ab]}\fm{0}=h_cC^{[ab],c}\fm{0}+h^{[a}C^{b]}\fm{0},\ee with
$C^{[ab],c}\fm{0}$ and $C^a\fm{0}$ taking values in {\smallpic\YoungBA} and
{\smallpic\YoungA} \lorentz-\irreps, the former taken in antisymmetric basis.
In order to determine which of the terms should be omitted, if any,
the decomposition of the first derivative of the Maxwell field
strength into Lorentz \irreps\ has to be analyzed \be
C^{[ab]|c}\equiv \pl_\mu C^{[ab]}h^{\mu
c}=\YoungA\otimes\YoungAA=\YoungAAA\oplus\YoungA\oplus\YoungBA,\ee
where ${\smallpic\YoungAAA}=\pl^{[a}C^{bc]}$, ${\smallpic\YoungA}=\pl_b C^{ab}$,  and
${\smallpic\YoungBA}=\pl^{c}C^{[ab]}-\pl^{[c}C^{ab]}+\frac2{d-1}\eta^{c[a}\pl_d
C^{b]d} $. $\pl^{[a}C^{bc]}$ is identically zero provided that
$C^{[ab]}$ is expressed via the first derivative of $A_\mu$, or
represents the second pair of Maxwell equations in terms of
field strength $C^{[ab]}$. In both cases, this component should be
zero and, moreover, this can not be kept nonzero as there is no
$h_cC^{[abc]}$-like term allowed on the r.h.s. of
(\ref{UnfldMaxwellExt}). The second, $\pl_b C^{ab}$, is the desired
Maxwell equations in terms of $A_\mu$, or the first pair of Maxwell
equations in terms of the field strength. Consequently, to impose
Maxwell equations term $h^{[a}C^{b]}$ has to be omitted, whereas
the third is the only component allowed to be nonzero \onmassshell.

Again, unfolding continues unambiguously and requires an infinite
set of fields $C^{[ab],c(k)}\fm{0}$ taking values\footnote{The
fields $C^{[ab],c(k)}$ are taken in antisymmetric basis, i.e., they
are antisymmetric in $a$, $b$, symmetric in $c(k)$ and
$C^{[ab,d]c(k-1)\equiv0}$.} in \RectAYoung{5}{$k$}{\YoungAA} to be
introduced \be dC^{[ab],c(k)}=h_dC^{[ab] ,c(k)d}, \qquad
k=0,1,...\ee The full system has the form
\be\label{UnfldMaxwellEquations}\begin{split}
    &dA\fm{1}=h_a h_b C^{a,b}\fm{0}\equiv F(A,C), \qquad \delta A\fm{1}=d\xi\fm{0},\\
    &dC^{[ab],c(k)}\fm{0}=h_dC^{[ab] ,c(k)d}\fm{0}\equiv F^{[ab],c(k)}(C), \qquad k=0,1,...\end{split}\ee
and represents two modules being glued together by the term on the
r.h.s of the first equation.

For $\Yy=\Y{(1,1)}$, i.e., $N=1$ and $p=1$, the general scheme of Section \ref{Results} results in
\begin{align*}
    \Yy_g&: & & \bullet & & \YoungAA & & \YoungBA & & \YoungCA & & ... \\
    \{n,k\}&: & & \{1,0\} & & \{0,0\} & &\{0,1\} & & \{0,2\} & & ...\\
    g&:&& g=0 & &g=1 & &g=2 & &g=3 & & ...\\
    q_g&:&& q_0=1 & & q_1=0 & & q_2=0 & & q_3=0 & & ...
\end{align*}
therefore, \be \label{UnfldMaxwellModule}\WW\fm{q}=\begin{cases}
    \{W\fm{0}\},& q=0,\\
    \{W\fm{q},W\fm{q-1}^{[ab]}, W\fm{q-1}^{[ab],c},...\}, &q>0,
\end{cases}\ee
let us point out the shortening of $\WW\fm{0}$. The action of
$\sigma_-$ on  $W\fm{q}\in \WW\fm{q}$ is defined as
\be\sigma_-(W^g\fm{q})=\begin{cases}
    0,& g=0,\\
    h_ah_bW^{[ab]}\fm{q-1}, &g=1,\\
    h_d W^{[ab] ,c(g-2)d}\fm{q-1}, & g>1.
\end{cases}\ee
With $\DD=D_L-\sigma_-$, unfolded system
(\ref{UnfldMaxwellEquations}) can be rewritten as \be\label{UnfldMaxwellSimple}
\DD\omega\fm{1}=0,\qquad \delta \omega\fm{1}=\DD\xi\fm{0},\ee where
$\omega\fm{1}\in\WW\fm{1}$ and $\xi\fm{0}\in\WW\fm{0}$. The fields
at $\{n=1,k=0\}\sim g=0$ form a finite-dimensional \iso-module,
whereas the fields with $\{n=0,k=0,1,...\}\sim g>0$ form an
infinite-dimensional \iso-module.

Again, the cohomology groups are easy to find: gauge parameters are
given by $\coh^0(\sigma_-)=\xi\fm{0}$, \dynamical\ fields by
$\coh^1(\sigma_-)=A\fm{1}$, Maxwell equations by
$\coh^2(\sigma_-)=h^{[a}B^{b]}\fm{0}$, and the Bianchi identities
for the first order gauge transformations by
$\coh^3(\sigma_-)=h^ah^bB\fm{0}$, the rest of groups are empty
$\coh^{k>3}(\sigma_-)=\emptyset$. Important is the direct
correspondence between the number of fields in $\coh^1$ and the
equations in $\coh^2$, the gauge parameters in $\coh^0$ and the
Bianchi identities in $\coh^3$.

There are infinitely many of non-gauge \dual\ formulations based on
$C^{[ab] ,c(k_0)}$, e.g., for $k_0=0$ one recovers the Maxwell
equations in terms of the field strength \be\begin{split}
    &\pl^{[a}C^{b,c]}=0, \\ &
    \pl_b C^{a,b}=0.
\end{split}\ee
$\blacksquare$\end{Example}

Let us note that the ambiguity at the first step of unfolding (the
terms $h^aC'$ and $h^{[a}C^{b]}$ in the above two examples)
expresses the possibility of introducing a mass term.

\begin{Example}\label{UnfldSymmetric}
\textbf{Unfolding a totally symmetric massless spin-$s$ field
\cite{Vasiliev:1988xc, Lopatin:1987hz, Vasiliev:2001wa}.} We start
from the \metric\ description of a massless spin-$s$
\field\ in terms of a double-traceless symmetric rank-$s$ field
$\phi_{\mu_1...\mu_s}$ satisfying (\ref{FlatMSFronsdal}), the second
order equations are invariant with respect to the first order gauge
transformation with traceless rank-$(s-1)$ gauge parameter
$\xi_{\mu_1...\mu_{s-1}}$. Inasmuch as the Young diagram
\RectARow{5}{$s$} of \mls\ is of the height one, there are gauge
transformations of the first level only. Therefore, the dynamical field
in the unfolded approach has to be a one-form. There is only one way
to identify $\phi_{\mu_1...\mu_s}$ and $\xi_{\mu_1...\mu_{s-1}}$
with certain $\omega\fm{1}^{\Yy_0}$ and $\xi\fm{0}^{\Yy_0}$ taking
values in the same \lorentz-\irrep\ $\Yy_0$, namely,
$\Yy_0=\RectARow{5}{$s-1$}$, i.e., $\omega^{a(s-1)}\fm{1}$ and
$\xi^{a(s-1)}\fm{0}$ takes values in a rank-$(s-1)$ symmetric
traceless tensors. Field $\phi_{\mu_1...\mu_s}$ is identified
with a totally symmetric part of $\omega^{a(s-1)}\fm{1}$, i.e.,
$\phi_{\mu_1...\mu_s}=h^{b_1}_{(\mu_1}...h^{b_{s-1}}_{\mu_{s-1}}\omega^{a_1...a_{s-1}}_{\mu_s)}\eta_{a_1b_1}...\eta_{a_{s-1}b_{s-1}}$,
the double-tracelessness condition is the consequence of the
tracelessness of $\omega^{a(s-1)}\fm{1}$ in $a_1...a_{s-1}$.
$\xi_{\mu_1...\mu_{s-1}}$ is identified with $\xi^{a(s-1)}\fm{0}$
directly
$\xi_{\mu_1...\mu_{s-1}}=h^{b_1}_{\mu_1}...h^{b_{s-1}}_{\mu_{s-1}}\xi^{a_1...a_{s-1}}\eta_{a_1b_1}...\eta_{a_{s-1}b_{s-1}}$
and one can make sure that the Fronsdal's gauge transformation law
is recovered from $\delta\omega^{a(s-1)}\fm{1}=d\xi^{a(s-1)}\fm{0}$.
Field $\omega^{a(s-1)}\fm{1}$ has a \redundant\ component with
the symmetry of \RectBRow{5}{1}{$s-1$}{}, which can be made
\Stueckelberg\ by virtue of gauge parameter $\xi^{a(s-1),b}$ with the same
symmetry type, i.e., \be \delta
\omega^{a(s-1)}\fm{1}=d\xi^{a(s-1)}\fm{0}+h_c\xi^{a(s-1),c}\fm{0}.\ee
It automatically follows that there is a degree-one gauge field
$\omega^{a(s-1),b}\fm{1}$ coupled with $\omega^{a(s-1)}\fm{1}$ as
\be d \omega^{a(s-1)}\fm{1}=h_c\omega^{a(s-1),c}\fm{1}.\ee The only
solution to Bianchi identity
$h_cd\omega^{a(s-1),c}\fm{1}\equiv0$ is\footnote{The proof is not
given as more general statement about the solutions of such
equations is proved in the next Section.} \be d
\omega^{a(s-1),b}\fm{1}=h_c\omega^{a(s-1),bc}\fm{1},\ee where a new
field with the symmetry of \RectBYoung{5}{$s-1$}{\YoungB} is
introduced, and so on \be d
\omega^{a(s-1),b(k)}\fm{1}=h_c\omega^{a(s-1),b(k)c}\fm{1}\ee until
the field with the symmetry of \RectBRow{5}{5}{$s-1$}{}, for which
the Bianchi identity solves as \be d
\omega^{a(s-1),b(s-1)}\fm{1}=h_ch_dC^{a(s-1)c,b(s-1)d}\fm{0},\ee
where field $C^{a(s),b(s)}\fm{0}$ represents a generalized Weyl
tensor and has the symmetry of \RectBRow{5}{5}{$s$}{}. The solution
of Bianchi identity $h_ch_d dC^{a(s-1)c,b(s-1)d}\fm{0}\equiv0$
has the form \be
\label{UfldHSWeylFirst}dC^{a(s),b(s)}\fm{0}=h_c\left(C^{a(s)c,b(s)}\fm{0}+{\frac{s}2}C^{a(s)b,b(s-1)c}\fm{0}\right),\ee
where $C^{a(s+1),b(s)}\fm{0}$ has the symmetry of
\RectBRow{6}{5}{$s+1$}{$s$}. The second term on the r.h.s. of
(\ref{UfldHSWeylFirst}) supplements the first one to have a proper
Young symmetry. Further unfolding requires a set of degree-zero forms
$C^{a(s+i),b(s)}\fm{0}$, taking values in \lorentz-\irreps\
characterized by Young diagrams \RectBRow{7}{5}{$s+i$}{$s$}. The
full system has the form \be\begin{split}\label{UnfldHSFullsystem}
    &d \omega^{a(s-1),b(k)}\fm{1}=h_c\omega^{a(s-1),b(k)c}\fm{1},\\
    &\delta \omega^{a(s-1),b(k)}\fm{1}=d\xi^{a(s-1),b(k)}\fm{0}+h_c\xi^{a(s-1),b(k)c}\fm{0},\qquad k\in[0,s-2],\\
    &d \omega^{a(s-1),b(s-1)}\fm{1}=h_ch_dC^{a(s-1)c,b(s-1)d}\fm{0}, \qquad \delta\omega^{a(s-1),b(s-1)}\fm{1}=d\xi^{a(s-1),b(s-1)}\fm{0}\\
    &dC^{a(s+i),b(s)}=h_c\left(C^{a(s+i)c,b(s)}+\frac{s}{i+2}C^{a(s+i)b,b(s-1)c}\right),\qquad i\in[0,\infty).\end{split}\ee
By the construction the system incorporates Fronsdal's field
$\phi_{\mu_1...\mu_s}$ with the correct gauge law, but there is
still to be proved that the Fronsdal's equations are really imposed
and that there are no other \dynamical\ fields in the system. Let us
reconstruct Fronsdal's equations (\ref{FlatMSFronsdal}), whereas
the second statement will be proved as a part of a more general
theorem. The first two equations of (\ref{UnfldHSFullsystem}) read
\begin{align}
\pl_{\mu}\omega^{a(s-1)}_\nu-\pl_{\nu}\omega^{a(s-1)}_\mu&=h_{c\mu}\omega^{a(s-1),c}_\nu-h_{c\nu}\omega^{a(s-1),c}_\mu,\\
\pl_\mu\omega^{a(s-1),b}_\nu-\pl_\nu\omega^{a(s-1),b}_\mu&=h_{c\mu}\omega^{a(s-1),bc}_\nu-h_{c\nu}\omega^{a(s-1),bc}_\mu
\end{align}
It is convenient to convert all world indices to the fiber ones
\begin{align}
\label{UnfldHSEqA}\pl^c\omega^{a(s-1)|d}-\pl^d\omega^{a(s-1)|c}&=\omega^{a(s-1),c|d}-\omega^{a(s-1),d|c},\\
\label{UnfldHSEqB}\pl^c\omega^{a(s-1),b|d}-\pl^d\omega^{a(s-1),b|c}&=\omega^{a(s-1),bc|d}-\omega^{a(s-1),bd|c},
\end{align}
where $\omega^{a(s-1)|b}\equiv\omega^{a(s-1)}_\mu h^{b\mu}$,
$\omega^{a(s-1),b|c}\equiv\omega^{a(s-1),b}_\mu h^{c\mu}$,
$\omega^{a(s-1),bb|c}\equiv\omega^{a(s-1),bb}_\mu h^{c\mu}$.
Contracting (\ref{UnfldHSEqB}) with $\eta_{bd}$ and, then,
symmetrizing $c$ with $a_1...a_{s-1}$ results in
\be\label{UnfldHSEqC}\pl_c\omega^{a(s-1),c|a}-\pl^a\omega^{a(s-1),c|}_{\phantom{a(s-1),c|}c}=0.\ee
By symmetrizing in (\ref{UnfldHSEqA}) $a_1...a_{s-1}$ with $d$ \be
\omega^{a(s-1),c|a}=\pl^c\omega^{a(s-1)|a}-\pl^a\omega^{a(s-1)|c}\ee
and, then, by contracting (\ref{UnfldHSEqA}) with $\eta_{da_{s-1}}$
\be\omega^{a(s-1),c|}_{\phantom{a(s-1),c|}c}=(s-1)\left(\pl_c\omega^{a(s-2)c|a}-\pl^a\omega^{a(s-2)c|}_{\phantom{a(s-2)c|}c}\right)\ee
all the terms of (\ref{UnfldHSEqC}) can be expressed via
$\omega^{a(s-1)|b}$. Plugging these to (\ref{UnfldHSEqC}) gives
\be\square
\omega^{a(s-1)|a}-\pl^a\pl_c\left((s-1)\omega^{a(s-2)c|a}+\omega^{a(s-1)|c}\right)+(s-1)\pl^a\pl^a\omega^{a(s-2)c|}_{\phantom{a(s-2)c|}c}=0,\ee
where $\omega^{a(s-1)|a}$ is to be identified with Fronsdal's
field $\phi^{a(s)}$ as $\omega^{a(s-1)|a}=\frac1s\phi^{a(s)}$, then,
$\omega^{a(s-2)c|}_{\phantom{a(s-2)c|}c}=\frac12\phi^{a(s-2)c}_{\phantom{a(s-2)c}c}$,
$(s-1)\omega^{a(s-2)c|a}+\omega^{a(s-1)|c}=\phi^{a(s-1)c}$ and the
equation acquires the Fronsdal's form (\ref{FlatMSFronsdal})
\be\label{UnfldFronsdal}\square\phi^{a(s)}-s\pl^a\pl_c\phi^{a(s-1)c}+{\scriptstyle
\frac{s(s-1)}2}\pl^a\pl^a\phi^{a(s-2)c}_{\phantom{a(s-2)c}c}=0.\ee
For $\Yy=\Y{(s,1)}$, i.e., $N=1$ and $p=1$, the general scheme of
Section \ref{Results} results in
\begin{align*}
    \Yy_g&: & & \RectARow{4}{$s-1$} & & ...&&   \RectBRow{4}{4}{$s-1$}{} & & \RectBRow{5}{5}{$s$}{} & & \RectBRow{6}{5}{$s+1$}{$s$} & & ... \\
    \{n,k\}&: & & \{1,0\} & &...&&         \{1,s-1\} & &\{0,0\} & & \{0,1\} & & ...\\
    g&:&& g=0 & &...&& g=s-1 & &g=s & &g=s+1 & & ...\\
    q_g&:&& q_0=1 & &...&&  q_{s-1}=1 & & q_s=0 & & q_{s+1}=0 & & ...,
\end{align*}
therefore,
\be\WW\fm{q}=\begin{cases}\label{UnfldHSModule}
    \{W^{a(s-1)}\fm{0},W^{a(s-1),b}\fm{0},...,W^{a(s-1),b(s-1)}\fm{0}\}, & q=0,\\
    \{W^{a(s-1)}\fm{q},W^{a(s-1),b}\fm{q},...,W^{a(s-1),b(s-1)}\fm{q},W^{a(s),b(s)}\fm{q-1},W^{a(s+1),b(s)}\fm{q-1},...\}, & q>0,
\end{cases}\ee
and \be\sigma_-(W^g\fm{q})=\begin{cases}
    0, & g=0,\\
    h_cW^{a(s-1),b(g-1)c}\fm{q}, &g\in[1,s-1],\\
    h_ch_dW^{a(s-1)c,b(s-1)d}\fm{q-1}, & g=s,\\
    h_c\left(W^{a(g-s-1)c,b(s)}\fm{q-1}+\frac{s}{g-s+1}W^{a(g-s-1)b,b(s-1)c}\fm{q-1}\right), & g>s.
\end{cases}\ee
With $\DD=D_L-\sigma_-$ full system (\ref{UnfldHSFullsystem})
can be rewritten as \be\label{UnfldHSSimple} \DD\omega\fm{1}=0,\qquad \delta
\omega\fm{1}=\DD\xi\fm{0},\ee where $\omega\fm{1}\in\WW\fm{1}$ and
$\xi\fm{0}\in\WW\fm{0}$. The fields at $\{n=1,k=0,...,s-1\}\sim
g\in[0,s-1]$ form a finite-dimensional \iso-module, whereas the
fields with $\{n=0,k=0,1,...\}\sim g\geq s$ form an
infinite-dimensional \iso-module.

The representatives of cohomology classes can be chosen as $\coh^0=\xi^{a(s-1)}\fm{0}$,
$\coh^1=h_b\phi^{a(s-1)b}$, $\coh^2= h^b h_c G^{a(s-1)c}-h^{a} h_c
G^{a(s-2)bc}+\gamma\eta^{aa}h^a h_c G^{a(s-4)bcn}_{\phantom{a(s-4)bcn}n}+(\beta\eta^{ab}h^ah_c +\alpha\eta^{aa}h^bh_c) G^{a(s-3)cn}_{\phantom{a(s-3)cn}n}$\footnote{$\alpha=-\frac{(s(d+s-5)-d+6)(s-2)}{2(d+s-4)(d+2s-6)}$, $\beta=\frac{(s-2)}{(d+s-4)}$, $\gamma=\frac{(d+s-6)(s-2)(s-3)}{2(d+s-4)(d+2s-6)}$.}, $\coh^3=h^bh^{a}h_c\chi^{a(s-2)c}$,
$\coh^{k>3}=\emptyset$, where $\phi^{a(s)}$ and $G^{a(s)}$ are
double-traceless tensor fields and $\xi^{a(s-1)}$ and
$\chi^{a(s-1)}$ are traceless. $\coh^{0,1}$ are nontrivial at the
lowest grade, whereas $\coh^{2,3}$ are nontrivial at grade-one,
therefore, the order of equations that are imposed on the
\dynamical\ field is equal to two.  It turns out that there is no need in the explicit form for $\alpha$, $\beta$ and $\gamma$, it is sufficient to find out the Young symmetry of representatives only. The uniqueness of Fronsdal's
theory at the level of action was demonstrated in \cite{Curtright:1979uz} and at the
level of equations in \cite{Francia:2002aa, Skvortsov:2007kz}. Consequently, there is
no need for equations (\ref{UnfldFronsdal}) to be found explicitly -
those of Fronsdal are the only possible. Again, there is a one-to-one
correspondence between the number of fields in $\coh^1$ and the
equations in $\coh^2$, the gauge parameters in $\coh^0$ and the
Bianchi identities in $\coh^3$. The system contains an infinitely
many of \dual\ formulations, those that have field
$W^{a(s-1),b(k_0)}\fm{1}$ at the lowest grade are gauge \dual\
descriptions, whereas those that have  field $C^{a(s+k_0),b(s)}$
at the lowest grade are non-gauge. \Dual\ description based on
field $W^{a(s-1),b(1)}\fm{1}$ was elaborated in
\cite{Matveev:2004ac}. $\blacksquare$\end{Example}

\begin{Example}\label{UnfldSymmetricFermionic}
\textbf{Unfolding a totally symmetric massless spin-$(s+\frac12)$  field\cite{Vasiliev:1980as, Vasiliev:1987tk, Vasiliev:1988xc}}.
Unfolding a totally symmetric  massless spin-$(s+\frac12)$ field results in literally the same unfolded system (\ref{UnfldHSFullsystem}, \ref{UnfldHSSimple}) but with fields taking values in
\lorentz-\irreps\ that are irreducible spin-tensors with the same tensor part as in the bosonic case. Note that operator $\sigma_-$ is not modified but the $\sigma_-$ cohomology groups are slightly changed such that the equations become of the first order.

\textbf{The Dirac equation.} Unfolded system (\ref{UnfldScalar}) or, equivalently, (\ref{UnfldScalarSimple}) describes a massless spin-$\frac12$ field provided that $C^k$ take values in $\RectARow{5}{s}_\ferm$ \lorentz-\irreps, i.e., $C^k\equiv C^{\alpha\fsp a(k)}$ and ${\Gamma^\alpha_b}_\beta C^{\beta\fsp ba(k-1)}=0$. Contracting the first equation of (\ref{UnfldScalar})
\be\pl_\mu C^{\alpha}=h_{\mu a}C^{\alpha\fsp a}\ee
with $h^{b\mu}{\Gamma^\beta_b}_\alpha$ gives the Dirac equation  \be{\Gamma^\alpha_a}_\beta \pl^a C^\beta=0.\ee
the rest of equations express auxiliary fields in terms of derivatives of \dynamical\ field $C^{\alpha}$.

\textbf{The Rarita-Schwinger equation.} Unfolded system (\ref{UnfldMaxwellEquations}) or, equivalently, (\ref{UnfldMaxwellSimple}) describes a massless spin-$\frac32$ field provided that irreducible tensors in (\ref{UnfldMaxwellModule}) are replaced by the corresponding irreducible tensor-spinors, i.e., $A\fm{1}\equiv A^{\alpha}_\mu$, $C^k\equiv C^{\alpha\fsp [ab],c(k)}$ and ${\Gamma^\alpha_b}_\beta C^{\beta\fsp [ab],c(k)}=0$, ${\Gamma^\alpha_d}_\beta C^{\beta\fsp [ab],dc(k-1)}=0$. Contracting the first equation of (\ref{UnfldMaxwellEquations})
\be \frac12(\pl_\mu A^\alpha_\nu-\pl_\nu A^\alpha_\mu)=h_{\mu a}h_{\nu b}C^{\alpha\fsp [ab]}\ee
with $\Gamma$-matrices and converting world indices $\mu$, $\nu$ to the fiber ones according to
$A^{\alpha\fsp a}\equiv A^{\alpha}_\mu h^{a\mu}$ gives the Rarita-Schwinger equation
\be \label{UnfldRarita}{\Gamma^\alpha_b}_\beta\pl^b A^{\beta\fsp a}-\pl^a {\Gamma^\alpha_b}_\beta A^{\beta\fsp b}=0,\qquad \delta A^{\alpha\fsp a}=\pl^a\xi^{\alpha}.\ee
The rest of unfolded equations express auxiliary fields in terms of derivatives of \dynamical\ field $A^{\alpha}_\mu$.

\textbf{The Fang-Fronsdal equation for a spin-$(s+\frac12)$ massless field.} Unfolded system (\ref{UnfldHSFullsystem}) or, equivalently, (\ref{UnfldHSSimple}) describes a massless spin-$(s+\frac12)$ field provided that irreducible tensors in (\ref{UnfldHSModule}) are replaced by the corresponding irreducible tensor-spinors. With all world indices converted to the fiber ones the first equation of (\ref{UnfldHSFullsystem}) has the form
\be \label{UnfldFermionFirst}G^{\alpha\fsp a(s-1)|[cd]}\equiv\pl^c\omega^{\beta\fsp a(s-1)|d}-\pl^d\omega^{\beta\fsp a(s-1)|c}=\omega^{\beta\fsp a(s-1),c|d}-\omega^{\beta\fsp a(s-1),d|c},\ee
where $\omega^{\alpha\fsp a(s-1)|b}\equiv\omega^{\alpha\fsp a(s-1)}_\mu h^{b\mu}$, $\omega^{\alpha\fsp a(s-1),b|c}\equiv\omega^{\alpha\fsp a(s-1),b}_\mu h^{c\mu}$.

Analogously to the bosonic case, nonsymmetric component of $\omega^{\beta\fsp a(s-1)|d}$ can be gauged away algebraically. Hence, the Fronsdal's field  $\phi^{\alpha\fsp a(s)}$ is identified as  $\omega^{\alpha\fsp a(s-1)|a}=\frac1s\phi^{\alpha\fsp a(s)}$. Then, ${\Gamma^\alpha_b}_\beta\phi^{\beta\fsp a(s)}={\Gamma^\alpha_b}_\beta\omega^{\beta\fsp a(s-1)|b}$, ${\Gamma^\alpha_b}_\beta{\Gamma^\beta_c}_\gamma\phi^{\gamma\fsp a(s-2)bc}=\eta_{bc}\omega^{\alpha\fsp a(s-2)b|c}$ and the Fronsdal's triple $\Gamma$-traceless constraint is a consequence of irreducibility of $\omega^{\alpha\fsp a(s-1)|b}$ in $\alpha$, $a_1...a_{s-1}$ and the lack of any conditions with respect to $a_1...a_{s-1}$ and $b$. Since the Fronsdal's field contains three irreducible components with the symmetry of $\RectARow{5}{s}_\ferm$, $\RectARow{5}{s-1}_\ferm$ and $\RectARow{4}{s-2}_\ferm$, the equations of motion has to be nontrivial in these three sectors too. The projector on the \dynamical\ equations  is simply ${\Gamma^\alpha_b}_\beta G^{\beta\fsp a(s-1),b|a}$ and gives
\be\label{UnfldFermionicA}{\Gamma^\alpha_b}_\beta\pl^b\phi^{\beta\fsp a(s)}-s\pl^a{\Gamma^\alpha_b}_\beta\phi^{\beta\fsp a(s-1)b}=0,\ee which is in accordance with (\ref{FlatMSFronsdalFermionic}). But we would like to note that
there are two components with the symmetry of $\RectARow{5}{s-1}_\ferm$, namely, $\eta_{bc}G^{\alpha\fsp a(s-2)b|ac}$ and ${\Gamma^\alpha_b}_\beta{\Gamma^\beta_c}_\gamma G^{\gamma\fsp a(s-1)|[bc]}$. The correct projector on this component of the equations is given by the combination
\be 2(s-1)\eta_{bc}G^{\alpha\fsp a(s-2)b|ac}-{\gamma^\alpha_b}_\beta{\gamma^\beta_c}_\gamma G^{\gamma\fsp a(s-1)|[bc]}\label{UnfldFermionsCombination},\ee which is identically zero when $G^{\alpha\fsp a(s-1)|[bc]}$ is expressed via $\omega^{\beta\fsp a(s-1),c|d}$ by virtue of (\ref{UnfldFermionFirst}). Therefore, it is (\ref{UnfldFermionsCombination}) that does not express certain part of $\omega^{\beta\fsp a(s-1),c|d}$ in terms of first derivatives of the dynamical field $\phi^{\alpha\fsp a(s)}$ and, thus, this is a dynamical equation. Obviously, (\ref{UnfldFermionsCombination}) is a representative of $\sigma_-$ cohomology group $\coh^2_{g=0}$. In terms of $\phi^{\alpha\fsp a(s)}$ the representative has the form
\be\label{UnfldFermionicB} {\Gamma^\alpha_b}_\beta\pl^b{\Gamma^\beta_c}_\gamma\phi^{\gamma\fsp a(s-1)c}-\pl_c\phi^{\alpha\fsp a(s-1)c}+\frac{(s-1)}2\pl^a\phi^{\alpha\fsp a(s-2)c}_{\phantom{\alpha\fsp a(s-2)c}c}=0,\ee
which is equal to the $\Gamma$-trace of (\ref{UnfldFermionicA}).
The gauge transformations has the form
\be\delta\phi^{\alpha\fsp a(s)}=\pl^{a}\xi^{\alpha\fsp a(s-1)},\ee
where ${\Gamma^\alpha_b}_\beta\xi^{\beta\fsp a(s-2)b}=0$.

Consequently, unfolded equations for totally symmetric bosonic fields are proved to completely determine the unfolded equations for totally symmetric fermionic fields. Though $\sigma_-$ is not modified in the fermionic case, $\sigma_-$ cohomology groups are changed, since the fermionic \dynamical\ equations are of the first order. $\blacksquare$\end{Example}
\section{Unfolding Mixed-Symmetry Fields}\setcounter{equation}{0}\label{MSUnfld}
First, in Section \ref{MSSimplestExample} an example of the
simplest massless mixed-symmetry field of spin-$\smallpic{\YoungBA}$
is investigated in detail and leading arguments are given for a
spin-$\smallpic{\YoungCB}$ field. Second, the general statement on arbitrary
mixed-symmetry fields is proved, whereas technical details are
collected in Sections \ref{MSJacobi} and \ref{MSCohomology}.
The analysis of the number of physical degrees of freedom is carried out in Section \ref{MSPDOFCounting}.

\subsection{Simplest Mixed-Symmetry Fields.}
\label{MSSimplestExample}
\begin{Example}\textbf{Spin-\smallpic{\YoungBA} field.}
In the \metric\ approach a spin-\smallpic{\YoungBA} \field\ is
minimally described\cite{Curtright:1980yk, Aulakh:1986cb} by the
subjected to (\ref{FlatMSHookA}) field $\phi_{[\mu\mu],\nu}$ that
takes values in a reducible \lorentz-representation
$\smallpic{\YoungBA}\oplus\smallpic{\YoungA}$, the latter component
is identified with the trace
$\phi_{\mu\nu,}^{\phantom{\mu\nu,}\nu}$. The gauge transformations
(\ref{FlatMSHookB}) has two levels of reducibility, with two gauge
parameters $\xi^S_{(\mu\nu)}$, $\xi^A_{[\mu\nu]}$ taking values in
$\smallpic{\YoungB}\oplus\bullet$, \smallpic{\YoungAA}
\lorentz-\irreps\ at the first level and one parameter $\xi_\mu$
taking values in \smallpic{\YoungA} \lorentz-\irrep\ at the second
level (\ref{FlatMSHookC}).

First, the \metric\ field $\phi_{[\mu\mu],\nu}$ has to be
incorporated into certain differential form $e^\Yy\fm{q}$, which is
called a \physical\ vielbein, the \lorentz-\irrep\ $\Yy$ and the
degree $q$  to be defined below. To make the symmetries both of the
first and of the second level of reducibility manifest the degree of
the \physical\ vielbein has to be two, at least.

Moreover, were the \physical\ vielbein taken to be a degree-one form, not all of the gauge symmetries even at the first level
would be manifest. Indeed, there are two possibilities for
\physical\ vielbein $e^\Yy\fm{1}$ to contain a component with the
symmetry of \smallpic{\YoungBA}. Namely, $\Yy=\smallpic{\YoungB}$
and $\Yy=\smallpic{\YoungAA}$ since
$\smallpic{\YoungB}\otimes\smallpic{\YoungA}=\smallpic{\YoungBA}\oplus\smallpic{\YoungC}\oplus\smallpic{\YoungA}$
and
$\smallpic{\YoungAA}\otimes\smallpic{\YoungA}=\smallpic{\YoungBA}\oplus\smallpic{\YoungAAA}\oplus\smallpic{\YoungA}$.
The associated differential gauge parameter is
$\xi^{\smallpic{\YoungB}}\fm{0}$, in the first case, and
$\xi^\smallpic{\YoungAA}\fm{0}$, in the second case. Inasmuch as the
gauge parameters are the forms of degree-zero only one of the
required two parameters is present in each of the cases.

The degree of the \physical\ vielbein has to be not greater than two
if the component associated with the \metric\ field
$\phi_{\mu\mu,\nu}$ is required not to be a certain trace part. The
reason why a \dynamical\ field should not be a certain trace part is
explained in the next Section.

Therefore, the \physical\ vielbein has to be of degree-two and has
to take values in \lorentz-\irrep\ that belongs to the vector
representation $\smallpic{\YoungA}$, i.e.,
$e^\smallpic{\YoungA}\fm{2}\equiv e^a\fm{2}\equiv e^a_{\mu\nu}$. The
associated gauge parameter
$\xi^\smallpic{\YoungA}\fm{1}\equiv\xi^a\fm{1}\equiv \xi^a_{\mu}$
contains in its Lorentz decomposition
$\smallpic{\YoungA}\otimes\smallpic{\YoungA}=\smallpic{\YoungAA}\oplus\smallpic{\YoungB}\oplus\bullet$
both anti-symmetric $\xi^A_{\mu\nu}=\xi^{a}_{[\mu}
h^b_{\nu]}\eta_{ab}$ and symmetric $\xi^S_{\mu\nu}=\xi^{a}_{(\mu}
h^b_{\nu)}\eta_{ab}$ gauge parameters, the latter enters along with
the trace $\xi^{a}_\mu h^{\mu}_a$. There exists a second level gauge
parameter $\chi^a\fm{0}$, which is directly identified with
$\xi_\mu$ as $\xi_\mu=h^b_\mu\chi^a\eta_{ab}$. There are no
\redundant\ components in the gauge parameter. However, \physical\ vielbein $e^\smallpic{\YoungA}\fm{2}$ contains one
\redundant\ component in its decomposition into Lorentz \irreps\
$\smallpic{\YoungA}\otimes\smallpic{\YoungAA}=\smallpic{\YoungBA}\oplus\smallpic{\YoungA}\oplus\smallpic{\YoungAAA}$.
The first and the second components are to be associated with a traceful
$\phi_{\mu\mu,\nu}$, whereas the third one, totally anti-symmetric
component $e^{[a|bc]}$ of $e^{a|bc}=e^a_{\mu\nu}h^{b\mu}h^{c\nu}$,
is a \redundant\ one and has to be made non-dynamical. In order to do so,
algebraic gauge parameter $\xi^{[abc]}\fm{0}$ with the symmetry
of \smallpic{\YoungAAA} is introduced. It is obvious from pure
algebraic gauge law \be \delta e^a\fm{2}=h_b h_c\xi^{[abc]}\fm{0}\ee
that the \redundant\ component can be fixed to zero and $e^a\fm{2}$
can be directly identified with $\phi_{\mu\mu,\nu}$ as
$e^a_{\mu\mu}=\phi_{\mu\mu,\nu}h^{a\nu}$. From gauge
transformation laws $\delta e^a\fm{2}=d\xi^a\fm{1}$, $\delta
\xi^a\fm{1}=d\chi^a\fm{0}$, i.e., $\delta
e^a_{\mu\nu}=\frac12\left(\pl_\mu\xi^a_\nu-\pl_\nu\xi^a_\mu\right)$,
$\delta \xi^a_\mu=\pl_\mu\chi^a$, required gauge transformations
(\ref{FlatMSHookB}), (\ref{FlatMSHookC}) for \metric\ field
$\phi_{\mu\mu,\nu}$ and for the first order gauge parameters
$\xi^S_{(\mu\nu)}$, $\xi^A_{[\mu\nu]}$ are easily recovered.

Since in the unfolded approach each gauge parameter possesses its own
gauge field and vice-verse, associated with $\xi^{[abc]}\fm{0}$
gauge field $\omega^{[abc]}\fm{1}$ enters as \be de^a\fm{2}=h_b
h_c\omega^{[abc]}\fm{1},\ee which determines the first equation.
Applying $d$ to this equation results in Bianchi identity $h_b
h_cd\omega^{[abc]}\fm{1}\equiv0$, which has a unique solution of the
form \be d \omega^{[abc]}\fm{1}=h_dh_fC^{[abc],[df]}\fm{0},\ee with
$C^{[abc],[df]}\fm{0}$ having the symmetry \smallpic{\YoungBBA} of
the generalized Weyl tensor of a spin-\smallpic{\YoungBA} \field, i.e., it is
anti-symmetric in $[abc]$, $[df]$ and  $C^{[abc,d]f}\fm{0}\equiv0$.
To clarify the form of the solution let $C^{[abc]|[df]}\fm{0}$ be a
degree-zero form anti-symmetric in $[abc]$, $[df]$ and with no
definite symmetry between these two groups of indices. The Bianchi
identity implies $h_bh_ch_dh_fC^{[abc]|[df]}\fm{0}\equiv0$, which is
equivalent to $C^{a[bc|df]}\fm{0}\equiv0$ since vielbeins
$h^a$ anticommute. Therefore, the solution is
parameterized by those components of
$C^{[abc]|[df]}\fm{0}\sim\smallpic{\YoungAAA}\otimes\smallpic{\YoungAA}$
that has the symmetry of Young diagrams with no more than three rows,
since the requirement for total anti-symmetrization of any four
indices to give zero is the characteristic property of Young
diagrams with at most three rows.
$\smallpic{\YoungAAA}\otimes\smallpic{\YoungAA}\ni\smallpic{\YoungBBA}$
is the only component of zero trace order\footnote{Although, there
are trace components with the number of rows less than four, e.g.,
\smallpic{\YoungAAA}, their tensor product by the metric $\eta^{ab}$
contains components with the symmetry of the \sld-Young diagrams
with more than three rows. Nontrivial trace with the symmetry of \smallpic{\YoungBA} corresponds to the mass-like term. Hence, traceful tensors either do not satisfy the
Bianchi identities or introduce mass-like terms.} with three rows and there are no components
with the number of rows less than three. Alternatively, one could
search for the solution in the sector of degree-one forms as $d
\omega^{[abc]}\fm{1}=h_dC^{[abc]|d}\fm{1}$, with some
$C^{[abc]|d}\fm{1}$, but Bianchi identity
$h_bh_ch_dC^{[abc]|d}\fm{1}\equiv0$ implies that the only component of
$C^{[abc]|d}\fm{1}\sim\smallpic{\YoungAAA}\otimes\smallpic{\YoungA}$
that is allowed must have less than
three rows, which is impossible since $[abc]$ makes Young diagrams of
$\smallpic{\YoungAAA}\otimes\smallpic{\YoungA}$ consist of three
rows, at least\footnote{Again, the trace component \smallpic{\YoungAA} enters the tensor product as $\eta\otimes\smallpic{\YoungAA}$ and has the number of rows not less than three.}. Thus, one has to search for the solution among the forms
of less degree.

Further unfolding requires a set of fields
$C^{[abc],[df],g(k)}$ with the symmetry of
\smallpic{\RectAYoung{5}{$k$}{\YoungBBA}}  to be introduced and
the full system has the form
\begin{align}\label{MSHookFullSystem} &de^a\fm{2}=h_b h_c\omega^{[abc]}\fm{1}, & &\delta
e^a\fm{2}=d\xi^a\fm{1}+h_b h_c\xi^{[abc]}\fm{0}, & &\delta\xi^a\fm{1}=d\chi^a\fm{0},\nonumber\\
&d \omega^{[abc]}\fm{1}=h_dh_fC^{[abc],[df]}\fm{0}, &
&\delta\omega^{[abc]}\fm{1}=d\xi^{[abc]}\fm{0},& & \nonumber\\
&d C^{[abc],[df],g(k)}\fm{0}=h_vC^{[abc],[df],g(k)v}\fm{0}. & & &&
\end{align}
Consequently, the unfolded system incorporates field
$\phi_{\mu\mu,\nu}$ with all required differential gauge parameters
at all levels of reducibility; the \redundant\ component of the
\physical\ vielbein does not contribute to the dynamics.
$\omega^{[abc]}\fm{1}\sim\smallpic{\YoungAAA}\otimes\smallpic{\YoungA}=\smallpic{\YoungAAAA}\oplus\smallpic{\YoungBAA}\oplus\smallpic{\YoungAA}$,
the first component is a field strength for the \redundant\ field,
once the field has been gauged away, the associated field strength
is zero and
$\omega^{[abc]}_\rho=h^{a\mu}h^{b\nu}h^{c\lambda}T_{[\mu\nu\lambda],\rho}$
for some $T_{[\mu\nu\lambda],\rho}$,
$T_{[\mu\nu\lambda,\rho]}\equiv0$, which incorporates both
\smallpic{\YoungBAA} and \smallpic{\YoungAA} (as the trace). To
recover equations (\ref{FlatMSHookA}) in terms of \metric\
fields it is convenient to convert all fiber indices to the world
ones in the first two equations
\begin{align}
    &\pl_{[\mu} e^a_{\nu\lambda]}=h_{b[\mu}h_{c\nu}\omega^{abc}_{\lambda]},\\
    &\pl_{[\mu}\omega^{abc}_{\nu]}=h_{d[\mu}h_{f\nu]}C^{abc,df},
\end{align}
to substitute $T_{\mu\nu\lambda,\rho}$ and to contract two indices
in the second equation
\begin{align}
    &\label{MSHookA}\pl_{[\mu} \phi_{\nu\lambda],\rho}=T_{\mu\nu\lambda,\rho},\\
    &\label{MSHookB}\pl^{\rho}T_{\mu\nu\rho,\lambda}-\pl_\lambda
    T_{\mu\nu\rho,}^{\phantom{\mu\nu\rho,}\rho}=0.
\end{align}
Plugging (\ref{MSHookA}) in (\ref{MSHookB}) gives equation
(\ref{FlatMSHookA}).

For $\Yy=\Y{(2,1),(1,1)}$, i.e., $N=2$ and $p=2$, the general scheme
of Section \ref{Results} results in
\begin{align*}
    \Yy_g&: && \YoungA & & \YoungAAA & & \YoungBBA & & \YoungCBA & & ... \\
    \{n,k\}&: && \{2,0\} & & \{1,0\} & &\{0,0\} & & \{0,1\} & & ...\\
    g&:&& g=0 & &g=1 & &g=2 & &g=3 & & ...\\
    q_g&:&& q_0=2 & & q_1=1 & & q_2=0 & & q_3=0 & & ...,
\end{align*}
therefore, \be\WW\fm{q}=\begin{cases}
    \{W^a\fm{0}\}, & q=0,\\
    \{W^a\fm{1},W^{[abc]}\fm{0}\}, & q=1,\\
    \{W^a\fm{q},W^{[abc]}\fm{q-1},W^{[abc],[df]}\fm{q-2},W^{[abc],[df],g}\fm{q-2},...\}, & q>1,
\end{cases}\ee
and \be\sigma_-(W^g\fm{q})=\begin{cases}
    0, & g=0,\\
    h_b h_c W^{[abc]}\fm{q-1}, &g=1,\\
    h_bh_c W^{[aaa],[bc]}\fm{q-2}, & g=2,\\
    h_d W^{[aaa],[bb],c(g-3)d}\fm{q-2}, & g>2.
\end{cases}\ee
with $\DD=D_L-\sigma_-$ unfolded system
(\ref{MSHookFullSystem}) can be rewritten as \be
\DD\omega\fm{2}=0,\qquad \delta
\omega\fm{2}=\DD\xi\fm{1},\qquad \delta
\xi\fm{1}=\DD\xi\fm{0}\ee where
$\omega\fm{2}\in\WW\fm{2}$, $\xi\fm{1}\in\WW\fm{1}$ and
$\xi\fm{0}\in\WW\fm{0}$. The fields at $\{n=2,k=0\}\sim g=0$ and
$\{n=1,k=0\}\sim g=1$ form two finite-dimensional \iso-modules,
whereas the fields with $\{n=0,k=0,1,...\}\sim g\geq 2$ form an
infinite-dimensional \iso-module. $\blacksquare$\end{Example}
\begin{Example}\textbf{Spin-\smallpic{\YoungCB} field (briefly).}
In the case of a massless spin-\smallpic{\YoungCB} \field, there are
two gauge parameters at the first level of reducibility, which have
the symmetry\footnote{To simplify the example no consideration is given
to the trace components of the fields.} of \smallpic{\YoungBB},
\smallpic{\YoungCA}, and one gauge parameter at the second level
with the symmetry of \smallpic{\YoungBA}. First, for
$\Yy=\Y{(3,1),(2,1)}$ the proposed scheme results in
\begin{align*}
    \Yy_g&: && \YoungBA & & \YoungBAA & & \YoungBBB & & \YoungCCB & & \YoungDCB & &... \\
    \{n,k\}&: && \{2,0\} & & \{2,1\} & &\{1,0\} & & \{0,0\} & & \{0,1\} & & ...\\
    g&:&& g=0 & &g=1 & &g=2 & &g=3 & &g=4 & & ...\\
    q_g&:&& q_0=2 & & q_1=2 & & q_2=1 & & q_3=0 & & q_4=0 & & ...
\end{align*}
therefore, the \physical\ vielbein has to be a two-form taking
values in
$\smallpic{\YoungBA}$, i.e.,
$e^{\smallpic{\YoungBA}}\fm{2}\equiv e^{aa,b}\fm{2}$. Associated
gauge parameter $\xi^{\smallpic{\YoungBA}}\fm{1}\equiv
\xi^{aa,b}\fm{1}$ has components
$\smallpic{\YoungBA}\otimes\smallpic{\YoungA}=\smallpic{\YoungCA}\oplus\smallpic{\YoungBB}\oplus\smallpic{\YoungBAA}$,
the first two having the symmetry of the required gauge parameters,
with the third one being \redundant. The last three components in
the decomposition of the \physical\ vielbein
$\smallpic{\YoungBA}\otimes\smallpic{\YoungAA}=\smallpic{\YoungCB}\oplus\smallpic{\YoungCAA}\oplus\smallpic{\YoungBBA}\oplus\smallpic{\YoungBAAA}$
are also \redundant. By virtue of \Stueckelberg\ symmetry with
parameter $\xi^{\smallpic{\YoungBAA}}\fm{1}$ these three components
can be gauged away, inasmuch as
$\smallpic{\YoungBAA}\otimes\smallpic{\YoungA}=\smallpic{\YoungCAA}\oplus\smallpic{\YoungBBA}\oplus\smallpic{\YoungBAAA}$.
The \redundant\ component \smallpic{\YoungBAA} of the first level
gauge parameter $\xi^{\smallpic{\YoungBA}}\fm{1}$ also can be gauged
away by virtue of \Stueckelberg\ symmetry with level-two gauge
parameter $\xi^{\smallpic{\YoungBAA}}\fm{0}$. Consequently, at least
the fact that the unfolded system incorporates the \metric\ field,
all the required gauge parameters, and \redundant\ components do not
contribute to the dynamics is proved.$\blacksquare$
\end{Example}

\subsection{General Mixed-Symmetry Fields.}
\textit{Given a massless field of spin
$\Yy=\Y{(s_1,p_1),...,(s_N,p_N)}$, described by a \metric\ field
$\phi_{\Yy}(x)$, taking values in the \irrep\ of the Lorentz algebra
\lorentz\ with the same symmetry type(\minimalformulation) there
exists a unique unfolded system with the
\physical\ vielbein at the \lowestgrade\ and all symmetries being
manifest that reproduces the original \metric\ system. The system
has the form (\ref{ResultsFullSystem}).}

\begin{Sketch}\
\par
\begin{enumerate}
    \item The lowest grade field (\physical\ vielbein) $e^{\Yy_0}\fm{q_0}$ turns out to be completely
    determined by the requirement for it to contain the \metric\
    field and for its gauge parameters $\xi^{\Yy_0}\fm{q_0-k}$, $k>0$ to contain all the necessary gauge parameters at all levels of reducibility (\ref{FlatMSGaugeParameters}).
    \item $e^{\Yy_0}\fm{q_0}$ and its gauge parameters appear to contain \redundant\ components that have to be made
    non-dynamical. The only way to achieve this is to introduce
    \Stueckelberg\ symmetry $\delta e^{\Yy_0}\fm{q_0}=\sigma_-^1(\xi^{\Yy_1}\fm{q_1-1})$, $\delta \xi^{\Yy_0}\fm{q_0-1}=\sigma_-^1(\xi^{\Yy_1}\fm{q_1-2})$, ... with certain $\xi^{\Yy_1}\fm{q_1-k}$, $k>0$,
    which turns out to be unambiguously defined by this requirement.
    \item $\xi^{\Yy_1}\fm{q_1-1}$ has associated gauge field $\omega^{\Yy_1}\fm{q_1}$ and gauge transformations
    $\delta \omega^{\Yy_0}\fm{q_0}=\sigma_-^1(\xi^{\Yy_1}\fm{q_1-1})$
    completely determines the first equation $d e^{\Yy_0}\fm{q_0}=\sigma_-^1(\omega^{\Yy_1}\fm{q_1})$, which implies
    Jacobi identity $\sigma_-^1(d\omega^{\Yy_1}\fm{q_1})\equiv0$. The general solution is proved in Section \ref{MSJacobi} to have
    the form $d
    \omega^{\Yy_1}\fm{q_1}=\sigma_-^2(\omega^{\Yy_2}\fm{q_2})$ with certain
    $\omega^{\Yy_2}\fm{q_2}$. The second equation implies the second Jacobi identity
    $\sigma_-^2(d\omega^{\Yy_2}\fm{q_2})\equiv0$ and so on.
    \item Once, the total unfolded system is
    defined, the proof of the facts that (i) the correct \dynamical\ equations are
    imposed; (ii) there are no other \dynamical\ fields or other differential gauge symmetries in the system; is obtained through the calculations of $\sigma_-$ cohomology
    groups. Ignoring the trace pattern of
    fields, i.e., for traceful tensors, the proof of (ii) can be simplified and is done in this
    section. Even stronger statement that there exists a duality $\coh^{p+1+k}\sim\coh^{p-k}$ is
    proved in Section \ref{MSCohomology}.
\end{enumerate}
\end{Sketch}
First, as it is clear from the examples above, the \metric\ field
$\phi_{\Yy_M}(x)$ has to be incorporated into a differential form
$e^{\Yy_0}\fm{q}$ of degree-$q$ taking values in certain
\lorentz-\irrep\ $\Yy_0$  (for this reason $e^{\Yy_0}\fm{q}$ is
called \physical\ vielbein). Having converted all world indices to fiber ones in accordance with (\ref{UnfldConvertToFiber}), \lorentz-tensor
product $\Yy_0\otimes\Y{(1,q)}$ of $\Yy_0$ with one column diagram
of the height $q$ must contain the component with the
symmetry of $\Yy_M$. In the case of the minimal formulation
$\Yy_M=\Yy$, where $\Yy=\Y{(s_1,p_1),...,(s_N,p_N)}$ characterizes
the spin.

It is useful to introduce a notion of the quotient of two Young
diagrams: let the quotient of two Young diagrams $\Yy_1/\Yy_2$ be a
direct sum of those diagrams, whose \lorentz-tensor product by
$\Yy_2$ contains $\Yy_1$. For the first sight, given $q\geq0$, any
element $Q$ of $\Yy_M/\Y{(1,q)}$ might be chosen as $\Yy_0$. If $q$
is greater than the height of $\Yy_M$ the component with the
symmetry $\Yy_M$ in the tensor product $Q\otimes\Y{(1,q)}$ has to be
certain trace, inevitably.

Although, there is no apparent problems in considering unfolded
systems with dynamical field hidden in certain trace component of
the \physical\ vielbein, e.g., Maxwell gauge potential $A_\mu$ might
be identified with the trace $e^a_{\mu\nu}h_a^{\mu}$ of two-form
$e^a_{\mu\nu}$. However, this exotic incorporating of the
\physical\ field does not take place in the already known systems.
Moreover, it can be proved that such exotic systems do not exist, if
irreducible. Nevertheless, they might be a part of some reducible
system, which describes more than one \field. Therefore, we
require the degree $q$ of the \physical\ vielbein to be not greater
than the height of $\Yy_M$.

Inasmuch as (i) the equations imposed on $\phi_{\Yy}(x)$ possess
reducible gauge transformations with the number of levels equal to
the height $p=\sum_{i=1}^{i=N}p_i$ of $\Yy$; (ii) for a degree-$q$
gauge field of an unfolded system there exist $q$ levels of gauge
transformations; the degree $q$ of \physical\ vielbein
$e^{\Yy_0}\fm{q}$ must be equal to $p$.

The quotient $\Yy/\Y{(1,p)}$ contains only one element (provided
$\phi_{\Yy}(x)$ is forbidden to be a trace component), it is the
diagram (\ref{ResultsPhysicalVielbein}) with the symmetry of
$\Yy_0=\Y{(s_1-1,p_1),...,(s_N-1,p_N)}$, i.e., it is equal to $\Yy$
without the first column. Therefore,  \physical\
vielbein $e^{\Yy_0}\fm{p}$ is completely defined. So does gauge parameters $\xi^{\Yy_0}\fm{p-k}$ at the $k$-th
level of reducibility, $k\in[1,p]$. The requirement for tensor
product $\Yy_0\otimes\Y{(1,p-k)}$ to contain all gauge parameters
$\xi^{i_k}$ given by (\ref{FlatMSGaugeParameters}) at the $k$-th level of reducibility is also satisfied. Consequently, the
whole pattern of fields/gauge parameters of the \metric\ formulation
is reproduced. Note that $\Yy_0\equiv\Yy_0$ and $q_0=p$ in the terms
of Section \ref{Results}.

Let us note that if one writes down the \physical \ vielbein
explicitly as $e^{a_1(s_1-1),...,a_p(s_p-1)}_{\mu_1...\mu_p}$, where
$s_i$ are the lengths of the rows of $\Yy$ and converts the form
indices to the tangent ones by virtue of inverse background
vielbein $h^{\mu a}$ \be\label{MSFrametoMetric}
e^{a_1(s_1-1),...,a_p(s_p-1)|[d_1...d_p]}=e^{a_1(s_1-1),...,a_p(s_p-1)}_{\mu_1...\mu_p}h^{\mu_1d_1}...h^{\mu_pd_p},\ee
\dynamical\ field $\phi_{\Yy}(x)$ should be identified with
$e^{a_1(s_1-1),...,a_p(s_p-1)|a_1...a_p}$, which has the symmetry of
$\Yy$ with some traces included. Consequently, the generalized
Labastida's double-tracelessness condition
(\ref{FlatMSLabastidaDoubleTracelessness}) \cite{Labastida:1986ft}
is obvious, inasmuch as the contraction of two metric tensor
$\eta^{a_ia_i}\eta^{a_ia_i}$ with the $i$-th group of indices
vanishes.

Generally, in addition to $\phi_{\Yy}(x)$,  decomposition
$\Yy_0\otimes\Y{(1,p)}$ of \physical\ vielbein $e^{\Yy_0}\fm{p}$
into Lorentz \irreps\ contains \redundant\ components, which must
not contribute to the physical degrees of freedom. So does
$\Yy_0\otimes\Y{(1,p-k)}$, i.e., in addition to $\xi^{i_k}$ it
contains a lot of components that can not be made genuine
differential gauge parameters. There are only two possible ways to
get rid of \redundant\ fields in the unfolded formalism:
\begin{enumerate}
  \item[A.] \redundant\ components could be directly fixed to zero by algebraic (\Stueckelberg) symmetry.
  \item[B.] \redundant\ components might be auxiliary fields for
        other fields (at lower grade) and so on.  The process stops since the grade
        is assumed to be bounded from below.
\end{enumerate}
We require \dynamical\ fields incorporated in the \physical\
vielbein to be at the lowest grade. Actually, (B) takes place for
\dual\ gauge descriptions, but not for the minimal. Therefore, the
only possibility to make \redundant\ components to be non-dynamical is to
introduce a \Stueckelberg\ symmetry with certain parameters
$\xi^{\Yy_1}\fm{q_1-k}$
\begin{align}
\delta e^{\Yy_0}\fm{p}&=d\xi^{\Yy_0}\fm{p-1}+\sigma_-^1(\xi^{\Yy_1}\fm{q_1-1}),\\
\delta
\xi^{\Yy_0}\fm{p-1}&=d\xi^{\Yy_0}\fm{p-2}+\sigma_-^1(\xi^{\Yy_1}\fm{q_1-2}),\\
...&...,
\end{align}
where $\sigma_-^1$ is a certain operator built of background
vielbein $h^a$ that contracts a number of indices of $\Yy_1$ with $h^{a}...h^{c}$ to
obtain $\Yy_0$. Diagram\footnote{The specific character of Minkowski massless fields is in that all \Stueckelberg\ parameters can be incorporated into a single $\Yy_1$-valued form, this not being the case for (anti)-de Sitter massless fields. } $\Yy_1$ and form degree $q_1$ can
be easily found since $\Yy_1\otimes\Y{1,q_1-i}$, $i>0$ must
contain all \redundant\ components of $e^{\Yy_0}\fm{p}$ and
$\xi^{\Yy_0}\fm{p-k}$.

In the unfolded formalism each gauge field possesses its own gauge
parameter and vice-versa. Therefore, there is a gauge field
$\omega^{\Yy_1}\fm{q_1}$, which contributes to the equations for
$e^{\Yy_0}\fm{p}$ as \be d e^{\Yy_0}\fm{p}
=\sigma_-^1(\omega^{\Yy_1}\fm{q_1}).\ee The first equation
determines the first Jacobi identity of the form
$\sigma_-^1(d\omega^{\Yy_1}\fm{q_1})\equiv0$, which can be solved as
\be d\omega^{\Yy_1}\fm{q_1}=\sigma_-^2(\omega^{\Yy_2}\fm{q_2})\ee
for certain gauge field $\omega^{\Yy_2}\fm{q_2}$ and operator
$\sigma_-^2$ that satisfies $\sigma^1_-\sigma_-^2\equiv0$. The
general solution of Jacobi identities is given in Section
\ref{MSJacobi}. From now on let the superscripts of $\sigma_-$ be
omitted. The second equation determines the second Jacobi identity
$\sigma_-(d\omega^{\Yy_2}\fm{q_2})\equiv0$, which can be also solved
and so on. Consequently, the knowledge of the lowest grade gauge
field $e^{\Yy_0}\fm{p}$ and of the \Stueckelberg\ symmetries
required to get rid of \redundant\ components determines the first
equation, which by virtue of Jacobi identities determines the second
and so on. The total unfolded system has the structure
\be\begin{split}
    d\omega^{\Yy_g}_{q_g}&=\sigma_-\left(\omega^{\Yy_{g+1}}_{q_{g+1}}\right),\quad g=0,1,...,\\
    \delta\omega^{\Yy_g}_{q_g}&=d\xi^{\Yy_g}_{q_g}+\sigma_-\left(\xi^{\Yy_{g+1}}_{q_{g+1}-1}\right),\\
    \delta\xi^{\Yy_g}_{q_g}&=....,
\end{split}\ee where $\sigma_-$ is certain
algebraic operator built of background vielbein $h^a$ and
$(\sigma_-)^2=0$.

The only facts remain to be proved are that there are no other
\dynamical\ fields and that the proper correspondence (duality)
between the cohomology groups holds true. To this end it is
sufficient to analyze only the \lorentz-\irreps\ content of
$\coh^q_g$, i.e., only the symmetry type and the multiplicity of
irreducible Lorentz tensors that by virtue of
(\ref{UnfldTraceStructure}) are the representatives of $\coh^q_g$.

It is proved in the next Section that there is a duality
$\coh^{p-k}_{g=0}\sim\coh^{p+k+1}_{g=1}$ between the $\sigma_-$
cohomology groups, which reveals the one-to-one correspondence
between the gauge \dynamical\ fields and equations of motion, the
level-$k$ gauge symmetries and the level-$k$ Bianchi identities.
The representatives of cohomology groups in the sector of fields and gauge symmetries
directly correspond to those of \metric\ approach and there is no
other nontrivial cohomology in these sectors.

Though, the trace pattern of fields/gauge parameters is also
important and the calculation of $\sigma_-$ cohomology groups
provides a comprehensive answer to this question, the very procedure
being a bit technical. Ignoring the trace pattern of the fields, it
can be easily proved that the required \metric\ field is
incorporated in \physical\ vielbein $e^{\Yy_0}_{q_0}$, all
the \redundant\ components of the \physical\ vielbein can be gauged
away by a pure algebraic symmetry and the same holds for
differential gauge parameters $\xi^{\Yy_0}_{q_0-k}$ at all levels of
reducibility, i.e., the required pattern of gauge parameters is
recovered, whereas all the \redundant\ components are either
\Stueckelberg\ or \auxiliary. This is the necessary condition only
and the traces has to be taken into account. Also, it is still to be
proved that the correct equations of motions are imposed. Provided
that there are no trace conditions on the fields, the system is
referred to as \offshell \cite{Vasiliev:2005zu}. An \offshell\
system may contain only constraints that express \auxiliary\ fields
via the derivatives of the \dynamical\ field, imposing no restrictions on the latter.

\textbf{The off-shell system.} Technically, ignoring the traces is
equivalent to the replacement of all \lorentz-\irreps\ by the
\sld-\irreps\ that are characterized by the same Young diagrams,
i.e., instead of taking \lorentz-tensor products one needs to apply
the \sld-tensor product's rules, which are much simpler. The
decomposition of the \physical\ vielbein $e^{\Yy_0}\fm{q_0}$ into
\sld-\irreps\ is given by diagrams $\Yy_{\{\alpha_i\}}$ of the
form \be\PhysicalVielbein,\label{UnfldDecompositionSL}\ee where
$\alpha_1+...+\alpha_{N+1}=q_0=p$. The decomposition of
\Stueckelberg\ gauge parameter $\xi^{\Yy_1}\fm{q_1-1}$ for
$e^{\Yy_0}\fm{q_0}$  has the components of the same form but with
$\alpha_{N+1}\geq1$ because $\Yy_1$ (\ref{ResultsFirstAuxiliary})
has already the form of $\Yy_0$ with one cell in the bottom-left.
Therefore, all components of $e^{\Yy_0}\fm{q_0}$ except for those
with $\alpha_{N+1}=0$ are of \Stueckelberg\ type, but there is only
one component of $e^{\Yy_0}\fm{q_0}$ with $\alpha_{N+1}=0$,
namely, it has $\alpha_i=p_i$, $i\in[1,N]$ and, hence, has the
symmetry of \Yy. The decomposition of level-$k$ differential
gauge parameter $\xi^{\Yy_0}\fm{q_0-k}$ has the form
(\ref{UnfldDecompositionSL}) with $\alpha_1+...+\alpha_{N+1}=p-k$
and, again, all the components of $\xi^{\Yy_0}\fm{q_0-k}$ with
$\alpha_{N+1}\geq1$ can be gauged away by pure algebraic
symmetry with $\xi^{\Yy_1}\fm{q_1-k-1}$. Consequently, there is a
complete matching between (\ref{FlatMSGaugeParameters}) and those in
$\xi^{\Yy_0}\fm{q_0-k}$ that are not pure gauge themselves. It can be
easily proved also that, if ignoring the traces, there are no other
\dynamical\ fields/differential gauge parameters at higher grade.
Consequently,
\begin{Lemma}
$\coh(\sigma_-)$ with respect to \sld\ are given by
\be\coh^k(\sigma_-)_{\sld}=\begin{cases}
            \Yy_{\{\alpha_i\}} : g=0,\quad \alpha_1+...+\alpha_{N}=k,\quad \alpha_{N+1}=0, & k\leq p \\
            \emptyset, & k>p.
        \end{cases}\ee
\end{Lemma}
The triviality of the higher $k>p$ cohomology groups for \sld, which
should contain equations of motion and Bianchi identities, is expected since the
system is \offshell. For the equations to have the second order in
derivatives the only nontrivial cohomology group $\coh^{p+1}$ must
be at the grade-one, i.e., $\coh^{p+1}_{g=1}\neq\emptyset$. As is
seen from the examples above, the equations correspond to those
representatives of $\coh^{p+1}_{g=1}$ that are certain traces.
Indeed, the total rank $|W^{\Yy_g}\fm{q_g}|$, which is equal to the
sum $q_g+|\Yy_g|$ of degree $q_g$ and rank of $\Yy_g$, is preserved by $\sigma_-$.
Therefore, $|W^{\Yy_g}\fm{q_g}|=|W^{\Yy_0}\fm{q_0}|+g$ and
$|R^{\Yy_1}\fm{q_1+1}|=|\omega^{\Yy_0}\fm{q_0}|+2=|\Yy|+2$, where
$e^{\Yy_0}\fm{q_0}\in\WW^{g=0}\fm{p}$ and $R^{\Yy_1}\fm{q_1+1}\in\WW^{g=1}\fm{p+1}$ contain \metric\
field $\phi_{\Yy_M}(x)$ and the equations on $\phi_{\Yy}(x)$, respectively. Consequently, the representative of
$\coh^{p+1}_{g=1}$ that corresponds to the equations on the
traceless part of $\phi_{\Yy}(x)$ has to be identified with
certain trace of the first order and so on for the traces of
$\phi_{\Yy}(x)$.

The calculation of $\sigma_-$ cohomology groups carried out in Section \ref{MSCohomology} implies
\begin{enumerate}
    \item The only nontrivial cohomology groups are $\coh^{p-k}_{g=0}$ and $\coh^{p+k+1}_{g=1}$, $k=1,...,p$. Therefore,
    the \dynamical\ fields and independent gauge parameters belong to the zero grade $\WW\fm{q}^{g=0}$ subspaces;
    the equations are of the second order and the Bianchi identities for the $k$-th level gauge symmetries are of the $(k+2)$-th order.
    \item $\coh^{p-k}_{r,g=0}\sim\coh^{p+k+1}_{r+k+1,g=1}$, i.e., there is a one-to-one correspondence between
    the elements of $\coh^{p-k}_{g=0}$ that are traces of the $r$-th order and the elements of $\coh^{p+k+1}_{g=1}$ that are traces of the $(r+k+1)$-th order. Roughly
        speaking, $\coh^p\sim\coh^{p+1}$, $\coh^{p-1}\sim\coh^{p+2}$ and so on. Therefore, there is a
        one-to-one correspondence between the \dynamical\ fields and the equations of motion, the
        level-$k$ gauge symmetries and the $k$-th Bianchi identities.
    \item The \lorentz-\irreps\ that correspond to the elements of $\coh^{p-k}(\sigma_-)$ and
    are the traces of the zeroth order are given by the \lorentz-Young diagrams that have the form (\ref{UnfldDecompositionSL})
     with $\alpha_{N+1}=0$ and $\alpha_1+...+\alpha_N=p-k$, which exactly reproduce the required pattern (\ref{FlatMSGaugeParameters}).
     Note that certain higher order traces are also the elements of cohomology groups.
     These fields represent the 'auxiliary' fields of the \metric\ formulation,
     which had to be introduced to make the gauge symmetry \offshell.
     To prove the theorem it is not necessary to know the concrete trace pattern,
     the duality between the cohomology groups is sufficient. The details of the trace pattern are in Section \ref{MSCohomology}.
\end{enumerate}

In the fermionic case the traces are substituted for $\Gamma$-traces. Important is that acting on tensor indices only, $\sigma_-$ does not break down the irreducibility of spin-tensors. The duality has the form $\coh^{p-k}_{r,g=0}\sim\coh^{p+k+1}_{r+k+1,g=0}$ or simply $\coh^{p-k}_{g=0}\sim\coh^{p+k+1}_{g=0}$, which means that equations are of the first order and the duality between fields/equations, gauge symmetries/Bianchi identities takes place, which completes the proof.

Note, that $\Yy_{g=s_1}\equiv\Yy_{\{n=0,k=0\}}$ has the symmetry of
the generalized Weyl tensor for a spin-$\Yy$ massless
\field \be\FieldContentWeylTensor.\ee
Analogously to the examples of Section \ref{UnfldExamples}, massless mixed-symmetry fields can be
described by the same unfolded system (\ref{ResultsSystem}) that is restricted to $\WW^g\fm{q}$ with $g\geq g_0$, i.e., the system contains infinitely many \dual\ formulations. As the generalized Weyl tensor and its descendants are non-gauge fields, the \dual\ descriptions
with $g_0\geq s_1$ are non-gauge and the \dual\ descriptions
with $0<g_0<s_1$ are gauge.

\subsection{Physical Degrees of Freedom Counting}\label{MSPDOFCounting}
Notwithstanding the simplicity of the unfolded form and the uniqueness of unfolding, there still might be a question of whether the unfolded equations do describe the correct number of physical degrees of freedom.

It is well-known that for systems with the first class constraints only, with massless
fields belonging to this class, the counting of degrees of freedom
is that one first level gauge parameter kills two degrees of freedom, one second level gauge parameters kills three degrees of freedom, and so
on. For example, a spin-two massless field possesses
$\frac{d(d-3)}2$ degrees of freedom, which is just
$\frac{d(d+1)}2-2d$, $\frac{d(d+1)}2$ and $d$ being the number of
components of $\phi_{\mu\nu}$ and $\xi_\mu$.

The complete
information concerning the 'number' of fields/gauge parameters,
i.e., the multiplicity and the symmetry of corresponding tensors, is
contained in $\coh^{k}_{g=0}(\sigma_-)$ for $k=0...p$. Not only can a number of physical
degrees of freedom be calculated but the whole exact sequence
that defines an \iso\ \irrep\ can be derived. The elements of this sequence are certain \msv\ tensors that define an \mls\ tensor as a quotient of \msv\ tensors, e.g., (\ref{FlatMSHookSequence}). The decomposition of \lorentz-fields into irreducible tensors of \msv\ can be done with the aid of $\pl_\mu$ or, after Fourier transform, with the aid of momentum $p_\mu$.

For example, for a spin-two field the cohomology groups that correspond to the \dynamical\ fields/differential gauge parameters and its decomposition into \msv\ \irreps\ have the form\\
\begin{tabular}{|c|c|c|}\hline
& \lorentz-representatives & reduction from \lorentz\ to \msv \\
\hline
$\coh^{0}$ & \phantom{\YoungAA} $\smallpic{\YoungA}\sim\xi_\mu$ & $\smallpic{\YoungA}\oplus\bullet$ \\
\hline $\coh^{1}$ &
\phantom{\YoungAA}$\smallpic{\YoungB}\oplus\bullet\sim\phi_{\mu\nu}$,
$\tr{\phi}{}{\nu}\neq0$  & $\smallpic{\YoungB}\oplus\smallpic{\YoungA}\oplus2\bullet$ \\
\hline
\end{tabular}\\
Quick sort of Young diagrams in $\ComplexC{2\coh^0}{\coh^1}{\Irrep{0}{\smallpic{\YoungB}}}$ gives the correct exact sequence
$\ComplexC{\smallpic{\YoungA}}{\smallpic{\YoungB}}{\Irrep{0}{\smallpic{\YoungB}}}$, which defines a massless spin-two \irrep\ $\Irrep{0}{\smallpic{\YoungB}}$.

For the simplest mixed-symmetry field, the hook-\smallpic{\YoungBA}, there are two levels of gauge transformation, hence, relevant cohomology groups are $\coh^{0}$, $\coh^{1}$ and $\coh^{2}$\\
\begin{tabular}{|c|c|c|}\hline
& \lorentz-representatives & reduction from \lorentz\ to \msv \\
\hline
$\coh^{0}$ & \phantom{\YoungAA} $\smallpic{\YoungA}\sim\xi_\mu$ & $\smallpic{\YoungA}\oplus\bullet$ \\
\hline $\coh^{1}$ &
\phantom{\YoungAA}$\smallpic{\YoungAA}\oplus\smallpic{\YoungB}\oplus\bullet\sim\xi^A_{\mu\nu}\oplus\xi^S_{\mu\nu}$,
$\xi_{\nu}^{S\nu}\neq0$  & $\smallpic{\YoungAA}\oplus\smallpic{\YoungB}\oplus2\smallpic{\YoungA}\oplus2\bullet$ \\
\hline
$\coh^{2}$ & \phantom{\YoungAA}$\smallpic{\YoungBA}\oplus\smallpic{\YoungA}\sim \phi_{\mu\nu,\lambda}$, $\phi_{\mu\nu,}^{\phantom{\mu\nu,}\nu}\neq0$ & $\smallpic{\YoungBA}\oplus\smallpic{\YoungAA}\oplus\smallpic{\YoungB}\oplus2\smallpic{\YoungA}\oplus\bullet$ \\
\hline\end{tabular}\\
Again, quick sort of diagrams in $\ComplexD{3\coh^0}{2\coh^1}{\coh^2}{{\Irrep{0}{\smallpic{\YoungBA}}}}$ gives  exact sequence (\ref{FlatMSHookSequence}).

Consequently, the problem of calculation the number of physical
degrees of freedom, to be precise, of deriving the exact sequence, is effectively reduced to
the simple combinatoric problem of (i) decomposition
\lorentz-Young diagram representatives of $\coh(\sigma_-)$ to
\msv-diagrams; (ii) cancellation of like terms in sequence \be\label{MSLongExactSequence}0\longrightarrow(p+1)\coh^0\longrightarrow
p\coh^1\longrightarrow...\longrightarrow2\coh^{p-1}\longrightarrow\coh^p\longrightarrow\Irrep{0}{\Yy}\longrightarrow0.\ee
Skipping combinatoric technicalities we state that in the general case of a spin-$\Yy$ massless mixed-symmetry field (\ref{MSLongExactSequence}) reduces to the correct exact sequence that defines a \uirrep\ \Irrep{0}{\Yy} of \iso.

\subsection{Solving the Generalized Jacobi Identities.}\label{MSJacobi}
Unfolding some dynamical system there arises a problem of solving
Jacobi identities (\ref{UnfldBianchi}) that have schematically a
form $h...h d\omega^{1...}\fm{q}\equiv0$, where a number of
background vielbeins $h^a$ is contracted with the fiber indices of
$\omega\fm{q}^{1...}$. The Jacobi identity restricts
$d\omega^{1...}\fm{q}$ to have a certain specific form
$d\omega^{1...}\fm{q}=h...h \omega^{2...}\fm{r}$, where the
\lorentz-\irrep, in which $\omega^{2...}\fm{r}$ takes values, the degree
$r$ and the projector built of $h...h$ are completely determined. The
solutions are given by\footnote{As was pointed out in Section \ref{UnfldExamples} nonzero traces either violate Bianchi identities or
introduce mass-like terms and, therefore, are ignored.}
\begin{Lemma} \textbf{A.} Let
$\omega^{a_1(s),...,a_p(s),b(k)}\fm{q}$ be a degree-$q$ form
taking values in $\Yy=\Y{(s,p),(k,1)}$ \lorentz-\irrep. The
general solution of \be \label{MSProlongationSimplest} h_c
d\omega^{a_1(s),...,a_p(s),b(k-1)c}\fm{q}=0\ee has the form\be
d\omega^{a_1(s),...,a_p(s),b(k)}\fm{q}=\begin{cases}
    h_c\omega^{a_1(s),...,a_p(s),b(k)c}\fm{q}, & k<s, \\
    h_c...h_c\omega^{a_1(s)c,...,a_p(s)c,a_{p+1}(s)c}\fm{q-p-1},
    & k=s,\quad q\geq p+1, \\
    0, & k=s,\quad q<p+1,\end{cases}\ee
where $\omega^{a_1(s),...,a_p(s),b(k+1)}\fm{q}$ and
$\omega^{a_1(s+1),...,a_p(s+1),a_{p+1}(s+1)}\fm{q-p-1}$ take
values in $\Y{(s,p),(k+1,1)}$ and $\Y{(s+1,p+1)}$
\lorentz-\irreps, respectively\footnote{The symmetric basis is
used, being more convenient in this case as the contracted with
vielbeins tensors are already irreducible.}.
\end{Lemma}
\begin{proof} The parametrization of $d\omega\fm{q}$ by
a degree-$(q+1)$ form taking values in the same \lorentz-\irrep,
$d\omega^{a_1(s),...,a_p(s),b(k)}\fm{q}=R^{a_1(s),...,a_p(s),b(k)}\fm{q+1}$,
obviously fails to satisfy (\ref{MSProlongationSimplest}). Inasmuch
as (\ref{MSProlongationSimplest}) contains a vielbein, the solution
has to contain a number of vielbeins too. The most general
parametrization of $d\omega^{a_1(s),...,a_p(s),b(k)}\fm{q}$ with
only one vielbein included has the form
$d\omega^{a_1(s),...,a_p(s),b(k)}\fm{q}=h_d
\omega^{a_1(s),...,a_p(s),b(k)|d}\fm{q}$ for some $q$-form
taking values in a tensor product of $\Yy$ by a vector
representation, i.e., there is no definite symmetry between index
$d$ and the rest of the indices. \be h_ch_d
\omega^{a_1(s),...,a_p(s),b(k-1)c|d}\fm{q}=0\longleftrightarrow\omega^{a_1(s),...,a_p(s),b(k-1)[c|d]}\fm{q}=0.\ee
Only those \irreps\ in $\Yy\otimes\YoungA$ are allowed that
have $c$ and $d$ symmetric, i.e., correspond to $\Y{(s,p),(k+1,1)}$,
for $k<s$. This is not possible in the case $k=s$, i.e.,
$\Yy=\Y{(s,p+1)}$, and the only possibility to have $c$ and $d$
symmetric is to add the whole column to $\Yy$, which requires
$d\omega^{a_1(s),...,a_p(s),b(k)}\fm{q}$ to be represented as
$d\omega^{a_1(s),...,a_p(s),b(k)}\fm{q}=h_c...h_c\omega^{a_1(s-1)c,...,a_p(s-1)c,a_{p+1}(s-1)c}\fm{q-p-1}$
and, therefore, $q$ must be large enough. Roughly speaking, the
proof is based on the fact that the anti-symmetrization of two
indices at the same row of a Young diagram is identically zero, the
anti-symmetrization being due to contraction with vielbeins.
\end{proof}
Note that choosing nonmaximal solutions of Jacobi identities results
in lowering the gauge symmetry so that not all \redundant\
components are excluded.

Unfolding totally-symmetric massless higher-spin fields the possible
r.h.s. terms in $d\omega^{a(s-1),b(t)}\fm{1}=...$ are restricted by
Jacobi identities and an essential use is made of
\begin{Corollary} The solution of Jacobi identity
\be h_c d\omega^{a(s-1),b(t-1)c}\fm{1}\equiv0\ee is of the form \be
d\omega^{a(s-1),b(t)}\fm{1}=
  \begin{cases}
    h_c \omega^{a(s-1),b(t)c}\fm{1} & t<s-1, \\
    h_ch_d C^{a(s-1)c,b(s-1)d}\fm{0}& t=s-1.
  \end{cases}
\ee
\end{Corollary}
The following statement is a generalization of the Lemma-A to the
case where $\{(s,p),(k,1)\}$ is a 'subdiagram' of a larger Young diagram
$\Yy$ that has a number of rows precedent/succedent to
$\{(s,p),(k,1)\}$. Appropriate Young symmetrizers have to be included as
the mere contraction of a number of vielbeins breaks the
irreducibility of the tensor. For instance,
$h_c\omega^{a(s_1),b(s_2-1)c}$ is already irreducible with the
symmetry of $\Y{s_1,s_2-1}$, but this is not the case for
$h_c\omega^{a(s_1-1)c,b(s_2)}$, which has to be added one term
$h_c\omega^{a(s_1-1)c,b(s_2)}+\frac1{s_1-s_2+1}h_c\omega^{a(s_1)b,b(s_2-1)c}$
to get the symmetry of $\Y{s_1-1,s_2}$.

\begin{Lemma} \textbf{B.} Let
$\omega^{....,a_1(s),...,a_p(s),b(k),....}\fm{q}$ be a degree-$q$
form taking values in $\Yy=\Y{(...,(s,p),(k,1),...}$
\lorentz-\irrep, where the dots stands for the blocks in $\Yy$
precedent/succedent to $\{(s,p),(k,1)\}$. The general solution of \be
\label{MSProlongationGeneral}
\YProjector{h_cd\omega^{....,a_1(s),...,a_p(s),b(k-1)c,....}\fm{q}}=0,\ee
where $\YProjector{...}$ is a Young symmetrizer to
$\Y{....,(s,p)(k-1,1),...}$, has the form\be
d\omega^{....,a_1(s),...,a_p(s),b(k),....}\fm{q}=\begin{cases}
    \YProjector{h_c\omega^{....,a_1(s),...,a_p(s),b(k)c,....}\fm{q}}, & k<s, \\
    \YProjector{h_c...h_c\omega^{....,a_1(s)c,...,a_p(s)c,a_{p+1}(s)c,....}\fm{q-p-1}},
    & k=s,\quad q\geq p+1, \\
    0, & k=s,\quad q<p+1,\end{cases}\ee
where $\omega^{....,a_1(s),...,a_p(s),b(k+1),....}\fm{q}$ and
$\omega^{....,a_1(s+1),...,a_p(s+1),a_{p+1}(s+1),....}\fm{q-p-1}$
takes values in \\$\Y{...,(s,p),(k+1,1),...}$ and
$\Y{...,(s+1,p+1),...}$ \lorentz-\irreps, respectively, and
$\YProjector{...}$ is a Young symmetrizer to
$\Y{....,(s,p),(k,1),...}$.
\end{Lemma}
\begin{proof} The proof is similar to that of Lemma-A with the only comment that the
symmetrizers do not affect the argumentation that
anti-symmetrization of two indices in the same row is identically
zero.
\end{proof}

\subsection{Dynamical Content via $\sigma_-$ Cohomology.}\label{MSCohomology}
Before discussing the calculation of $\sigma_-$ cohomology, let us
first  recall necessary facts concerning the evaluation of
\lorentz-tensor products. Let
$\partition{\epsilon_1,...,\epsilon_N|m}$ be a number of integer
partitions $k_1+...+k_N=m$ of $m$, where $k_i$ is constrained by
$k_i\leq\epsilon_i$ and different rearrangements of $k_i$ satisfying
the constraints are regarded as distinct partitions. The generating
function for the partition $\partition{\epsilon_1,...,\epsilon_N|m}$
is
\be\sum_{m}\partition{\epsilon_1,...,\epsilon_N|m}t^m=\prod_{i=1}^{i=N}\frac{(1-t^{\epsilon_i+1})}{1-t}\ee
These integer partitions gives the multiplicities of \irreps\ in
\lorentz-tensor products (Clebsh-Gordon coefficients) we are
interested in \cite{Barut}.

The \lorentz-tensor product of an arbitrary \irrep\ $\Yy^*=\Y{(s_i,p_i)}$ by
a one-column diagram of the height $q^*$ can be explicitly calculated
as \be\label{MSSoTensorProduct}\ABProduct,\ee where $\alpha_i$,
$\beta_i$ : $\alpha_i+\beta_i\leq p_i$, $i\in[1,N]$,
$\alpha_{N+1}\geq0$ and there exist $\rho\geq0$ such that
\be\label{MSRhoDefinition}
q^*=\sum_{i=1}^{i=N}\left(\alpha_i+\beta_i\right)+\alpha_{N+1}+2\rho.\ee
The multiplicity $N_{\{\alpha_j\},\{\beta_i\}}$ of
$\Yy^*_{\{\alpha_j,\beta_i\}}$ is given by integer
partition \be
N_{\{\alpha_j\},\{\beta_i\}}=\partition{\epsilon_1,...,\epsilon_N|\rho},\quad
\epsilon_i=p_i-\alpha_i-\beta_i\ee and the total trace order $r$ is
\be r=\sum_{i=1}^{i=N}\beta_i+\rho\ee Roughly speaking, to obtain an
element of the tensor product one should, first, cut off from
bottom-right of the $i$-th block a column of height
$\beta_i-\gamma_i$, $p_i\geq\beta_i-\gamma_i\geq0$ (to take
different traces) and, second, add an arbitrary number
$\alpha_i+\gamma_i$, $p_i\geq\alpha_i+\gamma_i\geq0$ of cells to
each block, provided the $\gamma_i$ cells annihilate, i.e., they are
added to the places from which $\gamma_i$ cells were removed at the
first stage. Multiplicity may be different from one due to different
rearrangements of $\gamma_i$, i.e., when multiplied by the rest of
$q$-column different traces can give rise to the same diagram. The
number of such rearrangements is given by
$\partition{\epsilon_1,...,\epsilon_N|\rho}$,
$\rho=\sum_{i=1}^{i=N}\gamma_i$. For instance, \sln-tensor
product $\smallpic{\YoungBA}\otimes\smallpic{\YoungAA}$ is given
simply by
\be\YoungBA\otimes_{\slns}\YoungAA=\YoungCB\oplus\YoungCAA\oplus\YoungBBA\oplus\YoungBAAA.\ee
The \sod-tensor product can be represented as a sum of the form
\begin{align} \YoungBA\otimes_{\sons}\YoungAA&=
\underbrace{\YoungBA\otimes_{\slns}\YoungAA}_{\mbox{zeroth order traces}}\bigoplus%
\underbrace{\YoungBAzA\bigoplus\YoungBAzB}_{\mbox{first order
traces}}\bigoplus\nonumber\\&\bigoplus\underbrace{\ \YoungBAzAB\ }_{\mbox{2nd
order}}=\left[\underbrace{\YoungCB\oplus
\YoungCAA\oplus\YoungBBA\oplus\YoungBAAA\ }_{\mbox{zeroth order traces}}\right]\bigoplus\nonumber\\&\bigoplus%
\left[\underbrace{\YoungC\oplus2\YoungBA\oplus\YoungAAA\
}_{\mbox{first order traces}}\right]
\bigoplus\underbrace{\YoungA}_{\begin{matrix} \mbox{2nd} \\
\mbox{order} \end{matrix}},\end{align} where the sum is over
trace order and the boxes that are connected by the arc are to be
contracted and, then, the $\sln$-product rules are to be applied to
the rest of the diagrams.

Evaluating the \lorentz-tensor product of an arbitrary spinor-\irrep\ $\Yy^*=\Y{(s_i,p_i)}_\ferm$ by
a one-column diagram of the height $q^*$ one can contract any number of $\Gamma$ matrices with the indices of the column and then multiply, therefore, the tensor product rule for spin-tensors can be reduced to  bosonic rule (\ref{MSSoTensorProduct}) as
\be\Y{(s_i,p_i)}_\ferm\otimes_{\sons}\Y{(1,q)}=\sum_{k=0}^{q}\left(\Y{(s_i,p_i)}\otimes_{\sons}\Y{(1,q-k)}\right)_\ferm,\ee for example,
\begin{align}\YoungA_\ferm\otimes_{\sons}\YoungAA=\underbrace{\YoungBA_\ferm\oplus\YoungAAA_\ferm}_{\mbox{zeroth order $\Gamma$-traces}}\bigoplus \underbrace{\YoungB_\ferm\oplus\YoungAA_\ferm}_{\mbox{first order $\Gamma$-traces}}\bigoplus\nonumber\\ \bigoplus\underbrace{2\YoungA_\ferm}_{\mbox{2nd order $\Gamma$-traces}}\bigoplus\underbrace{\bullet_\ferm}_{\mbox{3d order $\Gamma$-traces}}.\end{align}
The multiplicity $N^\ferm_{\{\alpha_j\},\{\beta_i\}}$ of
$\Yy^*_{\{\alpha_j,\beta_i\}\ferm}$ is given by
\be\label{MSFermMultiplicity} N^\ferm_{\{\alpha_j\},\{\beta_i\}}=\sum^{i=\rho}_{i=0}\partition{\epsilon_1,...,\epsilon_N|\rho-i}.\ee

Let us now analyze the origin of $\sigma_-$ cohomology. Its action
$\sigma_-: \WW^g\fm{q}\rightarrow\WW^{g-1}\fm{q+1}$ preserves $g+q$
and hence the whole complex $\Comp(\WW, \sigma_-)$ is a direct sum
of complexes $\Comp(q',\sigma_-)$, where $q'$ refers to the end
element $\WW^{g=0}\fm{q'}$ of the complex \be
\Comp(q',\sigma_-):\qquad
\ComplexC{...}{\WW^{g=1}\fm{q'-1}}{\WW^{g=0}\fm{q'}}.\ee Moreover,
$\sigma_-$ preserves the total rank (the form degree + rank of the
\lorentz-\irrep\ in which the field takes values). Let
$W^{a(s_1),...,a(s_m)}\fm{{q}}$ be an
element of $\WW^g\fm{q'}$. When all form indices of
$W^{a(s_1),...,a(s_m)}\fm{{q}}$ are converted to the fiber ones according to
(\ref{UnfldConvertToFiber}) \be
W^{a_1(s_1),...,a_m(s_m)|[d_1...d_{q}]}=W^{a_1(s_1),...,a_m(s_m)}_{\mu_1...\mu_{q}}h^{\mu_1d_1}...h^{\mu_{q}d_{q}},\ee
the action of $\sigma_-$ is just an anti-symmetrization of all
$d_1,...,d_{q}$ with those indices $a_i$ that are extra as compared
to $\Yy_{g-1}$ plus some terms to restore the correct Young symmetry
of $\Yy_{g-1}$. Let the decomposition of
$W^{a_1(s_1),...,a_m(s_m)|[d_1...d_{q}]}$ into \lorentz-\irreps\ be
of the form
\be\Y{a_1(s_1),...,a_m(s_m)}\bigotimes\Y{(1,{q})}=\bigoplus_{r=0}\bigoplus_{i_r}N^r_{i_r}\Yy^r_{i_r},\ee
where the summation is over the trace order $r$ and, then, over
$i_r$, which enumerates all \lorentz-\irreps\ that enters in the
tensor product as the traces of the $r$-th order, $N^r_{i_r}$ being
the multiplicity of $\Yy^r_{i_r}$. The multiplicity of the zeroth
order traces is always equal to one, $N^{0}_{i_0}=1$. In fact, the
diagrams $Y^0_{i_0}$ can be directly obtained by the $\sln$-tensor
product rule.

The very anti-symmetrization is insensitive to whether a certain
component enters as a trace or not. When decomposed into
\lorentz-\irreps\ the elements of $\WW^{g}\fm{q'}$ and
$\WW^{g-1}\fm{{q'+1}}$ have a number of components of the same
symmetry type. The action of $\sigma_-$ is just a linear
transformation that either sends the whole \lorentz-\irrep\ to
zero\footnote{provided that the appropriate basis on the space of
\lorentz-\irreps\ with the same symmetry type is chosen.} or sends
it to the components of $\WW^{g-1}\fm{{q'+1}}$ of the same symmetry
type. Important is that $\sigma_-$ does not act between different
\lorentz-\irreps. Therefore, complex $\Comp(q',\sigma_-)$ is a
direct sum of complexes, parameterized by \lorentz-\irreps\ $\Yy'$
that are given by various tensor products
\be\label{MSProductGeneral}\Yy_{\{n,k\}}\bigotimes\Y{(1,{q})}\ee
provided that the field $W^{\Yy_{\{n,k\}}}\fm{q}$ is an element of
$\WW^{g'-i}\fm{q'+i}$ for certain $i$. When reduced to such a complex
\be \Comp(\Yy',q',\sigma_-):\qquad 0\rightarrow...\rightarrow
V_{g}\rightarrow V_{g-1}\rightarrow...\rightarrow0\ee the action of
$\sigma_-$ is a linear transformation between the spaces
$V_g\rightarrow V_{g-1}$, dimensions of which are equal to the
multiplicity of $\Yy'$ in the decomposition of
$W^{Y_g}\fm{q_g}$ and $W^{Y_{g-1}}\fm{q_{g+1}}$. The action of
$\sigma_-$ is maximally non-degenerate  compatible with nilpotency.
Therefore, to find cohomology groups the dimension of each linear
space in every complex has to be calculated.

Let us note, that only the types and multiplicities of
\lorentz-\irreps\ are found below, i.e., no attention is paid to the
description of how the corresponding tensors are contained in the
fields $W^{\Yy}\fm{q'}$. This is sufficient for our purposes, though,
it is obvious (\ref{MSFrametoMetric}) how the \metric\ field
$\phi_{\Yy_M}$ is incorporated in an element of $\WW\fm{p}^{g=0}$.

The calculation of $\sigma_-$ cohomology is divided into three cases
that cover the whole variety of \lorentz-\irreps\ that can result
from the tensor products (\ref{MSProductGeneral}).

In the Lemmas below it is assumed that complex $\Comp(\Yy',q',\sigma_-)$ is parameterized by
\lorentz-\irrep\ $\Yy'=\Yy^*_{\{\alpha_j,\beta_i\}}$ of the form
(\ref{MSSoTensorProduct}) with $\Yy^*=\Yy_{\{n,k\}}$, $q^*=q$ such that $W^{\Yy_{\{n,k\}}}\fm{q}\in\Comp(q',\sigma_-)$ and $\rho$ is
defined according to (\ref{MSRhoDefinition}).

In the first case $\Yy'$ is an element of $\Yy_{0}\otimes\Y{(1,q)}$
\begin{Lemma}{\bf 1.} The $\sigma_-$ cohomology groups $\coh^{q}_{g=0}$ and $\coh^{q}_{g=1}$ are nontrivial and are given by
\be
\coh^q_g=
    \begin{cases}
        \Y{(s_i-1,p_i)}_{\{\alpha_j,\beta_i\}} M_{\{\alpha_j,\beta_i\}}: \begin{matrix} \alpha_{N+1}=0,\\ \sum_{i=0}^{i=N}\alpha_i=q,\end{matrix} & q\in[0,p], \quad g=0, \\
        \coh^{2p-q-1}, & q\in[p+1,2p+1],\quad g=1,\\
        \emptyset, & q>2p+1,\quad \mbox{any\ }g,
    \end{cases}
\ee
where $M_{\{\alpha_j,\beta_i\}}$ is the multiplicity of the \irrep\ $\Y{(s_i-1,p_i)}_{\{\alpha_j,\beta_i\}}$, defined below.
\end{Lemma}
\begin{proof} From the very form of $\Yy_{\{n=N,0\}}\equiv\Yy_{g=0}\equiv\Y{(s_i-1,p_i)}$ it
follows that the \lorentz-\irreps\ that the element of $\WW\fm{q}^{g=0}\sim\Yy_{0}\otimes\Y{(1,q)}$ decomposes into may be contained in $\WW\fm{q-1}^{g=1}$ only, i.e., $\WW\fm{q}^{g=0}$ and $\WW\fm{q-k}^{g=k}$ contain no \lorentz-\irreps\ of the same symmetry type for $k>1$. Therefore, the length of complex $\Comp(\Yy',q,\sigma_-)$, where $\Yy'\in\Yy_{0}\otimes\Y{(1,q)}$ is equal to one, i.e., $\Comp(\Y{(s_i-1,p_i)}_{\{\alpha_j,\beta_i\}},q,\sigma_-): \ComplexB{V_1}{V_0}$, where the dimensions of $V_0$ and $V_1$ are given by the multiplicities of the \lorentz-\irrep\
$\Y{(s_i-1,p_i)}_{\{\alpha_j,\beta_i\}}$ in $\WW\fm{q-1}^{g=1}$ and $\WW\fm{q}^{g=0}$, respectively,
\begin{align}
    &\dimension(V_1)=
        \begin{cases}
            \begin{cases}
                \partition{\epsilon_1,...,\epsilon_N|\rho-1} & \rho\geq1, \nonumber\\
                0, & \rho=0,
            \end{cases} & \alpha_{N+1}=0,\\
            \partition{\epsilon_1,...,\epsilon_N|\rho}, & \alpha_{N+1}>0,\end{cases}\label{MSLemmaAMultiplicity}\\
    &\dimension(V_0)=\partition{\epsilon_1,...,\epsilon_N|\rho},
\end{align}
If $\alpha_{N+1}>0$ the dimensions of $V_0$ and $V_1$ are equal and each irreducible component of $\WW^{g=0}\fm{q}$ with $\alpha_{N+1}>0$ can be gauge away by virtue of the corresponding element of $\WW\fm{q-1}^{g=1}$, thus being exact.
$\coh^{k>2p+1}=\emptyset$, inasmuch as $\rho\leq p$ and the
components of the forms with rank greater than $(2p+1)$ must have
$\alpha_{N+1}>0$, thus being exact. If $\alpha_{N+1}=0$ the dimensions are different, $\dimension(V_1)<\dimension(V_0)$ for $q\leq p$, $\dimension(V_1)=\dimension(V_0)$ for $q=p+1$ and $\dimension(V_1)>\dimension(V_0)$ for $q>p$. Therefore, the number of those $\Y{s_i,p_i}_{\{\alpha_j,\beta_i\}}\in\WW^{g=0}\fm{q}$, $q\leq p$ that are not exact is equal to $M_{\{\alpha_j,\beta_i\}}=|\dimension(V_1)-\dimension(V_0)|$, the same is the number of those $\Y{s_i,p_i}_{\{\alpha_j,\beta_i\}}\in\WW^{g=1}\fm{q-1}$, $q>p+1$ that are not exact. The obvious property of integer partitions
\be\partition{\epsilon_1,...,\epsilon_N|m}=\partition{\epsilon_1,...,\epsilon_N|\epsilon_1+...+\epsilon_N-m}\ee
results in the important duality in the cohomology groups
$\coh^{p-k}_{r,g=0}\sim\coh^{p+k+1}_{r+k+1,g=1}$ or, roughly
speaking, $\coh^p\sim\coh^{p+1}$, $\coh^{p-1}\sim\coh^{p+2}$ and so
on.
\end{proof}
The second case is the complex $\Comp(\Yy',q',\sigma_-)$, where $\Yy'$ is an element of the tensor product $\Yy_{\{n,k\}}\otimes\Y{(1,q)}$ for certain $q\geq0$ with $1<k<\kmax{n}-1$.
\begin{Lemma}{\bf 2.}\label{LemmaA} If $\Yy'$ is an element of the tensor product $\Yy_{\{n,k\}}\otimes\Y{(1,q)}$ for certain $q'\geq0$ with $1<k<\kmax{n}-1$, the complex $\Comp(\Yy',q',\sigma_-)$ is acyclic. \end{Lemma}
\begin{proof}The complex has the length two, i.e., $\ComplexC{V_0}{V_1}{V_2}$, where the dimensions are given by
$\dimension(V_0)=\partition{\epsilon_1,...,\epsilon_N|\rho}$,
$\dimension(V_1)=\partition{\epsilon_1,...,\epsilon_N,1|\rho+1}$,
$\dimension(V_2)=\partition{\epsilon_1,...,\epsilon_N|\rho+1}$.
Simple calculations with generating functions results in
$\dimension(V_0)<\dimension(V_1)$,
$\dimension(V_2)<\dimension(V_1)$ and
$\dimension(V_0)-\dimension(V_1)+\dimension(V_2)=0$ and,
consequently, the sequence is exact.
\end{proof}
Analogously,
\begin{Lemma}{\bf 3.} If $\Yy'$ is an element of the tensor product $\Yy_{\{n,k\}}\otimes\Y{(1,q)}$ for certain $q\geq0$ with $k=0$ and $n<N$, the complex $\Comp(\Yy',q',\sigma_-)$ is acyclic. \end{Lemma}
\begin{proof} The speciality of such $\Yy'$ is that the complex has the length three, i.e., $\ComplexD{V_0}{V_1}{V_2}{V_3}$, where the dimensions are given by
\begin{align*}
    &\dimension(V_0)=\partition{\epsilon_1,...,\epsilon_N|\rho},\\
    &\dimension(V_1)=\partition{\epsilon_1,...,\epsilon_{n-1},\epsilon_n+1,\epsilon_{n+1},...,\epsilon_N|\rho},\\
    &\dimension(V_2)=
        \begin{cases}
            \partition{\epsilon_1,...,\epsilon_{n-1},p_n-\epsilon_n,\epsilon_{n+1},...,\epsilon_N|\rho-\epsilon_n-1} & \rho\geq\epsilon_n+1, \\
            0, & \text{otherwise},
        \end{cases},\\
      &\dimension(V_3)=
        \begin{cases}
            \partition{\epsilon_1,...,\epsilon_{n-1},p_n-\epsilon_n-1,\epsilon_{n+1},...,\epsilon_N|\rho-\epsilon_n-2} & \rho\geq\epsilon_n+2, \\
            0, & \text{otherwise},
        \end{cases},
\end{align*}
Again, simple calculations with generating functions results in
appropriate inequalities and
$\dimension(V_0)-\dimension(V_1)+\dimension(V_2)-\dimension(V_3)=0$
and, consequently, the sequence is exact.
\end{proof}
The above three cases covers the whole variety of the \lorentz-tensors that can result from the tensor products $\Yy_{\{n,k\}}\otimes\Y{(1,q)}$ for any $n$, $k$ and $q$.
The cases with $s_N=1$ and $\Yy=\Y{0}$ are special but the calculation of cohomology groups leads to the same result.

In the fermionic case, the computations are similar due to the multiplicity given by (\ref{MSFermMultiplicity}). The difference is that all nontrivial cohomology is concentrated in grade zero. Indeed, taking into account (\ref{MSLemmaAMultiplicity}), where $\partition{\epsilon_1,...,\epsilon_N|\rho}$ is to be replaced by (\ref{MSFermMultiplicity}), it follows that $\dimension(V_0)\geq\dimension(V_1)$ for all $q$ and $\coh^{p-k}_{r,g=0}\sim\coh^{p+k+1}_{r+k+1,g=0}$, where $r$ is referred to the $\Gamma$-trace order, or, roughly
speaking, $\coh^p\sim\coh^{p+1}$, $\coh^{p-1}\sim\coh^{p+2}$ and so
on.

\section{Conclusions}\label{Conclusions}

The unfolded system constructed in the paper has a simple
form of a covariant constancy equation and describes arbitrary mixed-symmetry bosonic and fermionic massless fields in $d$-dimensional Minkowski space. The gauge
fields/parameters are differential forms with values in certain
finite-dimensional irreducible representations of the Lorentz
algebra that are uniquely determined by the generalized spin.

The key moment  is that all gauge symmetries are manifest within the
unfolded formulation, which is of most importance in controlling the number of physical degrees of freedom when trying to introduce interactions.  Unfolded systems are formulated in terms of differential forms,
which is a natural way to respect diffeomorphisms and, hence, to describe systems that include gravity.

Though, the necessary conditions for the system to have a
lagrangian description are satisfied, i.e., the fields are in
one-to-one correspondence with the equations and the $k$-th level
gauge symmetries are in one-to-one correspondence with the $k$-th
Bianchi identities, the very lagrangian remains to be constructed.

Another interesting moment is that the unfolded equations for bosons are the same as for fermions, namely, the operators involved, i.e., exterior differential $d$ and $\sigma_-$, remains unmodified when tensors are replaced with spin-tensors. Though this nice property partly breaks down in (anti)-de Sitter space, within the unfolded approach bosons and fermions have much in common, the fact being very useful for supersymmetric theories.

The proposed unfolded system includes all non-gauge
\dual\ descriptions, which are based on the generalized Weyl tensor and its
descendants, but not all of gauge \dual\ descriptions. It would
be interesting to construct an unfolded system that contains all
\dual\ formulations.

The interactions of the totally symmetric massless higher-spin
fields are known to require a nonzero cosmological constant, i.e., are formulated in (anti)-de
Sitter space \cite{Fradkin:1986qy}. Mixed-symmetry fields exhibit some interesting features in the
presence of cosmological constant. For example, not all of the
Minkowski gauge symmetries can be deformed to (anti)-de Sitter
\cite{Metsaev:1995re}. As a result massless mixed-symmetry fields
have more degrees of freedom in (anti)-de Sitter compared to its
Minkowski counterparts \cite{Brink:2000ag} and in the Minkowski limit
a massless mixed-symmetry field splits in a certain
collection of massless fields, generally. Contrariwise, a single
mixed-symmetry field can not be smoothly deformed to (anti)-de
Sitter. Another interesting effect is the existence of the
so-called partially-massless fields
\cite{Deser:1983mm,Deser:1983tm,Higuchi:1986wu,Higuchi:1986py,Higuchi:1989gz,Deser:2001us,Deser:2001xr,Zinoviev:2001dt,Deser:2004ji,Deser:2003gw,Skvortsov:2006at},
the fields that have a number of degrees of freedom intermediate
between that of massless and massive and split in a set of
massless fields in the Minkowski limit. Therefore, the extension of the proposed approach to (anti)-de
Sitter space seems to be non-trivial but nevertheless worth being investigated.

In a series of papers
\cite{Alkalaev:2003hc,Alkalaev:2003qv,Alkalaev:2005kw,Alkalaev:2006rw}
it was suggested so-called \framelike\ approach to the general
mixed-symmetry fields in (anti)-de Sitter. To generalize
the proposed in the present paper unfolded system to (anti)-de
Sitter case and to compare to that of \cite{Alkalaev:2006rw} is the
next step to perform.

We consider the proposed unfolded system as the first
stage in constructing the full interacting theory of mixed-symmetry
fields.

\section*{Acknowledgements}
The author appreciates sincerely M.A. Vasiliev for reading the manuscript and making many detailed and valuable suggestions and illuminating comments, the very work was initiated as the result of discussions with M.A. Vasiliev.
The author is also grateful to R.R. Metsaev and K.B. Alkalaev for useful discussions and to O.V. Shaynkman for discussion of the fermionic case.
The work was supported in part by
grants RFBR No. 05-02-17654, LSS No. 4401.2006.2 and
INTAS No. 05-7928, by the Landau Scholarship and by the Scholarship of Dynasty foundation.

\appendix
\renewcommand{\theequation}{\Alph{section}.\arabic{equation}}
\section*{Appendix A: Multi-index convention}
\setcounter{equation}{0}
\setcounter{section}{1} \label{AppNotation}
The multi-index notations is used: a group of indices in which certain tensor is symmetric or is to be symmetrized is denoted either by one letter with the number of indices indicated in round brackets, or
by placing a group of indices in round brackets, e.g.,
\be T^{a(s)}\equiv T^{a_1 a_2 ... a_s}:\qquad T^{a_1...a_i...a_j...a_s}=T^{a_1...a_j...a_i...a_s}\ee
\be V^{a}T^{a(s)}\equiv V^{(a_1}T^{a_2...a_{s+1})}\equiv\frac1{s+1}\left(V^{a_1}T^{a_2 a_3...a_{s+1}}+V^{a_2}T^{a_1 a_3...a_{s+1}}+...+V^{a_{s+1}}T^{a_1 a_2...a_{s}}\right)\ee
\be V^{(b}T^{a(s))}\equiv V^{(b}T^{a_1...a_{s})}\equiv\frac1{s+1}\left(V^{b}T^{a_1 a_2...a_{s}}+V^{a_1}T^{b a_2 a_3...a_{s}}+...+V^{a_{s}}T^{b a_1 a_2...a_{s-1}}\right)\ee
Analogously, the group of indices in which certain tensor is anti-symmetric or is to be anti-symmetrized is denoted by placing indices in square brackets, e.g.,
\be T^{a[s]}\equiv T^{a_1 a_2 ... a_s}:\qquad T^{a_1...a_ia_{i+1}...a_s}=-T^{a_1...a_{i+1}a_i...a_s}\ee
\be V^{[b}T^{a[s]]}\equiv V^{[b}T^{a_1...a_s]}\equiv \frac1{s+1}\left(V^{b}T^{a_1 a_2...a_{s}}-V^{a_1}T^{b a_2 a_3...a_{s}}+...+(-)^s V^{a_{s}}T^{b a_1 a_2...a_{s-1}}\right)\ee
The operators of (anti)-symmetrization are weighted to be projectors (the factor $\frac1{s+1}$ above).

\section*{Appendix B: Young diagrams}
\setcounter{equation}{0}\setcounter{section}{2}\label{AppYoung}
Comprehensive information on Young diagrams can be found, for
example, in the textbook \cite{Barut}.
\begin{Definition} Given an integer partition, i.e., a nonincreasing sequence $\{s_i, i\in[1,n]\}$, $s_i\geq s_{i+1}$ of
positive integers (or nonnegative when it is convenient to work with
a sequence of a fixed length), associated Young diagram
\Y{s_1,s_2,...,s_n} is a graphical representation consisting of
$n$ left-justified rows made of boxes, with the $i$-th row containing $s_i$ boxes.\\
\BlockDEnum{8}{4}{5}{3}{3}{2}{1}{2}{s}
\end{Definition}
Finite-dimensional irreducible representation(\irrep) of $\sld$,
i.e., various irreducible \sld-tensors, are in one-to-one
correspondence with Young diagrams of the form
\Y{s_1,s_2,...,s_{[\frac{d}2]}}. The associated irreducible tensors
\be T^{\overbrace{a...a}^{s_1},\overbrace{b...b}^{s_2},...}\ee or,
in condensed notation, $T^{a(s_1),b(s_2),...}$ have at most $[\frac{d}2]$
groups of indices, being symmetric in each group separately, and
satisfy the condition that the symmetrization of any group of
indices with one index of any of the subsequent groups is
identically zero, i.e., \be
T^{a(s_1),...,(b(s_k),...,b)c(s_j-1),...}\equiv0,\qquad k<j.\ee If
$s_{[\frac{d}2]}\neq0$ and $d$ is even the (anti)-selfduality
condition has to be imposed for appropriate signature. (anti)-selfduality is conventionally
denoted by the sign factors $+(-)$ before $s_{[\frac{d}2]}$. In this paper we do
not consider (anti)-self dual representations.

A scalar representation \Y{0,0,...,0} is denoted by $\bullet$, a
vector \irrep\ \Y{1,0,...,0} by \YoungA, rank-two symmetric tensor
\irrep\ by \YoungB, rank-two antisymmetric tensor \irrep\ by
\YoungAA\ and so on.

Finite-dimensional irreducible representations of $\sod$ are of the two types: tensor and spin-tensor; and
are also characterized by Young diagrams, which in the case of spin-tensor \irreps\ refer to the symmetry of the tensor part. Young diagrams that correspond to spin-tensor \irreps\ are labeled by $\ferm$-subscript. Spinor indices $\alpha, \beta, \gamma=1...2^{[\frac{d}2]}$ are placed first and are separated from tensor indices by "$\fsp$".  For example, a spinor $\psi^\alpha$ \irrep\ is denoted by $\bullet_\ferm$, a vector-spinor \irrep\ $A^{\alpha\fsp a}$by $\YoungA_\ferm$ and so on. To make tensors irreducible, in
addition to the Young symmetry condition, the tracelessness
condition with respect to each pair of indices is to be imposed
\be \eta_{cc} T^{a(s_1),...,cb(s_i-1),...,cd(s_j-1),...,f(s_n)}\equiv0,\qquad i=1...n,\quad j=1...n,\ee
To make spin-tensors irreducible, in
addition to the Young symmetry condition, $\Gamma$-tracelessness condition with respect to each tensor index and a spinor index is to be imposed
\be {\Gamma_{c}^\alpha}_\beta T^{\beta\fsp a(s_1),...,cb(s_i-1),...,f(s_n)}\equiv0,\qquad i=1...n,\ee
where ${\Gamma_{c}^\alpha}_\beta$ satisfy ${\Gamma_{a}^\alpha}_\beta{\Gamma_{b}^\beta}_\gamma+{\Gamma_{b}^\alpha}_\beta{\Gamma_{a}^\beta}_\gamma=2\eta_{ab}$. Additional conditions on spinors, viz., Majorana, Weyl and Majorana-Weyl are irrelevant to the problems concerned. Also, in both cases it is required for the sum of the heights of the first two columns of Young diagrams to be not greater than $d$.
Note that the $\Gamma$-tracelessness condition is stronger than the tracelessness one and applying the $\Gamma$-tracelessness twice to two symmetric indices gives the tracelessness
\be 0=2{\Gamma_{a}^\alpha}_\beta{\Gamma_{b}^\beta}_\gamma T^{\gamma\fsp (ab)}=\eta_{ab}T^{\alpha\fsp ab}.\ee

To handle with arbitrary large Young diagrams the so-called block
representation is used, i.e., the rows of equal lengths are combined
to blocks
\be\Y{(s_1,p_1),...,(s_n,p_n)}\equiv\Y{\overbrace{s_1,...,s_1}^{p_1},\overbrace{s_2,...,s_2}^{p_2},...,\overbrace{s_n,...,s_n}^{p_n}}.\ee

\providecommand{\href}[2]{#2}\begingroup\raggedright\endgroup

\end{document}